\documentclass[onecolumn]{IEEEtran}
\usepackage{cite}
\usepackage{amsmath,amssymb,amsfonts}
\usepackage{dsfont}
\usepackage{cancel}
\usepackage{textcomp}
\usepackage[inline]{enumitem}
\usepackage[utf8]{inputenc}
\usepackage[english]{babel}
\usepackage{amsthm}
\usepackage{mathtools}
\usepackage{amssymb}
\usepackage[pdftex]{graphicx}
\usepackage{tikz}
\usetikzlibrary{positioning}
\graphicspath{{./pdf/}{./jpeg/}{./png}}
\DeclareGraphicsExtensions{.pdf,.jpeg,.png}

\definecolor{phiGreen}{RGB}{0,128,32}

\def\BibTeX{{\rm B\kern-.05em{\sc i\kern-.025em b}\kern-.08em
    T\kern-.1667em\lower.7ex\hbox{E}\kern-.125emX}}

\newtheorem{theorem}{Theorem}
\newtheorem{corollary}{Corollary}[theorem]
\newtheorem{lemma}{Lemma}
\theoremstyle{definition}

\theoremstyle{remark}

\newcommand{\qbar}{\mkern3mu \raisebox{-0.6ex}[\height][4pt] {\text{-}} \mkern-8mu q }

\title{Subspace Decomposition of Coset Codes}

\author{\IEEEauthorblockN{David Hunn}, \IEEEmembership{Student Member, IEEE}, and \IEEEauthorblockN{Willie K. Harrison}
\IEEEmembership{Senior Member, IEEE}
\thanks{This paper was presented in part at ISIT 2022. }
\thanks{This work was funded by the US National Science Foundation: Grant Award Number \#1910812.}}

\begin{document}

\maketitle

\begin{abstract}
A new method is explored for analyzing the performance of coset codes over the binary erasure wiretap channel (BEWC) by decomposing the code over subspaces of the code space. This technique leads to an improved algorithm for calculating equivocation loss. It also provides a continuous-valued function for equivocation loss, permitting proofs of local optimality for certain finite-blocklength code constructions, including a code formed by excluding from the generator matrix all columns which lie within a particular subspace. Subspace decomposition is also used to explore the properties of an alternative secrecy code metric, the $\chi^2$ divergence. The $\chi^2$ divergence is shown to be far simpler to calculate than equivocation loss. Additionally, the codes which are shown to be locally optimal in terms of equivocation are also proved to be globally optimal in terms of $\chi^2$ divergence.     
\end{abstract}



\section{Introduction}
\label{sec:introduction}
Wiretap coding has been explored since its introduction by Wyner in 1975 as a mechanism for secure and reliable data transmission over communication channels with eavesdroppers~\cite{Wyner1975,Aghdam2019OverviewPhysicalLayer,Bloch2021OverviewInformationTheoreticSecurity,Harrison2013ErrorControlForSecrecy,Bloch2015ErrorControlForSecrecy}. Arguably the simplest nontrivial wiretap channel, the binary erasure wiretap channel (BEWC) was among the first channel models employed to analyze the performance of wiretap codes. Several wiretap coding schemes have been developed which achieve the secrecy capacity of the BEWC in the limit of large blocklength~\cite{Mahdavifar2011PolarCapacityAchieving,Thangaraj2007WiretapLDPC,Subramanian2010StrongAndWeakSecrecy,Bloch2013StrongSecrecyResolvability,Bloch2015ErrorControlForSecrecy,Harrison2013ErrorControlForSecrecy}. On the other hand, a number of emerging applications impose latency, power, and computational limits, which require codes of short blocklength~\cite{Dursi2016ShortPacketRequirements}. The task of finding good secrecy codes with short blocklength has typically been approached using combinatorial techniques~\cite{Harrison2019AttributesOfGenerators}, and finding \emph{best} codes of a given blocklength and dimension remains an open problem even for simple channels such as the BEWC~\cite{Harrison2018DualRelationships, Harrison2018AnalysisShortSecrecyCodes}.

In this work, we present a novel approach which we call \emph{subspace decomposition} for analyzing the properties of linear block codes. We show subspace decomposition to be an especially powerful tool for computing the performance metrics of coset secrecy codes, particularly over the BEWC. 

This paper is an extension of \cite{Hunn2022SubspaceDecompositionExtremeRate}, which gave a preliminary development of subspace decomposition, including prefatory functions and notation used to describe subspace decomposition, as well as a weaker version of the central theorem of this paper. This paper extends these findings in the following ways: \begin{enumerate*}
\item the foundational functions and notation have been generalized to incorporate finite blocklengths and to allow for either fixed erasure probability or fixed number of erased bits; 
\item the analysis has been expanded to accommodate linear block codes whose generator matrices include the all-zero vector; 
\item the central theorem has been strengthened to include a proof of the previously conjectured pattern of coefficients in the central theorem; 
\item subspace decomposition is used to prove the local optimality of certain classes of coset codes, including members of a novel class of codes termed \emph{subspace exclusion codes}; 
\item some remarkable properties of the $\chi^2$ divergence as a coset code metric are explored via subspace decomposition; 
\item results on local optimality in terms of equivocation are extended to global optimality in terms of $\chi^2$ divergence; 
\item the complexity of computing both equivocation and $\chi^2$ divergence via subspace decomposition are analyzed and shown to provide advantage over existing techniques in many cases. \end{enumerate*}

\section{Background}
The wiretap channel used in this paper is the BEWC, shown in Fig.~\ref{fig:wiretap_channel}. In this channel, a message $M$ of size $k$ bits is encoded into a codeword $X$ of size $n$ and sent by the sender. It is then received via a noiseless channel by a legitimate receiver, and a degraded codeword $Z$ is received via a binary erasure channel (BEC) by an eavesdropper. The BEC erasure probability is denoted $\epsilon$. 

\subsection{Coset Coding}
The usual form of secrecy coding for this type of channel is called coset coding and involves encoding the original message $m$ to a codeword $x$ by allowing $m$ to select a coset of a base code $\mathcal{C}$, with $x$ being chosen uniformly at random from the selected coset. In practice, this is done by creating a full-rank $\kappa \times n$ generator matrix $G$ for the linear block code $\mathcal{C}$, where $\kappa = n-k$ is the dimension of the code $\mathcal{C}$. Next, a $k \times n$ auxiliary generator matrix $G'$ is defined which is linearly independent relative to $G$. A random $\kappa$-bit auxiliary message $m'$ is then appended to the message $m$, and the resulting vector is multiplied by the matrix $G^*$, which is formed by the vertical concatenation of $G'$ and $G$. The final codeword $x$ is given by 
\begin{equation}
    x = \begin{bmatrix} m & m' \end{bmatrix} \begin{bmatrix} G' \\ G \end{bmatrix}\mathrm{.}
\end{equation}

\begin{figure}
    \centering
    \scriptsize
    \begin{tikzpicture}
        \node at (-0.1,0) [rectangle, minimum height = 0.4cm, anchor = north, node font=\bfseries, inner sep=0.05cm, outer sep=0] (S) {Sender}; 
        \node at (1.5,0) [rectangle, fill=blue!20, draw=blue, minimum height = 0.6cm, minimum width = 1.1cm] (En) {Encoder}; 
        \node at (3.9,0) [rectangle, align=center, fill=green!20, draw=green, minimum height = 0.6cm, minimum width = 1.1cm] (NL) {Noiseless\\Channel};
        \node at (6.3,0) [rectangle, fill=blue!20, draw=blue, minimum height = 0.6cm, minimum width = 1.1cm] (De) {Decoder}; 
        \node at (7.7,0) [rectangle, minimum height = 0.4cm, minimum width = 1.1cm, align=left, anchor = north, node font=\bfseries] (LR) {Legitimate\\Receiver}; 
        \node at (3.9,-1.3) [rectangle, align=center, fill=green!20, draw=green, minimum height = 0.6cm, minimum width = 1.1cm] (BEC) {Binary Erasure\\Channel (BEC)};
        \node at (6.6,-1.3) [rectangle, minimum height = 0.4cm, minimum width = 1.1cm, align=left, anchor = north, node font=\bfseries] (Eve) {Eavesdropper};
        \draw[->] (En) -- (NL) node[pos=0.5, anchor=south]{\!$X\!$ ($n\!$ bits)}; 
        \draw[->] (NL) -- (De) node[pos=0.5, anchor=south]{\!$X\!$ ($n\!$ bits)};; 
        \draw[->] (S.north) + (-0.15cm,0) -- (En) node[pos=0.5, anchor=south]{\!$M\!$ ($k\!$ bits)};; 
        \draw[->] (De) -- (8.0cm,0) node[pos=0.5, anchor=south]{\!$M\!$ ($k\!$ bits)};
        \draw[->] (En.east)+(0.3,0) |- (BEC.west) ;
        \draw[->] (BEC) -- (6.8cm,-1.3) node[pos=0.95, anchor=south]{$Z\!$ ($n\!$ symbols, $Z_i \!\in\! \{0,\!1,\!?\}$)};
    \end{tikzpicture}

    \caption{Binary erasure wiretap channel.}
    \label{fig:wiretap_channel}
\end{figure}
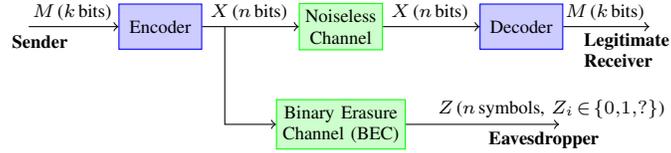

\subsection{Code Performance Characterization}
Leakage of information to an eavesdropper is typically measured using information theoretic metrics on the eavesdropper's codeword $Z$. Many such metrics take the form of a measure of the difference between two probability distributions: the joint distribution $p_{MZ}$ of the message $M$ and received codeword $Z$, and the product $p_M p_Z$ of the marginal distributions of $M$ and $Z$. The most common such metric is the mutual information between the the message and the eavesdropper's codeword, $I(M;Z)$. This metric is referred to as the equivocation loss or leakage, with the eavesdropper message entropy $H(M \!\! \mid \!\! Z) = H(M) - I(M;Z)$ being referred to as the eavesdropper's equivocation. The equivocation loss is equal to the Kullback-Leibler divergence between $p_{MZ}(m,z)$ and $p_M p_Z$. That is, 
\begin{equation}
    \label{eqn:equivocationDefinition}
    I(M;Z) = \mathds{D}(p_{MZ},p_M p_Z)\mathrm{.}
\end{equation}
Other information-theoretic secrecy metrics include norms on the difference between joint and marginal distributions, 
\begin{equation}
    \label{eqn:normDefinition}
    \| p_{MZ},p_M p_Z\|_\beta = \!\! \left(\!\!\!\!\!\!\!\!\!\! \sum_{\;\;\;\;\;\;\;\;\substack{m\in \{0,1\}^{\kappa}, \\ z\in\{0,1,?\}^n}}{\!\!\!\!\!\!\!\!\!\!\left|p_{MZ}(m,z)-p_M(m)p_Z(z)\right|^\beta}\!\!\right)^{\frac{1}{\beta}}\!\!\!\!\mathrm{,}
\end{equation}
with another common metric, the total variation distance, being half the 1-norm, 
\begin{equation}
    \label{eqn:oneNorm}
    \mathds{V}(p_{MZ},p_M p_Z) = \frac{1}{2} \!\!\!\!\!\!\!\!\!\!\!\! \sum_{\;\;\;\;\;\;\;\;\substack{m\in \{0,1\}^{\kappa}, \\ z\in\{0,1,?\}^n}}{\!\!\!\!\!\!\!\!\!\!\!\!\left|p_{MZ}(m,z)-p_M(m)p_Z(z)\right|}\mathrm{.}
\end{equation}


\subsection{Equivocation Properties}
\label{sec:messageEquivocationDefinition}
Equivocation is the most widely-studied information-theoretic measure of secrecy code performance, but calculating equivocation for a specific code is nontrivial, even over a channel as simple as the BEWC. The eavesdropper's equivocation can clearly be calculated as an expectation over all possible values of $z$, 
\begin{equation}
    \label{eqn:EquivocationGeneralExpectation}
    H(M \mid Z) = \sum_{z \in \{0,1,?\}^n}{\mathrm{Pr}(Z=z) \cdot H(M \mid Z=z)}\mathrm{.}
\end{equation}

An important result presented in \cite{Pfister2017QuantifyingEquivocation} is that for uniformly distributed $M$, the equivocation may be calculated based on the set, denoted $r(z)$, of revealed bit positions in the observation $z$, as 
\begin{equation}
    \label{eqn:EquivocationPfister}
    H(M \! \mid \! Z=z) = H(M) - \lvert r(z) \rvert + \mathrm{rank}(G_{r(z)})\mathrm{,}
\end{equation}
where $G_{r(z)}$ is the submatrix of $G$ formed by concatenating the columns of $G$ which are indexed by the set $r(z)$. It is assumed throughout this work that the assumption required by \eqref{eqn:EquivocationPfister} holds--- namely, that the message is uniformly distributed across the $2^k$ possibilities and, therefore, $H(M) = k$.
One of the implications of \eqref{eqn:EquivocationPfister} is that the eavesdropper's equivocation may be calculated via a matrix rank calculation for each of the $2^n$ possible erasure patterns $r(z)$, as 
\begin{equation}
    \label{eqn:EquivocationTotalPfister}
    H(M \! \mid \! Z) = \!\!\!\!\!\!\!\!\sum_{r \in P(\{1 \dots n\})}  \epsilon^{n-\lvert r \rvert} (1-\epsilon)^{\lvert r \rvert} \left(k - \lvert r \rvert + \mathrm{rank}(G_{r})\right)\mathrm{,}
\end{equation}
where $P(\cdot)$ is the power set function. Because $\mathcal{O}(n^2)$ calculations over $\mathds{F}_2^{\kappa}$ are required to calculate the rank of a matrix with a size up to $\kappa \times n$ ~\cite{Bard2009GF2Operations}, the total complexity of an equivocation calculation by this method is $\mathcal{O}(n^2 2^n)$. 
Using \eqref{eqn:EquivocationTotalPfister}, it is practical to calculate a code's equivocation for a given $\epsilon$ up to a blocklength of a few tens. Beyond this blocklength, researchers have generally resorted to approximate equivocation values generated using, e.g., Monte Carlo simulations, as studied in ~\cite{Pfister2017QuantifyingEquivocation} or asymptotic methods, as in ~\cite{Al-Hassan2013BestKnownLinearCodes}. 

\section{Notation and Definitions}
\label{sec:Definitions}
\subsection{Continuous Code Specification}
The usual method of specifying a coset code $\mathcal{C}$ is via the generator matrix $G$. In the BEWC, because erasure probability is independent of bit position, the ordering of columns of G does not affect code performance. For this reason, it is also possible without loss of generality to characterize a coset code $\mathcal{C}$ with generator matrix $G$ of size $n \times \kappa$ by a (zero-indexed) vector $Q$ with integer elements $Q_i, i \in [\![ 0,2^\kappa-1]\!]$. Then each $Q_i$ represents the number of columns of $G$ equal to the binary vector expansion $\nu(i)$ of $i$.

In this work, we use a further generalization of the coset code characterization described above. Instead of the vector $Q$ with integer-valued elements, we use a vector $q$ with continuous-valued elements $q_i, i \in [\![ 0,2^\kappa-1]\!]$, where each $q_i$ represents the fraction of the columns of $G$ equal to $\nu(i)$. In this work, we will generally require $q$ to meet both a nonnegativity constraint, 
\begin{equation}
    \label{eqn:QConstraintPositive}
    q_i \geq 0 \mathrm{,}
\end{equation}
and a unit-sum constraint,
\begin{equation}
    \label{eqn:QConstraintTotal}
    \sum_{i=0}^{2^\kappa-1}q_i = 1 \mathrm{.}
\end{equation}
In general, it is not possible to realize an $(n,\kappa)$ code characterized by $q$ unless $q$ additionally satisfies 
\begin{equation}
    \label{eqn:QConstraintRealizable}
    n q_i \in \mathbb{N} \mathrm{.}
\end{equation}
Nevertheless, it is shown in Section \ref{sec:AnalyticalResults} below that it is possible to define continuous functions acting on $q$ which yield correct code performance characteristics whenever \eqref{eqn:QConstraintRealizable} is satisfied.  

\subsection{Subspace Structure Definitions}
\label{sec:FunctionDefinitions}
Let $W$ be the vector space comprised of all the vectors in $\mathbb{F}_2^{\kappa}$, where binary vector addition and scalar multiplication are defined in the usual way. We then define the function $\Xi(S,d)$ acting on a space $S$ (equal to $W$ or a subspace of $W$) and a scalar dimension $d$. This function returns the set of all dimension-$d$ subspaces of $S$: 
\begin{equation}
    \label{eqn:defXi}
    \Xi(S,d) = \{T: T \subseteq S, \mathrm{dim}(T) = d\} \mathrm{.}
\end{equation}
In this work, braced superscript notation may be used to indicate the dimension of a vector space, so for example $S^{\{d\}}$ indicates a vector space S which has dimension $d$. This notation is for clarity and may be omitted if the dimension of the vector space is clear. The number of subspaces of dimension $d'$ contained within a binary space $S$ of dimension $d$ is given by the Gaussian binomial coefficient
\begin{equation}
    \label{eqn:GaussianBinomial}
    \lvert \Xi(S^{\{d\}},d') \rvert = \binom{d}{d'}_2 = 
    \begin{cases}
        \prod_{\iota=0}^{d' - 1}{\frac{2^{d-\iota} - 1}{2^{d'-\iota} - 1} } &\text{if }d \geq d' \geq 0 \\
        0 &\text{otherwise}
    \end{cases}
    \mathrm{.}    
\end{equation}

Next, several functions are defined which act on a subspace $S$ of $W$. These functions also take a number of other parameters as inputs, including the code-characterization vector $q$, blocklength $n$, erasure probability $\epsilon$, and/or revealed bit count $\mu$. Throughout this work, however, these additional inputs may be considered as implied parameters and omitted in the notation. Additionally, the code dimension $\kappa$ is always considered an implied parameter. First, a function $\zeta(S,q)$ is defined which expresses the fraction of the columns of $G$ that lie within $S$. (We refer to a particular revealed bit being within a space if the column of $G$ indexed by the bit is a member of the space.) This function is defined as 
\begin{equation}
    \label{eqn:defZeta}
    \zeta(S) = \zeta(S,q) = \sum_{i: \nu(i) \in S}{q_i}\mathrm{.}
\end{equation}

Next, we define $\Phi(S,n,\mu,q)$, which is used to represent the probability, given $\mu$ revealed bits, of all of the revealed bits lying within $S$. This probability is given by 
\begin{equation}
    \label{eqn:defPhi}
    \Phi(S) = \Phi(S,n,\mu,q) = \prod_{i=0}^{\mu-1}{\frac{\zeta(S,q)-\frac{i}{n}}{1-\frac{i}{n}}}\mathrm{.}
\end{equation}
Note that if $\zeta(S,q) < {(\mu-1)}/{n}$, although this function may have a negative value in general, it is continuous, and as long as \eqref{eqn:QConstraintRealizable} is satisfied, $\Phi(S)$ will always be positive or zero. 

A related function $\phi(S,n,\epsilon,q)$ is defined
which is used to represent the probability of all revealed bits lying within $S$ for a code with blocklength $n$ and erasure probability $\epsilon$. (The case of zero revealed bits occurring is included in this probability, even if $\zeta(S,q) = 0$.) The value of this probability is given by 
\begin{equation}
    \label{eqn:def_phi}
    \phi(S) = \phi(S,n,\epsilon,q) = \epsilon^{n(1-\zeta(S,q))}\mathrm{.}
\end{equation}



Another set of functions is defined to represent the probability of a set of revealed bits exactly spanning a space $S$. In this work, we use the terminology \emph{exactly spans} to indicate that a set of revealed bits all lie within a space and do not all lie within a proper subspace of that space. Equivalently, we could say that the column vectors indexed by the revealed bits span the space and do not span any higher-dimensional space. Note that from this definition, it is clear that any set of revealed bits spans exactly one subspace of $W$. 

Using this definition, the function $\Psi(S,n,\mu,q)$ is used to represent the probability, given $\mu$ revealed bits, of the set of revealed bits exactly spanning $S$. This probability may be calculated by finding the probability $\Phi(S)$ of all revealed bits lying within $S^{\{d\}}$, then subtracting the probability of $\mu$ revealed bits lying within any of the proper subspaces of $S$. Because a set of revealed bits exactly spans only one space, $\Psi(S)$ may be defined recursively as 
\begin{equation}
    \label{eqn:defPsi}
    \begin{split}
    \Psi(S^{\{d\}}) = \Psi(S^{\{d\}}\!,n,\mu,q) = \Phi(S) - \! \sum_{i = 0 }^{d - 1}{ \left(\!\!\!\!\!\!\!\!\!\!\!\!\!\!\!\!\sum_{\;\;\;\;\;\;\;\;\;\;\;\;T^{\{i\}} \in \Xi(S,i)}\!\!\!\!\!\!\!\!\!\!\!\!\!\! {\Psi(T)}\right)}\mathrm{.} 
    \end{split}
\end{equation}

Finally, the related function $\psi(S,n,\epsilon,q)$ is defined which is used to represent the probability of a set of revealed bits exactly spanning $S$ for a code with blocklength $n$ and erasure probability $\epsilon$. The value of this probability may also be calculated recursively, as 
\begin{equation}
    \label{eqn:def_psi}
    \begin{split}
    \psi(S^{\{d\}}) = \psi(S^{\{d\}},n,\epsilon,q) = \phi(S) - \sum_{i = 0 }^{d - 1}{ \left(\!\!\!\!\!\!\!\!\!\!\!\!\!\!\!\!\sum_{\;\;\;\;\;\;\;\;\;\;\;\;T^{\{i\}} \in \Xi(S,i)}\!\!\!\!\!\!\!\!\!\!\!\!\!\! {\psi(T)}\right)}\mathrm{.}
    \end{split}
\end{equation}

Using these definitions, a number of results may be obtained related to the subspace structure of coset codes. These results are presented below. 

\subsection{Subspace Decomposition Example}
Here we present a simple example of the application of the functions and variables defined in the preceding sections. Consider the code defined by generator matrix
\begin{equation}
    \label{eqn:ex_G}
    G = 
    \begin{bmatrix}
        0 & 1 & 0 & 0 & 1 \\
        0 & 0 & 1 & 1 & 1 \\
        0 & 0 & 0 & 0 & 1 \\
    \end{bmatrix}
    \mathrm{.} 
\end{equation}
Computing the expected equivocation for a given $\epsilon$ as described in Section \ref{sec:messageEquivocationDefinition} requires iterating through each of the possible revealed bit patterns $r(z)$ and finding the rank of the associated submatrix $G_{r(z)}$ and its probability, as shown in Table \ref{tbl:exampleBruteForce}. 

The first step of subspace decomposition is to convert $G$ into a code definition vector $q$, given in this case by 
\begin{equation}
    \label{eqn:ex_q}
    q = \begin{bmatrix}
        0.2 & 0.2 & 0.4 & 0 & 0 & 0 & 0 & 0.2
    \end{bmatrix}^{\intercal} \mathrm{.}
\end{equation}
Next, rather than enumerating all possible erasure patterns, subspace decomposition enumerates all of the subspaces of the global space $W$ and evaluates the properties of the code over each subspace. For the code defined by \eqref{eqn:ex_G}, these subspaces and their properties are illustrated in Table \ref{tab:exampleSubspaceDecomposition}. In the following sections, we demonstrate the utility of information of the type presented in Table \ref{tab:exampleSubspaceDecomposition} in analyzing coset code performance.

\begin{table}
\caption{Erasure pattern enumeration for coset code of \eqref{eqn:ex_G}.}
\centering
\begin{tabular}{|c |c |c ||c |c |c | }
    \hline
    $r(z)$ & $\! \mathrm{rank}(G_{\!r(z)}) \!$ & $\mathrm{Pr}[r(z)]$ & $r(z)$ & $\! \mathrm{rank}(G_{\!r(z)}) \!$ & $\mathrm{Pr}[r(z)]$
    \\
    \hline
    00000 & 0 & 0.00032 & 11100 & 2 & 0.02048 \\
    10000 & 0 & 0.00128 & 11010 & 2 & 0.02048 \\
    01000 & 1 & 0.00128 & 11001 & 2 & 0.02048 \\
    00100 & 1 & 0.00128 & 10110 & 1 & 0.02048 \\
    00010 & 1 & 0.00128 & 10101 & 2 & 0.02048 \\
    00001 & 1 & 0.00128 & 10011 & 2 & 0.02048 \\
    11000 & 1 & 0.00512 & 01110 & 2 & 0.02048 \\
    10100 & 1 & 0.00512 & 01101 & 3 & 0.02048 \\
    10010 & 1 & 0.00512 & 01011 & 3 & 0.02048 \\
    10001 & 1 & 0.00512 & 00111 & 2 & 0.02048 \\
    01100 & 2 & 0.00512 & 11110 & 2 & 0.08192\\
    01010 & 2 & 0.00512 & 11101 & 3 & 0.08192\\
    01001 & 2 & 0.00512 & 11011 & 3 & 0.08192\\
    00110 & 1 & 0.00512 & 10111 & 2 & 0.08192\\
    00101 & 2 & 0.00512 & 01111 & 3 & 0.08192\\
    00011 & 2 & 0.00512 & 11111 & 3 & 0.32768\\
     \hline    
\end{tabular}
\label{tbl:exampleBruteForce}
\end{table}

\begin{table}
\caption{Subspace decomposition of coset code defined by \eqref{eqn:ex_q}.}
\centering
\begin{tabular}{|c|c|c|c|c|c|}
    \hline
    \rule[0.28cm]{0pt}{0pt} Subspace \rule[-0.12cm]{0pt}{0pt} & $\zeta(S)$ & $\Phi(S)$ & $\Psi(S)$ & $\phi(S)$ & $\psi(S)$ \\
    \hline
    \multicolumn{6}{|c|}{\rule[0.28cm]{0pt}{0pt} Dimension-0 subspaces: $S \in \Xi(W,0)$ \rule[-0.10cm]{0pt}{0pt}} \\
    \hline
    $\rule[0.27cm]{0pt}{0pt}\{\nu(0)\}$ & 0.2 & 0 & 0 & 0.0016 & 0.0016\\
    \hline
    \multicolumn{6}{|c|}{\rule[0.28cm]{0pt}{0pt} Dimension-1 subspaces: $S \in \Xi(W,1)$ \rule[-0.10cm]{0pt}{0pt}} \\
    \hline
    \rule[0.27cm]{0pt}{0pt}$\{\nu(0),\nu(1)\}$ & 0.4 & 0.1 & 0.1 & 0.008 & 0.0064\\
    \hline
    \rule[0.27cm]{0pt}{0pt} $\{\nu(0),\nu(2)\}$ & 0.6 & 0.3 & 0.3 & 0.04 & 0.0384\\
    \hline
    \rule[0.27cm]{0pt}{0pt} $\{\nu(0),\nu(3)\}$ & 0.2 & 0 & 0 & 0.0016 & 0\\
    \hline
    \rule[0.27cm]{0pt}{0pt} $\{\nu(0),\nu(4)\}$ & 0.2 & 0 & 0 & 0.0016 & 0\\
    \hline
    \rule[0.27cm]{0pt}{0pt} $\{\nu(0),\nu(5)\}$ & 0.2 & 0 & 0 & 0.0016 & 0\\
    \hline
    \rule[0.27cm]{0pt}{0pt} $\{\nu(0),\nu(6)\}$ & 0.2 & 0 & 0 & 0.0016 & 0\\
    \hline
    \rule[0.27cm]{0pt}{0pt} $\{\nu(0),\nu(7)\}$ & 0.4 & 0.1 & 0.1 & 0.008 & 0.0064\\
    \hline
    \multicolumn{6}{|c|}{\rule[0.28cm]{0pt}{0pt} Dimension-2 subspaces: $S \in \Xi(W,2)$ \rule[-0.10cm]{0pt}{0pt}} \\
    \hline
    \rule[0.27cm]{0pt}{0pt} $\{\nu(0),\nu(1),\nu(2),\nu(3)\}$ & 0.8 & 0.6 & 0.2 & 0.2 & 0.1536\\
    \hline
    \rule[0.27cm]{0pt}{0pt} $\{\nu(0),\nu(1),\nu(4),\nu(5)\}$ & 0.4 & 0.1 & 0 & 0.008 & 0\\
    \hline
    \rule[0.27cm]{0pt}{0pt} $\{\nu(0),\nu(1),\nu(6),\nu(7)\}$ & 0.6 & 0.3 & 0.1 & 0.04 & 0.0256\\
    \hline
    \rule[0.27cm]{0pt}{0pt} $\{\nu(0),\nu(2),\nu(4),\nu(6)\}$ & 0.6 & 0.3 & 0 & 0.04 & 0\\
    \hline
    \rule[0.27cm]{0pt}{0pt} $\{\nu(0),\nu(2),\nu(5),\nu(7)\}$ & 0.8 & 0.6 & 0.2 & 0.2 & 0.1536\\
    \hline
    \rule[0.27cm]{0pt}{0pt} $\{\nu(0),\nu(3),\nu(4),\nu(7)\}$ & 0.4 & 0.1 & 0 & 0.008 & 0\\
    \hline
    \rule[0.27cm]{0pt}{0pt} $\{\nu(0),\nu(3),\nu(5),\nu(6)\}$ & 0.2 & 0 & 0 & 0.0016 & 0\\
    \hline
    \multicolumn{6}{|c|}{\rule[0.28cm]{0pt}{0pt} Dimension-3 subspaces: $S \in \Xi(W,3)$ \rule[-0.10cm]{0pt}{0pt}} \\
    \hline
    \rule[0.27cm]{0pt}{0pt} $W$ \rule[-0.05cm]{0pt}{0pt} & 1 & 1 & 0 & 1 & 0.6144 \\
    \hline
\end{tabular}
\label{tab:exampleSubspaceDecomposition}
\end{table}

\section{Analytical Results}
\label{sec:AnalyticalResults}
With the core concepts of subspace decomposition in place, we now show how this technique may be applied to derive useful results in the area of coset coding theory. This section first presents several lemmas necessary for the development of the main result and other theorems. Next, we present one of the principal results, a novel expression for the expected equivocation of coset code, followed by a complexity analysis and several secondary results arising from this expression. An alternative secrecy metric, the $\chi^2$ divergence, is also explored using subspace decomposition and found to exhibit remarkable properties. The secondary results and complexity analysis are repeated for the $\chi^2$ divergence, for which they provide stronger conclusions.

Proofs of the principal results are provided in this section. Supporting lemmas and secondary results are proved in the appendices, as the proofs are somewhat involved and provide little insight into the mechanics or applications of subspace decomposition.

\subsection{Supporting Lemmas}
\begin{lemma}
\label{PhiNonnegativityLemma}
For any subspace $S$ of $W^{\{\kappa\}}$, any positive integer $n$, any $0 \leq \epsilon < 1$, and any $\kappa$-dimensional code definition vector $q$ satisfying \eqref{eqn:QConstraintPositive} and \eqref{eqn:QConstraintTotal}, 
\begin{equation}
    \label{eqn:psiNonnegativityStatement}
    \psi(S,n,\epsilon,q) \geq 0 \mathrm{.} 
\end{equation} 
\end{lemma}
\begin{proof}
    See Appendix \ref{Appendix:phiNonnegativityProof}. 
\end{proof}

\begin{lemma}
\label{EtaPrimeSumLemma}
The function $\eta'(a,b)$ defined by 
\begin{equation}
    \label{eqn:defEtaPrime}
    \eta'(a,b) = \binom{a}{b}_2 (-1)^{a-b} 2^{(a-b)(a-b-1)/2} 
\end{equation}
has the following properties: 
\begin{equation}
    \label{eqn:etaPrimeSum}
    \sum_{i=0}^{a}{\eta'(a,i)} = 
    \begin{cases}
        1 & \text{if } a = 0 \\
        0 & \text{otherwise} 
    \end{cases}
    \mathrm{,}
\end{equation}
\begin{equation}
    \label{eqn:etaPrimePartialSum}
    \sum_{i=b>0}^{a}{\eta'(a,i)} = 2^{a-b}\eta'(a-1,b-1) \mathrm{,}
\end{equation}
\begin{equation}
    \label{eqn:etaPrimeWeightedSum}
    \sum_{i=b}^{a}{\binom{a}{i}_2 \eta'(i,b)} = 
    \begin{cases}
        1 & \text{if } a = b \\
        0 & \text{otherwise}          
    \end{cases}
    \mathrm{,}
\end{equation}
and 
\begin{equation}
    \label{eqn:KConstantsLemmaStatement}
    \sum_{j=1}^{a}{\left( j \cdot \eta'(a,j) \right)} = \prod_{i=1}^{a-1}{(1-2^i)} \mathrm{.}
\end{equation}
\end{lemma}
\begin{proof}
    See Appendix \ref{Appendix:etaPrimeProof}. 
\end{proof}

\begin{lemma}
\label{thm:exponentialMersenneLemma}
For integers $b > 0$ and $n > 0$, 
\begin{equation}
    \label{eqn:MersenneExponentialProperty}
    \sum_{i=0}^{n}{\left( 2^{bi} \prod_{j=i}^{n}{(1-2^j)} \right)} < 0 \mathrm{.}
\end{equation}
\end{lemma}
\begin{proof}
See Appendix \ref{apx:exponentialMersenneLemmaProof}. \end{proof}

\begin{lemma}
\label{SuperexponentialLemma}
For constant $\beta \geq 1$ and integer $n > 0$, 
\begin{equation}
    \label{eqn:SuperexponentialProperty}
    \sum_{i=0}^{n}{\left( 2^i \beta^{2^i} \prod_{j=i}^{n}{(1-2^j)} \right)} < 0 \mathrm{.}
\end{equation}
\end{lemma}
\begin{proof}
See Appendix \ref{Appendix:SuperexponentialProof}
\end{proof}

\subsection{Equivocation Loss Via Subspace Decomposition}
The functions defined in Section \ref{sec:FunctionDefinitions} have direct application to the problem of calculating a coset code's expected equivocation loss. When equivocation is computed in this way, a remarkable expression, given in Theorem \ref{thm:ExpectedEquivocationFormula}, may be derived which gives the equivocation loss only in terms of constant integer multiples of $\phi(\cdot)$ or $\Phi(\cdot)$ for each of the subspaces of $W$. In addition to providing what is in many cases a more efficient means for calculating equivocation, this expression yields theoretical insight into the properties of good coset codes and facilitates proofs of a number of analytical results. 

\begin{theorem}
\label{thm:ExpectedEquivocationFormula}
The expected equivocation loss $I(M;Z)$ for a given number $\mu > 0$ of revealed codeword bits, denoted $L(n,\mu,q)$, is equal to 
\begin{equation}
    \label{eqn:expectedEquivocationMu}
    L(n,\mu,q) = \mu - \kappa + \!\! \sum_{\delta = 1 }^{\kappa}{\left( K_{\delta} \!\!\! \sum_{S \in \Xi(W,\kappa-\delta)} \!\!\!\!\!\!  \Phi(S)\right)}\mathrm{,} 
\end{equation}
and the expected equivocation loss for a given erasure probability $\epsilon$, denoted $l(n,\epsilon,q)$, is equal to 
\begin{equation}
    \label{eqn:expectedEquivocationEpsilon}
    l(n,\epsilon,q) = n (1-\epsilon) - \kappa + \!\! \sum_{\delta = 1 }^{\kappa}{\left( K_{\delta} \!\!\! \sum_{S \in \Xi(W,\kappa-\delta)} \!\!\!\!\!\!  \phi(S)\right)}\mathrm{,} 
\end{equation}
with $K_\delta$ being a series of constants given by 
\begin{equation}
    \label{eqn:expectedEquivocationConstantsPattern}
    K_\delta = \prod_{i=1}^{\delta-1} {(1-2^{i})}\mathrm{.}
\end{equation}

\end{theorem}

\begin{proof}
Begin by observing that the equivocation loss for a given erasure pattern $r(z)$ may be derived from \eqref{eqn:EquivocationPfister} as 
\begin{equation}
    \label{eqn:equivocationLossMuPfister}
    \begin{split}
        L(n,\mu,q) &= I(M;Z) = H(M) - H(M \lvert Z = z)
        \\
        &\!\!\!\!\!\!\!\!\!\!\!\!\!\!\!\!\!\! = \mathbb{E}\left[\lvert r(z) \rvert \!-\! \mathrm{rank}(G_{r(z)})\right] = \mu \!-\! \mathbb{E}\left[ \mathrm{rank}(G_{r(z)}) \right]\mathrm{,}
    \end{split}
\end{equation}
or for a given $\epsilon$ as 
\begin{equation}
    \label{eqn:equivocationLossEpsilonPfister}
    \begin{split}
        l(n,\epsilon,q) &=  \mathbb{E}\left[\lvert r(z) \rvert - \mathrm{rank}(G_{r(z)})\right]
        \\
        &= n(1-\epsilon) - \mathbb{E}\left[ \mathrm{rank}(G_{r(z)}) \right]\mathrm{,}
    \end{split}
\end{equation}
Note that the rank of $G_{r(z)}$ is equal to the dimension of the subspace of $W$ exactly spanned by $r(z)$. Because the possibility of $r(z)$ exactly spanning a subspace $S$ is disjoint for each $S$, the expected rank of $G_{r(z)}$ may be calculated by accumulating the dimension of each $S \subseteq W$ times the probability $\Psi(S)$ or $\psi(S)$ of that subspace being exactly spanned by the revealed bits. Symbolically, 
\begin{equation}
    \label{eqn:equivocationLossMu1}
    \begin{split}
        \mathbb{E}\left[ \mathrm{rank}(G_{r(z)}) \mid \mu  \right] = \sum_{d=0}^{\kappa}{\left( \!\!\!\!\!\!\!\!\sum_{\;\;\;\;\;\;S^{\{d\}} \in \Xi(W,d)}{\!\!\!\!\!\!\!\!d \cdot \Psi(S)} \right)} \mathrm{,}
    \end{split}
\end{equation}
or 
\begin{equation}
    \label{eqn:equivocationLossEpsilon1}
    \begin{split}
        \mathbb{E}\left[ \mathrm{rank}(G_{r(z)}) \mid \epsilon \right] = \sum_{d=0}^{\kappa}{\left( \!\!\!\!\!\!\!\!\sum_{\;\;\;\;\;\;S^{\{d\}} \in \Xi(W,d)}{\!\!\!\!\!\!\!\!d \cdot \psi(S)} \right)} \mathrm{.}
    \end{split}
\end{equation}

Next, consider the expansion of the recursive formulas \eqref{eqn:defPsi} and \eqref{eqn:def_psi}. In this analysis, we will refer specifically to the case of \eqref{eqn:defPsi} which involves $\Psi(T)$ and $\Phi(T)$, but an identical analysis holds for \eqref{eqn:def_psi} using $\psi(T)$ and $\phi(T)$. This expansion begins with dimension $d$ and terminates with dimension zero, in which no further calls to $\Psi(T)$ are made. With each instance of a call to $\Psi(T)$, the expanded expression accumulates another $\Phi(T)$ term, so that the final expression consists entirely of $\Phi(T)$ terms for the various subspaces $T$ of $S$, some of which are repeated multiple times, and which may have coefficients 1 or -1. Because of the symmetry of the subspaces of $S$, all subspaces $T^{\{d'\}}$ of a particular dimension $d'$ appear the same number of times in the expanded expression. Thus, the final expansion may be expressed as 
\begin{equation}
    \label{eqn:PsiExpanded1}
    \begin{split}
    \Psi(S^{\{d\}}) = \sum_{d' = 0 }^{d}{ \left(\!\!\!\!\!\!\!\!\!\!\!\!\!\!\!\!\sum_{\;\;\;\;\;\;\;\;\;\;\;\;T^{\{d'\}} \in \Xi(S,d')} \!\!\!\!\!\!\!\!\!\!\!\!\!\! {c(d,d') \cdot \Phi(T^{\{d'\}})}\right)} \mathrm{,} 
    \end{split}
\end{equation}
where $c(d,d')$ defines a series of constants used for summing subspaces of a $d$-dimensional space. 

Calculating the $c(d,d')$ of \eqref{eqn:PsiExpanded1} involves tracing the expansion of \eqref{eqn:defPsi} to determine how many times $\Psi(T^{\{d'\}})$ is invoked for a given $T^{\{d'\}} \subseteq S^{\{d\}}$ with sign 1 and with sign -1. Equivalently, we may determine how many times $\Psi(T^{\{d'\}})$ is invoked for any $d'$-dimensional $T^{\{d'\}} \subset S^{\{d\}}$, then divide by the number $\lvert \Xi(S,d') \rvert = \binom{d}{d'}_2$ of such $d'$-dimensional subspaces. To find how many times $\Psi(T^{\{d'\}})$ is invoked for any $T^{\{d'\}} \subset S^{\{d\}}$, first observe that \eqref{eqn:defPsi} will directly (that is, without any recursive expansion) invoke $\Psi(T^{\{d'\}})$ once for each $T^{\{d'\}} \subset S^{\{d\}}$. Then, if $d > d' + 1$, $\Psi(T^{\{d'\}})$ will also be invoked indirectly at least once via a call to $\Psi(U^{\{d''\}})$ with $d > d'' > d'$ for some subspace $U^{\{d''\}} \subset S^{\{d\}}$ of which  $T^{\{d'\}}$ is a subspace. The concepts of direct and indirect invocations, along with the remainder of the calculation of $c(d,d')$, are illustrated in Fig. \ref{fig:psiFunctionRecursion}. 

\begin{figure*}
    \centering
    \begin{tikzpicture}
        [phiNode/.style={text=phiGreen, font=\fontsize{10}{10}\selectfont}, psiNode/.style={text=blue, anchor=north west, inner sep=0, font=\fontsize{10}{10}\selectfont}, pathNode/.style={anchor=base,text=orange,font=\fontsize{8}{8}\selectfont}, multNode/.style={anchor=base,text=red, font=\fontsize{8}{8}\selectfont}, subNode/.style={anchor=northeast,text=brown,font=\fontsize{7}{7}\selectfont}]
        
        \node[psiNode] (parent) at (0,0) {$\Psi(S^{\{4\}})$};
        \node[phiNode] (parentPhi) [below right=0.1cm and -0.7cm of parent] {$\Phi(S^{\{4\}})$};
        \draw [->] (parent.south)+(-0.4cm,0) [rounded corners, phiGreen] |- (parentPhi.west);
        \node[subNode] (parentSub) [below right = -0.15cm and 0.0cm of parentPhi.south east,anchor=north east] {$\theta=\{4\}\;\;\,$(1$\times$)};
        
        \node[psiNode] (2_1) at (5.6cm,-2cm) {$\Psi(S^{\{3\}})$};
        \node[phiNode] (2_1_Phi) [below right=0.1cm and -0.7cm of 2_1] {$\Phi(S^{\{3\}})$};
        \draw [->] (2_1.south)+(-0.4cm,0) [rounded corners, phiGreen] |- (2_1_Phi.west);
        \node[subNode] (2_1_Sub) [below right = -0.15cm and 0.5cm of 2_1_Phi.south east,anchor=north east] {$\theta=\{4,3\}\;\;\,$(-15$\times$)};

        \draw [->] (parent.south)+(-0.55cm,0) [blue] |- (2_1.west) node[pos=0.8, blue, font=\fontsize{8}{8}\selectfont, anchor=south]{(-)15};

        \node[psiNode] (3_1) at (2.8cm,-4cm) {$\Psi(S^{\{2\}})$};
        \node[phiNode] (3_1_Phi) [below right=0.1cm and -0.7cm of 3_1] {$\Phi(S^{\{2\}})$};
        \draw [->] (3_1.south)+(-0.4cm,0) [rounded corners, phiGreen] |- (3_1_Phi.west);
        \node[subNode] (3_1_Sub) [below right = -0.15cm and 0.5cm of 3_1_Phi.south east,anchor=north east] {$\theta=\{4,2\}\;\;\,$(-35$\times$)};

        \draw [->] (parent.south)+(-0.55cm,0) [blue] |- (3_1.west) node[pos=0.8, blue, font=\fontsize{8}{8}\selectfont, anchor=south]{(-)35};

        \node[psiNode] (3_2) at (8.4cm,-4cm) {$\Psi(S^{\{2\}})$};
        \node[phiNode] (3_2_Phi) [below right=0.1cm and -0.7cm of 3_2] {$\Phi(S^{\{2\}})$};
        \draw [->] (3_2.south)+(-0.4cm,0) [rounded corners, phiGreen] |- (3_2_Phi.west);
        \node[subNode] (3_2_Sub) [below right = -0.15cm and 0.8cm of 3_2_Phi.south east,anchor=north east] {$\theta=\{4,3,2\}\;\;\,$(105$\times$)};

        \draw [->] (2_1.south)+(-0.55cm,0) [blue] |- (3_2.west) node[pos=0.8, blue, font=\fontsize{8}{8}\selectfont, anchor=south]{(-)7};

        \node[psiNode] (4_1) at (1.1cm,-6cm) {$\Psi(S^{\{1\}})$};
        \node[phiNode] (4_1_Phi) [below right=0.1cm and -0.7cm of 4_1] {$\Phi(S^{\{1\}})$};
        \draw [->] (4_1.south)+(-0.4cm,0) [rounded corners, phiGreen] |- (4_1_Phi.west);
        \node[subNode] (4_1_Sub) [below right = -0.15cm and 0.3cm of 4_1_Phi.south east,anchor=north east] {$\theta=\{4,1\}\;\;\,$(-15$\times$)};

        \draw [->] (parent.south)+(-0.55cm,0) [blue] |- (4_1.west) node[pos=0.8, blue, font=\fontsize{8}{8}\selectfont, anchor=south]{(-)15};

        \node[psiNode] (4_2) at (3.9cm,-6cm) {$\Psi(S^{\{1\}})$};
        \node[phiNode] (4_2_Phi) [below right=0.1cm and -0.7cm of 4_2] {$\Phi(S^{\{1\}})$};
        \draw [->] (4_2.south)+(-0.4cm,0) [rounded corners, phiGreen] |- (4_2_Phi.west);
        \node[subNode] (4_2_Sub) [below right = -0.15cm and 0.6cm of 4_2_Phi.south east,anchor=north east] {$\theta=\{4,2,1\}\;\;\,$(105$\times$)};

        \draw [->] (3_1.south)+(-0.55cm,0) [blue] |- (4_2.west) node[pos=0.8, blue, font=\fontsize{8}{8}\selectfont, anchor=south]{(-)3};

        \node[psiNode] (4_3) at (6.7cm,-6cm) {$\Psi(S^{\{1\}})$};
        \node[phiNode] (4_3_Phi) [below right=0.1cm and -0.7cm of 4_3] {$\Phi(S^{\{1\}})$};
        \draw [->] (4_3.south)+(-0.4cm,0) [rounded corners, phiGreen] |- (4_3_Phi.west);
        \node[subNode] (4_3_Sub) [below right = -0.15cm and 0.6cm of 4_3_Phi.south east,anchor=north east] {$\theta=\{4,3,1\}\;\;\,$(105$\times$)};

        \draw [->] (2_1.south)+(-0.55cm,0) [blue] |- (4_3.west) node[pos=0.8, blue, font=\fontsize{8}{8}\selectfont, anchor=south]{(-)7};

        \node[psiNode] (4_4) at (9.5cm,-6cm) {$\Psi(S^{\{1\}})$};
        \node[phiNode] (4_4_Phi) [below right=0.1cm and -0.7cm of 4_4] {$\Phi(S^{\{1\}})$};
        \draw [->] (4_4.south)+(-0.4cm,0) [rounded corners, phiGreen] |- (4_4_Phi.west);
        \node[subNode] (4_4_Sub) [below right = -0.15cm and 0.9cm of 4_4_Phi.south east,anchor=north east] {$\theta=\{4,3,2,1\}\;\;\,$(-315$\times$)};

        \draw [->] (3_2.south)+(-0.55cm,0) [blue] |- (4_4.west) node[pos=0.8, blue, font=\fontsize{8}{8}\selectfont, anchor=south]{(-)3};

        \draw (-0.2,1.2) to (-0.2,-7.5);
        \draw (-1.6,0.5) to (-1.6,-7.5);
        \draw (12.5,1.2) to (12.5,-7.5);
        \draw (13.9,1.2) to (13.9,-7.5);
        \draw (15.3,1.2) to (15.3,-7.5);

        \draw (-0.2,1.2) to (15.3,1.2);
        \draw (-1.6,0.5) to (15.3,0.5);
        \draw (-1.6,-1.5) to (15.3,-1.5);
        \draw (-1.6,-3.5) to (15.3,-3.5);
        \draw (-1.6,-5.5) to (15.3,-5.5);
        \draw (-1.6,-7.5) to (15.3,-7.5);

        \draw [dashed] (-1.6,-0.5) to (12.5,-0.5);
        \draw [dashed] (-1.6,-2.5) to (12.5,-2.5);
        \draw [dashed] (-1.6,-4.5) to (12.5,-4.5);
        \draw [dashed] (-1.6,-6.5) to (12.5,-6.5);

        \node at (6.2,0.7) [anchor=base] {Function Evaluation Path};

        \node at(-0.9,0.1) [anchor=base] {$d'\!\!=\!4$};
        \node at(-0.9,-0.3) [anchor=base,color=blue, font=\fontsize{10}{10}\selectfont] {$\Psi(\cdot)$};
        \node at(-0.9,-0.9) [anchor=base,color=phiGreen, font=\fontsize{10}{10}\selectfont] {$\Phi(\cdot)$};
        \node at(-0.9,-1.3) [anchor=base,color=brown, font=\fontsize{7}{7}\selectfont] {path  (mult.)};

        \node at(-0.9,-1.9) [anchor=base] {$d'\!\!=\!3$};
        \node at(-0.9,-2.3) [anchor=base,color=blue, font=\fontsize{10}{10}\selectfont] {$\Psi(\cdot)$};
        \node at(-0.9,-2.9) [anchor=base,color=phiGreen, font=\fontsize{10}{10}\selectfont] {$\Phi(\cdot)$};
        \node at(-0.9,-3.3) [anchor=base,color=brown, font=\fontsize{7}{7}\selectfont] {path  (mult.)};

        \node at(-0.9,-3.9) [anchor=base] {$d'\!\!=\!2$};
        \node at(-0.9,-4.3) [anchor=base,color=blue, font=\fontsize{10}{10}\selectfont] {$\Psi(\cdot)$};
        \node at(-0.9,-4.9) [anchor=base,color=phiGreen, font=\fontsize{10}{10}\selectfont] {$\Phi(\cdot)$};
        \node at(-0.9,-5.3) [anchor=base,color=brown, font=\fontsize{7}{7}\selectfont] {path  (mult.)};

        \node at(-0.9,-5.9) [anchor=base] {$d'\!\!=\!1$};
        \node at(-0.9,-6.3) [anchor=base,color=blue, font=\fontsize{10}{10}\selectfont] {$\Psi(\cdot)$};
        \node at(-0.9,-6.9) [anchor=base,color=phiGreen, font=\fontsize{10}{10}\selectfont] {$\Phi(\cdot)$};
        \node at(-0.9,-7.3) [anchor=base,color=brown, font=\fontsize{7}{7}\selectfont] {path  (mult.)};

        \node at(13.2,0.7) [anchor=base] {$\eta(4,d')$};
        \node at(14.6,0.7) [anchor=base] {$c(4,d')$};

        \node at(13.2,-0.5) [anchor=mid] {1};
        \node at(14.6,-0.5) [anchor=mid] {1};
        \node at(13.2,-2.5) [anchor=mid] {-15};
        \node at(14.6,-2.5) [anchor=mid] {-1};
        \node at(13.2,-4.5) [anchor=mid] {70};
        \node at(14.6,-4.5) [anchor=mid] {2};
        \node at(13.2,-6.5) [anchor=mid] {-120};
        \node at(14.6,-6.5) [anchor=mid] {-8};


    \end{tikzpicture}
    \caption{Calculation of $c(d,d')$ for $d=4$ and $d' \in [\![ 0,4]\!]$ by computing the function $\eta(d,d')$ via expansion of $\Psi(S)$ through the recursion paths $\theta \in \Theta_{d,d'}$.}
    \label{fig:psiFunctionRecursion}
\end{figure*}

For each instance of a call to $\Phi(T^{\{d'\}})$ for a $d'$-dimensional subspace $T^{\{d'\}}$ in the final expansion of \eqref{eqn:defPsi}, a ``path'' $\theta$ may be specified to indicate how that instance was obtained in the expansion of \eqref{eqn:defPsi}. The path $\theta$ is represented as an ordered (decreasing) set of integers $\theta_i \in [\![d' , d]\!]$, $\theta_{i+1} < \theta_i$, which indicates the dimension $\theta_i$ of each call to $\Psi(T^{\{\theta_i\}})$, starting with $\theta_1 = d$ and ending with $\theta_{|\theta|} = d'$. Next, let $\Theta_{d,d'}$ represent the set of all valid paths which start with $d$ and end with $d'$. (Note that if $d=d'$, there is exactly one path $\theta$, of size one, in $\Theta_{d,d'}$. Also note that if $d=d'+1$, there is exactly one path $\theta$, of size two, in $\Theta_{d,d'}$.) For each path $\theta \in \Theta_{d,d'}$, the number of times $\Phi(T^{\{d'\}})$ is invoked for any subspace $T^{\{d'\}} \in \Xi(S^{\{d\}},d')$ is equal to the product of $\binom{\theta_{i-1}}{\theta_i}_2$ for each $1 < i \leq |\theta|$. Because the sign of summation alternates with each successive recursive call to $\Psi(U)$, the sign of any given $\Phi(T^{\{d'\}})$ encountered via path $\theta$ during the expansion of \eqref{eqn:defPsi} is positive if $|\theta|$ is even and negative if $|\theta|$ is odd. It is useful to represent the product over such a path set by a function $\eta(d,d')$ defined as 
\begin{equation}
    \label{eqn:defEta}
    \eta(d,d') = \sum_{\theta \in \Theta_{d,d'}}(-1)^{|\theta|+1} \cdot \left(\prod_{j=2}^{|\theta|}{\binom{\theta_{j-1}}{\theta_j}_2}\right)    \mathrm{.} 
\end{equation}
This function $\eta(d,d')$ has some interesting properties the follow from its definition and which are useful to our analysis. First note that if $d \neq d'$, then there are multiple elements in each path in $\Theta(d,d')$, and therefore $\eta(d,d')$ may be expressed as a summation over the possible values of the second path element as  
\begin{equation}
    \label{eqn:def_eta}
    \begin{split}
    \eta(d,d') = 
    \begin{cases}
        1 & \text{if } d = d' \\
         -\sum\limits_{i=d'}^{d-1} {\binom{d}{i}_2 \eta(i,d')} \mathrm{,} & \text{otherwise }     \mathrm{.} 
    \end{cases}
    \end{split}
\end{equation}
Indeed, this property is sufficient to fully define $\eta(d,d')$, as for any $d > d'$, $\eta(d,d')$ may be determined from the values of $\eta(d-1,d'), \eta(d-2,d'), \dots ,\eta(d',d')$. This property may also be restated by subtracting $\eta(d,d')$ as 
\begin{equation}
    \label{eqn:etaSplitProperty}
    \begin{split}
    \sum\limits_{i=d'}^{d}{\binom{d}{i}_2 \eta(i,d')} = 
    \begin{cases}
        1 & \text{if } d = d' \\
        0 & \text{otherwise} 
        \mathrm{.} 
    \end{cases}
    \end{split}
\end{equation}
This property is identical to the property \eqref{eqn:etaPrimeWeightedSum} of the $\eta'(d,d')$ function defined in \eqref{eqn:defEtaPrime} of Lemma \ref{EtaPrimeSumLemma}. Because these functions have the same fully-defining property, the functions must be identical.  

Then, because $\eta(d,d')$ represents the number of times $\Psi(T^{\{d'\}})$ is invoked for any $d'$-dimensional subspace $T^{\{d'\}} \subset S^{\{d\}}$, the $c(d,d')$ of \eqref{eqn:PsiExpanded1} are given by 
\begin{equation}
    \label{eqn:def_c_constants}
    c(d,d') = \eta(d,d')/\binom{d}{d'}_2 = (-1)^{d-d'}2^{(d-d')(d-d'-1)/2} \mathrm{.}
\end{equation}

Next, we return to the expression for the expected rank of $G_{r(z)}$. Combining \eqref{eqn:equivocationLossMu1} with \eqref{eqn:PsiExpanded1}, it is clear that the expected rank of $G_{r(z)}$ will also be expressible as a sum of $\Phi(T)$ terms for every subspace $T$ of $W$, with the coefficient of each term depending only on the dimension of $T$. Then (exchanging the subspace variable $T$ in favor of $S$, as the subspace previously represented by $S$ is no longer used) we have  
\begin{equation}
    \label{eqn:equivocationLossMu2}
    \begin{split}
        \mathbb{E}\left[ \mathrm{rank}(G_{r(z)}) \mid \mu  \right] = \sum_{d=0}^{\kappa}{\left( \!\!\!\!\!\!\!\!\sum_{\;\;\;\;\;\;S^{\{d\}} \in \Xi(W,d)}{\!\!\!\!\!\!\!\!C(\kappa,d) \cdot \Phi(S)} \right)} \mathrm{,}
    \end{split}
\end{equation}
with $C(\kappa,d)$ being a series of constants used for summing subspaces of a $\kappa$-dimensional coset code. 

To find the values of the $C(\kappa,d)$, consider the summations in \eqref{eqn:equivocationLossMu1} and \eqref{eqn:PsiExpanded1} as applied to a given subspace $T^{\{d'\}}$. For a given $d \geq d'$ in \eqref{eqn:equivocationLossMu1}, a term of $c(d,d') \cdot \Phi(T^{\{d'\}})$ will be accumulated for every $d$-dimensional space $S^{\{d\}}$ of which $T^{\{d'\}}$ is a subspace. The number of $d$-dimensional subspaces $S^{\{d\}}$ of $W^{\{\kappa\}}$ which are superspaces of a $d'$-dimensional space $T^{\{d'\}}$ is given by 
\begin{equation}
    \label{eqn:numberOfSuperspaces}
    \lvert \{S: S^{\{d\}} \!\! \subseteq \! W^{\{\kappa\}} , T^{\{d'\}} \!\! \subseteq \! S^{\{d\}}\} \rvert = \frac{\binom{\kappa}{d}_2 \binom{d}{d'}_2}{\binom{\kappa}{d'}_2} = \!\binom{\kappa\!-\!d'}{d\!-\!d'}_2 \!\mathrm{.}
\end{equation}
Using this value with \eqref{eqn:PsiExpanded1} and summing over $d = [\![d',\kappa]\!]$ from \eqref{eqn:equivocationLossMu1} (and substituting the summation variables $d \to i$ and $d' \to d$) yields  
\begin{equation}
    \label{eqn:defCPreliminary1}
    \begin{split}
        C(\kappa,d) = \sum_{i=d}^{\kappa}{\left( i \binom{\kappa-d}{i-d}_2 c(i,d) \right)}\mathrm{,}
    \end{split}
\end{equation}
and substituting in \eqref{eqn:def_c_constants} yields 
\begin{equation}
    \label{eqn:defCPreliminary2}
    \begin{split}
        C(\kappa,d) &= \sum_{i=d}^{\kappa}{\left( i \binom{\kappa-d}{i-d}_2 (-1)^{i-d} 2^{(i-d)(i-d-1)/2} \right)} \mathrm{.}
    \end{split}
\end{equation}

Considering the case of $d = \kappa$, three observations can aid the analysis. First, from this equation \eqref{eqn:defCPreliminary2}, it is easy to show that $C(\kappa,\kappa) = \kappa$. Second, the only dimension-$\kappa$ subspace of the global space $W$ is $W$ itself. That is, $\Xi(W^{\{\kappa\}},d = \kappa) = \{W\}$. Third, it is clear that $\Phi(W) = 1$. Combining these observations, the $d = \kappa$ case may be split out of \eqref{eqn:equivocationLossMu1} to yield 
\begin{equation}
    \label{eqn:equivocationLossMu3}
    \begin{split}
        \mathbb{E}\left[ \mathrm{rank}(G_{r(z)}) \mid \mu  \right] = \kappa  + \! \sum_{d=0}^{\kappa-1}{\!\left( \!\!\!\!\!\!\!\!\!\!\sum_{\;\;\;\;\;\;\;\;S^{\{d\}} \in \Xi(W,d)}{\!\!\!\!\!\!\!\!\!\!\!\!C(\kappa,d) \! \cdot \! \Phi(S)} \right)} \mathrm{.}
    \end{split}
\end{equation}

Next, we define the constants $K_{\delta}$ for $\delta \geq 1$ based on the $C(\kappa,d)$ as follows: 
\begin{equation}
    \label{eqn:defKBasedOnC}
    K_{\delta} = -C(\kappa,\kappa - \delta) \mathrm{.} 
\end{equation}
Then using \eqref{eqn:defCPreliminary2}, substituting $j=\kappa-i$, and utilizing \eqref{eqn:etaPrimeSum} from Lemma \ref{EtaPrimeSumLemma}, we obtain \begin{equation}
    \label{eqn:defKPreliminary1}
    \begin{split}
        K_{\delta} &= - \!\!\!\! \sum_{i=\kappa - \delta}^{\kappa}{\!\! \left( i \binom{\delta}{i \!-\! \kappa \!+\! \delta}_{\!2} \!(-1)^{i - \kappa + \delta} 2^{(i - \kappa + \delta)(i - \kappa + \delta-1)/2} \right)} \\
        &= - \sum_{j=0}^{\delta}{ \left( (\kappa -  j) \binom{\delta}{j}_{\!2} \!(-1)^{\delta - j} 2^{(\delta - j)(\delta - j-1)/2} \right)} \\
        &= - \sum_{j=0}^{\delta}{ \left( \kappa \cdot \eta(\delta,j) \right)} + \sum_{j=0}^{\delta}{ \left( j \cdot \eta(\delta,j) \right)} \\
        &= \sum_{j=0}^{\delta}{ \left( j \cdot \eta(\delta,j) \right)}
        \mathrm{.}
    \end{split}
\end{equation}

At this point, the expression for $K_\delta$ has no dependence on $\kappa$- that is, it is a universal sequence of constants. Additionally, substituting the relation \eqref{eqn:defKBasedOnC} into \eqref{eqn:equivocationLossMu3} with the additional substitution $d = \kappa - \delta$ yields 
\begin{equation}
    \label{eqn:equivocationLossMu4}
    \begin{split}
        \mathbb{E}\left[ \mathrm{rank}(G_{r(z)}) \mid \mu  \right] = \kappa -\sum_{\delta=1}^{\kappa}{\left(K_{\delta} \!\!\!\!\!\!\!\! \sum_{\;\;\;\;\;\;S \in \Xi(W,\kappa - \delta)}{\!\!\!\!\!\!\!\! \Phi(S)} \right)} \mathrm{,}
    \end{split}
\end{equation}
and substituting this into \eqref{eqn:equivocationLossMuPfister} yields exactly \eqref{eqn:expectedEquivocationMu}, as required. As mentioned previously, an identical analysis may be performed using $\psi(S)$ and $\phi(S)$ to yield \eqref{eqn:expectedEquivocationEpsilon}. All that remains, then, is to show that \eqref{eqn:expectedEquivocationConstantsPattern} gives correct values for the $K_\delta$ in \eqref{eqn:defKPreliminary1}, which is established in \eqref{eqn:KConstantsLemmaStatement} of Lemma \ref{EtaPrimeSumLemma}. 
\end{proof}

\subsection{Complexity of Equivocation Computation}
\label{sec:ComplexityAnalysis1}
The computation of $l(\epsilon,n,q)$ requires the evaluation of $\phi(S)$ for each subspace $S$ of $W$. The evaluation of $\phi(S^{\{d\}})$ requires $\mathcal{O}(1)$ floating point operations plus the evaluation of $\zeta(S^{\{d\}})$, which in turn requires $2^d$ floating point addition operations. Then for each dimension $d \in [\![0 , \kappa]\!]$, a total of $2^d \binom{\kappa}{d}_2$ operations are required. A well-known result of Euler~\cite{Bell2005EulerTheoremSummary} states that a product of consecutive Mersenne numbers approximates a corresponding product of consecutive powers of two to within a constant factor. Thus, using Bachmann-Landau notation, 
\begin{equation}
    \label{eqn:MersenneConstantBound}
    \prod_{i=a>0}^{b}{(2^i-1)} = \Theta\left(\prod_{i=a}^{b}{2^i}\right) = \Theta \left( 2^{(a+b)(a-b+1)/2} \right) \mathrm{.} 
\end{equation}
Because a Gaussian binomial is a quotient of two such Mersenne products, its order is given by 
\begin{equation}
    \label{eqn:GaussianBinomailOrder}
    \begin{split}
    \binom{\kappa}{d}_2 &= \frac{\prod_{i=\kappa-d+1}^{\kappa}{(2^i-1)}}{\prod_{i=1}^{d}{(2^i-1)}} = \Theta\left(\frac{2^{(2\kappa-d+1)(d)/2}}{2^{(d+1)(d)/2}} \right) = \Theta \left( 2^{(\kappa-d)(d)} \right) \mathrm{.} 
    \end{split}
\end{equation}
The number of operations required for a given dimension $d$ is then $\Theta(2^{(\kappa-d+1)d})$, which is maximized for odd $\kappa$ at a value of $d_{\mathrm{max}}=\frac{\kappa+1}{2}$ and for even $\kappa$ at $d_{\mathrm{max}} = \frac{\kappa}{2}$. In both cases, this gives a complexity of $\Theta(2^{(\kappa^2 + 2\kappa)/4})$. For any $d \neq d_{\mathrm{max}}$, the complexity falls exponentially in $(d-d_{\mathrm{max}})^2$, so the non-maximal $d$ terms do not contribute to the asymptotic complexity. Thus, the total complexity of the calculation of $l(\epsilon,n,q)$ is $\mathcal{O}(2^{(\kappa^2 + 2\kappa)/4})$ floating point operations. 

Recall that the complexity of calculating equivocation loss by enumerating all possible revealed bit patterns as described in Section \ref{sec:messageEquivocationDefinition} is $\mathcal{O}(n^2 2^n)$ operations in $\mathds{F}_2^{\kappa}$. Then if we assume that one operation in $\mathds{F}_2^{\kappa}$ is $\mathcal{O}(1)$ regardless of $\kappa$ (which is a reasonable assumption within the limits of practical computation on modern hardware), the subspace decomposition method provides an improvement provided that $n \gtrapprox (\kappa^2 + 2\kappa)/4$. 

\subsection{All-Zero Column Exclusion}
We now present several results relating to code optimality which are made possible by the continuous
nature of the functions $L(n,\mu,q)$ and $l(n,\epsilon,q)$ given by \eqref{eqn:expectedEquivocationMu} and \eqref{eqn:expectedEquivocationEpsilon}, respectively. The first result relates to the presence of the all-zero column in a coset code's generator matrix. 

\begin{theorem}
    \label{thm:zeroColumn}
    Any code definition $q$ satisfying the nonnegativity \eqref{eqn:QConstraintPositive} and unit sum \eqref{eqn:QConstraintTotal} constraints can be locally optimal in terms of expected equivocation loss $l(n,\epsilon,q)$ only if $q_0 = 0$. 
\end{theorem}
\begin{proof}
See Appendix \ref{Appendix:zcDerivative}. 
\end{proof}

\begin{corollary}
    \label{thm:zeroColumnCorollary}
    If a code defined by $q$ is locally optimal in terms of equivocation loss $l(n,\epsilon,q)$ with respect to all the elements $q_{i}$ of $q$ except $q_0$ and subject to nonnegativity constraint \eqref{eqn:QConstraintPositive} and unit sum constraint \eqref{eqn:QConstraintTotal}, and if $q_0 = 0$, then $q$ is locally optimal with respect to all elements $q_i$ of $q$ subject to \eqref{eqn:QConstraintPositive} and \eqref{eqn:QConstraintTotal}.     
\end{corollary}

Throughout the remainder of this work, Corollary \ref{thm:zeroColumnCorollary} is used to simplify the proofs of results related to code optimality, as the zero element $q_0$ of a code definition vector $q$ may be set to zero and treated as a constant. We use the variable $\qbar$ to represent such a simplified code definition. The vector $\qbar$ has indices in the range $[1 \dots 2^{\kappa}-1]$, and its use to define a code implies that the all-zero column is not used in the code. Thus $q$ may be extracted from $\qbar$ as 
\begin{equation}
     q = \left[ \begin{array}{c} 0 \\ \qbar \end{array} \right] \mathrm{.}
\end{equation}
With a slight abuse of notation, we allow all of the subspace decomposition-related functions which take $q$ as a parameter to use $\qbar$ as a parameter instead.

\subsection{Uniform Vector Fraction Construction}
We now introduce a code construction which we term the \emph{uniform vector fraction} construction. The code is defined by a vector $\bar{\qbar}$ with uniform element values. Because there are $2^\kappa-1$ elements in a $\qbar$ vector for a $\kappa$-dimensional code, the code definition vector is given by 
\begin{equation}
    \label{eqn:def_bar_q}
    \bar{\qbar}_i = \frac{1}{2^\kappa-1} \mathrm{.}
\end{equation}
The uniform vector fraction code definition is realizable for blocklengths which are integer multiples of $2^\kappa-1$, and for $n=2^\kappa-1$, the resulting code is the simplex code. 

\begin{theorem}
    \label{thm:uvf}
    The code defined by $\bar{\qbar}$ is locally optimal in terms of expected equivocation loss $l(n,\epsilon,\qbar)$ subject to the nonnegativity constraint \eqref{eqn:QConstraintPositive} and the unit sum constraint \eqref{eqn:QConstraintTotal}. 
\end{theorem}
\begin{proof}
See Appendix \ref{Appendix:UVF_Proof}
\end{proof}

\subsection{Subspace Exclusion Code}
In the previous section, we have shown that a locally optimal code may be constructed by forming a generator matrix which includes all possible nonzero columns in equal proportion. In most practical situations, however, the blocklength is less than $2^\kappa-1$, and it is not possible to include all nonzero columns in equal proportion. In this situation, it is natural to ask whether there are other values of $\qbar$ with similar optimality properties, but which satisfy the realizability constraint \eqref{eqn:QConstraintRealizable} for $n < 2^\kappa-1$. To answer this question, we introduce the \emph{subspace exclusion code} construction. 

A subspace exclusion code is a linear code defined by a generator matrix which contains all possible columns, except those which lie within a chosen subspace $U^{\{u\}}$ of the global space $W^{\{\kappa\}}$, in equal proportions. We denote the notation $\check{\qbar}^{[U]}$ for a subspace exclusion code vector which excludes all elements within subspace $U$. For purposes of coset code performance, any two such codes which exclude subspaces of the same dimension are equivalent, so it is convenient to define a canonical subspace exclusion code definition vector $\check{\qbar}^{\{u\}}$ for each dimension $u$. For simplicity, we choose the subspace $U$ containing the first $2^u$ column vectors, so the canonical subspace exclusion code definition vector is given by
\begin{equation}
    \label{eqn:sec_defSEC}
    \check{\qbar}^{\{u\}}_i = \begin{cases}
        0 &\text{if } i < 2^{u} \\
        \frac{1}{2^\kappa-2^u} &\text{otherwise.}
    \end{cases}
\end{equation}
Note that this code construction is realizable for $n=2^\kappa-2^u$. Note also that for $u=0$, this equates to the uniform vector fraction code. That is, $\check{\qbar}^{\{0\}} = \bar{\qbar}$. It is also worth noting that the first subspace exclusion code (with $u=\kappa-1$) with $n=2^{\kappa-1}$ defines the augmented Hadamard code. 

For $u>0$, $\check{\qbar}^{\{u\}}$ is not locally optimal because movement toward $\bar{\qbar}$ always reduces equivocation loss. Such movement is not conducive to the problem at hand, however, because all $q$ that satisfy the realizability constraint \eqref{eqn:QConstraintRealizable} for $n=2^\kappa-2^u$ have the same distance $|\qbar-\bar{\qbar}|$ from the uniform fraction code. Therefore it is natural to add the constraint that $q$ must maintain a constant radius from $\bar{\qbar}$. That is, we require that 
\begin{equation}
    \label{eqn:sec_radiusConstraint}
    |\qbar-\bar{\qbar}| = |\rho(u)| \mathrm{,}
\end{equation}
where $\rho(u)$ is the difference vector between $\bar{\qbar}$ and $\check{\qbar}^{\{u\}}$, and the elements of $\rho(u)$ are given by 
\begin{equation}
    \label{eqn:sec_defRadius}
    \begin{split}
    \rho(u)_i &= \begin{cases}
        -\frac{1}{2^\kappa-1} \text{ if } 1 \leq i < 2^u \\
        \frac{2^u-1}{(2^\kappa-1)(2^\kappa-2^u)} \text{ if } i \geq 2^u \mathrm{.}
    \end{cases} 
    \end{split}
\end{equation}
The magnitude of this vector is then given by 
\begin{equation}
    \label{eqn:sec_radiusMag}
    \begin{split}
    |\rho(u)| &= \! \sqrt{(2^u \! - \! 1) \! \left(\frac{1}{2^\kappa \! - \! 1}\right)^{\!2} \!\! \!+\! (2^\kappa \! - \! 2^u) \! \left(\frac{1}{2^\kappa \! - \! 2^u}\!-\!\frac{1}{2^\kappa \! - \! 1}\right)^{\!2}} = \sqrt{\frac{2^u-1}{(2^\kappa \! - \! 2^u) (2^\kappa-1)}}\mathrm{.} 
    \end{split}
\end{equation}

Using the radius constraint \eqref{eqn:sec_radiusConstraint}, we may make the following statement about the optimality of the first ($u=\kappa-1$) subspace exclusion code: 

\begin{theorem}
    \label{thm:sec_localOptimality1}
    The code definition vector $\check{\qbar}^{\{\kappa-1\}}$ is locally optimal in terms of equivocation loss $l(n,\epsilon,\qbar)$ subject to the nonnegativity constraint \eqref{eqn:QConstraintPositive}, the unit-sum constraint \eqref{eqn:QConstraintTotal}, and the radius constraint \eqref{eqn:sec_radiusConstraint}. 
\end{theorem}
\begin{proof}
    See Appendix \ref{Appendix:secLocalOptimalityProof}. 
\end{proof}

\subsection{$\chi^2$ Divergence}
Theorem \ref{thm:ExpectedEquivocationFormula} provides a succinct expression for the equivocation of a coset code. This expression, however, requires a summation over all the subspaces of the global space. As shown in \ref{sec:ComplexityAnalysis1}, for large $n$, this calculation is simpler than the brute force calculation based on all possible erasure patterns. The complexity is, however, still prohibitive for large $\kappa = n-k$. An interesting effect is observed, though, when considering an alternative metric. Specifically, we consider the $\chi^2$ divergence between joint and marginal distributions on $M$ and $Z$. This approach is described in detail below. 

The $\chi^2$ divergence between discrete probability distributions $p_A(a)$ and $p_B(b)$ with alphabets $\mathcal{A}$ and $\mathcal{B}$, respectively, is given by 
\begin{equation}
    \label{eqn:x2_generalFormulation}
    \chi^2(p_A(a),p_B(b)) = \sum_{a \in \mathcal{A}, b \in \mathcal{B}}{\left( \frac{p_A(a)}{p_B(b)} - 1 \right)^2 p_B(b)} \mathrm{.}
\end{equation}
As a metric for secrecy code performance, we desire the $\chi^2$ divergence between the joint distribution $p_{MZ}$ of $M$ and $Z$ and the product $p_Mp_Z$ of the marginal distributions of $M$ and $Z$ (which represents the ideal joint distribution). For a code defined by $q$ of blocklength $n$, We denote this divergence as $\Lambda(n,\mu,q)$ for the case of a fixed number $\mu$ of revealed bits or $\lambda(n,\epsilon,q)$ for the case of fixed erasure probability $\epsilon$. Using \eqref{eqn:x2_generalFormulation}, we can define these functions as 
\begin{equation}
    \label{eqn:x2_Lambda_def}
    \Lambda(n,\mu,q) = \!\!\!\!\!\!\!\!\!\!\sum_{\;\;\;\;\substack{m\in \{0,1\}^{k}, \\ z\in\{0,1,?\}^n}}{\!\! \left( \frac{p_{MZ}(m,z|\mu)}{p_M(m)p_Z(z|\mu)}-1\right)^2p_M(m) p_Z(z|\mu)}
\end{equation}
and 
\begin{equation}
    \label{eqn:x2_lambda_def}
    \lambda(n,\epsilon,q) = \!\!\!\!\!\!\!\!\!\!\sum_{\;\;\;\;\substack{m\in \{0,1\}^{k}, \\ z\in\{0,1,?\}^n}}{\!\! \left( \frac{p_{MZ}(m,z|\epsilon)}{p_M(m)p_Z(z|\epsilon)}-1\right)^2p_M(m) p_Z(z|\epsilon)} \mathrm{.}
\end{equation}


Considering the expression in the sum of \eqref{eqn:x2_Lambda_def}, we note that in the case that there is no equivocation loss, $p_{MZ}(m,z)$ is uniform across $m$, and $p_{MZ}(m,z)-p_M(m)p_Z(z)$=0. In the case of $l$ bits of equivocation loss, $p_{MZ}(m,z)$ is equal to zero for all those $m$ which are inconsistent with the revealed bits and uniform over all those $m$ which are consistent. Conveniently, the number of possible messages consistent with a given codeword is calculated from the equivocation loss as 
\begin{equation}
    \label{eqn:x2_consistentMessageFraction}
    \lvert\{m: m_i = z_i \;\; \forall \; i:z_i \neq ? \}\rvert = 2^{H(M|Z=z)} \mathrm{.}
\end{equation}
Then for a given $m$ and $z$, and again assuming uniformly distributed $M$, the joint distribution is given by 
\begin{equation}
    \label{eqn:x2_jointDistribution}
    p_{MZ}(m,z) = \begin{cases}
        2^{-H(M|Z=z)} \cdot p_Z(z)\!\! &\text{if } m  \text{ consistent with } z \\
        0 &\text{otherwise.} 
    \end{cases} 
\end{equation}
Then the sum required in \eqref{eqn:x2_Lambda_def} and \eqref{eqn:x2_lambda_def} can be computed for all $m$ and for a given $z$ as 
\begin{equation}
    \label{eqn:x2_equivocation1}
    \begin{split}
    &\sum_{\;\;\;\;m\in \{0,1\}^{k}}{\left( \frac{p_{MZ}(m,z)}{p_M(m)p_Z(z)}-1\right)^2p_M(m) p_Z(z)} \\ 
    & \;\;\;\;\;\;\;\;\;\;\;\; = \!\!\!\!\!\!\!\! \sum_{\;\;\;\;m \text{ inconsistent with } z}{\frac{\left( p_{MZ}(m,z) - p_M(m)p_Z(z) \right)^2}{p_M(m) p_Z(z)}} + \!\!\!\!\!\!\!\! \sum_{\;\;\;\;m \text{ consistent with } z}{\frac{\left( p_{MZ}(m,z) - p_M(m)p_Z(z) \right)^2}{p_M(m) p_Z(z)}} \\
    & \;\;\;\;\;\;\;\;\;\;\;\; = (2^k - 2^{H(M|Z=z)}) \frac{(0 - 2^{-k}p_Z(z))^2}{2^{-k}p_Z(z)} + 2^{H(M|Z=z)} \frac{(2^{-H(M|Z=z)}p_Z(z) - 2^{-k}p_Z(z))^2}{2^{-k}p_Z(z)} \\
    & \;\;\;\;\;\;\;\;\;\;\;\; = p_Z(z) \left(  1-2^{H(M|Z=z)-k} \right) + p_Z(z) \left( 2^{k-H(M|Z=z)} - 2 + 2^{H(M|Z=z)-k} \right) \\
    & \;\;\;\;\;\;\;\;\;\;\;\; = p_Z(z) \left( 2^{k-H(M|Z=z)} - 1\right) \mathrm{.}
    \end{split}
\end{equation}
Taking this sum over $z$ as required in \eqref{eqn:x2_Lambda_def} constitutes an expectation: 
\begin{equation}
    \label{eqn:x2_LambdaExpectation}
    \begin{split}
    \Lambda(n,\mu,q) &= \!\! \sum_{\;\; z: \{0,1,?\}^n}{\!\! \left(p_Z(z) \left( 2^{k-H(M\left|Z=z,\mu\right.)} - 1\right) \right)} = \mathds{E}\left( 2^{k-H(M|Z,\mu)}\right)-1 \mathrm{,}    
    \end{split}
\end{equation}
and similarly for fixed $\epsilon$,
\begin{equation}
    \label{eqn:x2_lambdaExpectation}
    \begin{split}
    \lambda(n,\epsilon,q) &= \!\! \sum_{\;\; z: \{0,1,?\}^n}{\!\! \left(p_Z(z) \left( 2^{k-H(M\left|Z=z,\epsilon\right.)} - 1\right) \right)} = \mathds{E}\left( 2^{k-H(M|Z,\epsilon)}\right)-1 \mathrm{.}    
    \end{split}
\end{equation}

To gain some intuitive understanding of the meaning of the $\chi^2$ divergence applied to coset code performance, it is helpful to contrast the $\chi^2$ divergence as expressed in \eqref{eqn:x2_lambdaExpectation} with equivocation loss, which may be expressed as the expectation 
\begin{equation}
    \label{eqn:equivocationLossAsExpectation}
    l(n,\epsilon,q) = \mathds{E}(k-H(M|Z)) \mathrm{.}
\end{equation}
It may be observed that the $\chi^2$ divergence is an exponentially-weighted expectation on the equivocation loss (with a unit offset, which ensures that the $\chi^2$ divergence is zero if no information is leaked). That is, if we consider a penalty function for information revealed to an eavesdropper, the equivocation loss metric imposes a unit penalty for each bit of information leaked. On the other hand, the $\chi^2$ divergence metric imposes a unit penalty for the first bit leaked, with the penalty doubling on each subsequent bit leaked. Such a penalty function provides a penalty value inversely proportional to number of messages consistent with the eavesdropper's observation, or directly proportional to the probability of successful maximum likelihood (ML) decoding by the eavesdropper. This penalty function thus has practical application in situations such as, e.g., communication of a passcode in which the adversary is allowed only a single attempt to produce the passcode (or in which the number of allowed attempts is small compared to the number of possibilities).

\subsection{$\chi^2$ Divergence Via Subspace Decomposition}
We next show how subspace decomposition may be used to establish a remarkable expression for the $\chi^2$ divergence for a given code over the BEC. Specifically, we derive an expression which gives the $\chi^2$ divergence as a sum of functions only of the dimension-$(\kappa-1)$ subspaces of $W$. 
\begin{theorem}
    \label{thm:x2_Divergence}
For a coset code defined by $q$ of blocklength $n$, the $\chi^2$ divergence between the joint distribution $p_{MZ}(m,z)$ and the product $p_M(m)p_Z(z)$ of the marginal distributions of $m$ and $z$ is equal to 
\begin{equation}
    \label{eqn:x2_LambdaFinal}
    \Lambda(n,\mu,q) = 2^{\mu - \kappa} \left( 1+ \!\! \sum_{S \in \Xi(W,\kappa-1)} \!\!\!\!\!\!  \Phi(S,n,\mu,q) \right) -1 \mathrm{,} 
\end{equation}
for a fixed number $\mu$ of revealed bits and is equal to 
\begin{equation}
    \label{eqn:x2_lambdaFinal}
    \lambda(n,\epsilon,q) = (2-\epsilon)^{n}2^{-\kappa} \left( 1+ \!\! \sum_{S \in \Xi(W,\kappa-1)} \!\!\!\!\!\!  \varphi(S,n,\epsilon,q) \right) -1 \mathrm{,} 
\end{equation}
for a fixed erasure probability $\epsilon$, where
\begin{equation}
    \label{eqn:x2_def_varphi}
    \varphi(S,n,\epsilon,q) = \left(\frac{\epsilon}{2-\epsilon}\right)^{n(1-\zeta(S,q))} \mathrm{.} 
\end{equation}

\end{theorem}

\begin{proof}
We first consider the case of a fixed number $\mu$ of revealed bits. The proof begins in the same manner as the proof of Theorem \ref{thm:ExpectedEquivocationFormula}, with the calculation of $\Psi(S,n,\mu)$ for each subspace $S$ of $W$. Recall that the value of $\Psi(S,n,\mu)$ represents the probability of $\mu$ revealed bits exactly spanning $S$. Recall further that if the revealed bits $r(z)$ of a received codeword $z$ exactly span a space of dimension $d$, then $\mathrm{rank}(G_{r(z)}) = d$. Thus combining \eqref{eqn:x2_LambdaExpectation} and \eqref{eqn:equivocationLossMuPfister}, the term in the expectation of \eqref{eqn:x2_LambdaExpectation} for a revealed bit pattern of size $\mu$ that exactly spans a space of dimension $d$ is given by 
\begin{equation}
    \label{eqn:x2_rankByDimension}
    2^{k-H(M|Z=z)} = 2^{\mu - d} \mathrm{.} 
\end{equation}
Using this expression, $\Lambda(n,\mu,q)$ may be expressed as  
\begin{equation}
    \label{eqn:x2_Lambda1}
    \begin{split}
    \Lambda(n,\mu,q) &= \mathds{E} \left( 2^{k-H(M|Z)} | \;\lvert r(Z) \rvert = \mu \right)-1 \\
    &= \sum_{d=0}^{\kappa}{\left( \sum_{S \in \Xi(W,d)}{ \left( 2^{\mu-d} \Psi(S,n,\mu,q) \right) } \right) }-1 \\
    &= 2^{\mu} \sum_{d=0}^{\kappa}{\left( \sum_{S \in \Xi(W,d)}{ \left( 2^{-d} \Psi(S,n,\mu,q) \right) } \right) }-1 \mathrm{.}
    \end{split}
\end{equation}
Then, because each $\Psi(S,n,\mu,q)$ consists of a weighted sum of $\Phi(T,n,\mu,q)$ terms for the subspaces $T$ of $S$, $\Lambda(n,\mu,q)$ may be expressed as 
\begin{equation}
    \label{eqn:x2_Lambda_phi}
    \Lambda(\mu,n,q) = 2^{\mu} \sum_{d=0}^{\kappa}{\left( \sum_{S \in \Xi(W,d)}{ \left( \Gamma(\kappa,d) \Phi(S,n,\mu,q)\right) } \right) }-1 \mathrm{.}
\end{equation}
where the $\Gamma(\kappa,d)$ are a series of weighting constants. Making use of \eqref{eqn:numberOfSuperspaces} as in Theorem \ref{thm:ExpectedEquivocationFormula}, $\Gamma(\kappa,d)$ may be expressed as a weighted sum of the $c(d,d')$ defined in \eqref{eqn:def_c_constants}. Similarly to the expression for the $C(\kappa,d)$ constants in \eqref{eqn:defCPreliminary1}, the $\Gamma(\kappa,d)$ constants are equal to the $c(d,d')$ constants multiplied by the number of superspaces of each $d$-dimensional subspace, but instead of weighting by the intermediate dimension $i$, the terms are weighted by $2^{-i}$. This yields 
\begin{equation}
    \label{eqn:x2_Gamma_c_1}
    \begin{split}
    \Gamma(\kappa,d) &= \sum_{i=d}^{\kappa}{\left( 2^{-i} \binom{\kappa-d}{i-d}_2 c(i,d) \right)} \\
    &= \sum_{i=d}^{\kappa}{\left( 2^{-i} \binom{\kappa-d}{i-d}_2 \eta(i-d,0) \right)} \mathrm{,} 
    \end{split}
\end{equation}
or in terms of $\delta=\kappa-d$, 
\begin{equation}
    \label{eqn:x2_Gamma_delta}
    \begin{split}
    \Gamma(\kappa,\kappa-\delta) &= \sum_{i=\kappa-\delta}^{\kappa}{\left( 2^{-i} \binom{\delta}{i-\kappa+\delta}_2 \eta(i-\kappa+\delta,0) \right)} \\
    &= \sum_{i=\kappa-\delta}^{\kappa}{\left( 2^{-i} \eta(\delta,\kappa-i) \right)}  \\
    &= 2^{-\kappa} \sum_{j=0}^{\delta}{\left( 2^{j} \eta(\delta,j) \right)} \mathrm{.} 
    \end{split}
\end{equation}
It is clear that for $d=\kappa$ (or $\delta=0$), we have 
\begin{equation}
    \label{eqn:x2_GammaOf0}
    \Gamma(\kappa,\kappa) = 2^{-\kappa}  \eta(0,0) = 2^{-\kappa} \mathrm{.}
\end{equation}
For $d < \kappa$, we may begin by splitting out the first ($j = 0$) term and transforming it using \eqref{eqn:etaSplitProperty} to yield
\begin{equation}
    \label{eqn:x2_Gamma_eta_1}
    \begin{split}
    \Gamma(\kappa,\kappa-(\delta>0)) &= 2^{-\kappa} \sum_{j=1}^{\delta}{\left( 2^{j} \eta(\delta,j) \right)} + 2^{-\kappa} \eta(\delta,0) \\
    &\!\!\!\!\!\!\!\!\!\!\!\!\!\!\!\!\!\!\! \!\!\!\!\!\!\!\!\!\!\!\!\!\!\!\!\!\!\! = 2^{-\kappa} \sum_{j=1}^{\delta}{\left( 2^{j} \eta(\delta,j) \right)} - 2^{-\kappa} \sum_{i=0}^{\delta-1}{\left( \binom{\delta}{i}_2 \eta(i,0) \right)} \mathrm{.}
    \end{split}
\end{equation}
Next, by substituting $j=\delta-i$ in the second sum, the sums may be combined, giving
\begin{equation}
    \label{eqn:x2_Gamma_eta_2} 
    \begin{split}
    \Gamma(\kappa,\kappa-(\delta>0)) &= \\
    &\!\!\!\!\!\!\!\!\!\!\!\!\!\!\!\!\!\!\! \!\!\!\!\!\!\!\!\!\!\!\!\!\!\!\!\!\!\! = 2^{-\kappa} \sum_{j=1}^{\delta}{\left( 2^{j} \eta(\delta,j) \right)} - 2^{-\kappa} \sum_{j=1}^{\delta}{\left( \binom{\delta}{\delta-j}_2 \eta(\delta-j,0) \right)} \\
    &\!\!\!\!\!\!\!\!\!\!\!\!\!\!\!\!\!\!\! \!\!\!\!\!\!\!\!\!\!\!\!\!\!\!\!\!\!\! = 2^{-\kappa} \sum_{j=1}^{\delta}{\left( 2^{j} \eta(\delta,j) \right)} - 2^{-\kappa} \sum_{j=1}^{\delta}{\left( \eta(\delta,j) \right)}  \\
    &\!\!\!\!\!\!\!\!\!\!\!\!\!\!\!\!\!\!\! \!\!\!\!\!\!\!\!\!\!\!\!\!\!\!\!\!\!\! = 2^{-\kappa} \sum_{j=1}^{\delta}{\left( (2^{j}-1) \eta(\delta,j) \right)} \mathrm{.}
    \end{split}    
\end{equation}
Then using the definition \eqref{eqn:defEtaPrime} of $\eta'(\cdot) = \eta(\cdot)$ and properties of Gaussian binomials, we obtain 
\begin{equation}
    \label{eqn:x2_Gamma_eta_3} 
    \begin{split}
    \Gamma(\kappa,\kappa-(\delta>0)) &= \\
    &\!\!\!\!\!\!\!\!\!\!\!\!\!\!\!\!\!\!\! \!\!\!\!\!\!\!\!\!\!\!\!\!\!\!\!\!\!\! = 2^{-\kappa} \sum_{j=1}^{\delta}{\left( (2^{j}-1) \binom{\delta}{j}_2 c(\delta,j) \right)} \\
    &\!\!\!\!\!\!\!\!\!\!\!\!\!\!\!\!\!\!\! \!\!\!\!\!\!\!\!\!\!\!\!\!\!\!\!\!\!\! = 2^{-\kappa} \sum_{j=1}^{\delta}{\left( (2^{\delta}-1) \binom{\delta-1}{j-1}_2 c(\delta-1,j-1) \right)} \\
    &\!\!\!\!\!\!\!\!\!\!\!\!\!\!\!\!\!\!\! \!\!\!\!\!\!\!\!\!\!\!\!\!\!\!\!\!\!\! = 2^{-\kappa} \sum_{j=1}^{\delta}{\left( (2^{\delta}-1) \eta(\delta-1,j-1) \right)}  \\
    &\!\!\!\!\!\!\!\!\!\!\!\!\!\!\!\!\!\!\! \!\!\!\!\!\!\!\!\!\!\!\!\!\!\!\!\!\!\! = 2^{-\kappa} (2^{\delta}-1) \sum_{i=0}^{\delta-1}{\left(  \eta(\delta-1,i) \right)} \mathrm{.}
    \end{split}
\end{equation}
Now, by \eqref{eqn:etaPrimeSum}, \eqref{eqn:x2_Gamma_eta_3} equals zero for all values of $\delta$ except $\delta=1$, at which we have 
\begin{equation}
    \label{eqn:x2_GammaOf0}
    \Gamma(\kappa,\kappa-1) = 2^{-\kappa}  (2^{1}-1) = 2^{-\kappa} \mathrm{.}
\end{equation}
Thus, the closed-form solution to \eqref{eqn:x2_Gamma_c_1} is 
\begin{equation}
    \label{eqn:x2_Gamma_Final}
    \Gamma(\kappa,d) = \begin{cases}
        2^{\mu-\kappa} &\text{if } d = \kappa \text{ or } d=\kappa-1 \\
        0 &\text{otherwise} \mathrm{.} 
    \end{cases}
\end{equation}
Substituting this into \eqref{eqn:x2_Lambda_phi} yields \eqref{eqn:x2_LambdaFinal} as required. 

Next, we consider the case of fixed $\epsilon$. Here, the value of $2^{k-H(M|Z)}$ is not fixed even if it is known that the revealed bits exactly span a particular space $S$. It is therefore necessary to leave the $2^{k-H(M|Z=z)}$ term in an expectation, yielding 
\begin{equation}
    \label{eqn:x2_lambda1}
    \begin{split}
    \lambda(n,\epsilon,q) &= \mathds{E} \left( 2^{k-H(M|Z)} \right)-1 \\
    &\!\!\!\!\!\!\!\!\!\!\!\!\!\!\!\!= \sum_{d=0}^{\kappa}{\left( \sum_{S \in \Xi(W,d)}{ \psi(S) \mathds{E} \left( 2^{k-H(M|Z)} | S \right) } \right) }-1  \\
    &\!\!\!\!\!\!\!\!\!\!\!\!\!\!\!\!= \sum_{d=0}^{\kappa}{\left( \sum_{S \in \Xi(W,d)}{ 2^{-d} \psi(S) \mathds{E} \left( 2^{|r(z)|} | S \right) } \right) }-1 \mathrm{.}
    \end{split}
\end{equation}
(The notation $\mathds{E}(X|S)$ where $S$ is a space indicates the expectation conditioned on the set of revealed bits exactly spanning $S$.) 

Considering the expectation $\mathds{E} \left( 2^{|r(z)|} | S \right)$, using the properties of conditional expectation and the facts that the events that $r(z)$ exactly spans any given $T \subseteq S$ are disjoint and that their union is the event that $r(z)$ spans $S$, we can write
\begin{equation}
    \label{x2_expectedValuePsi}
    \begin{split}
    \psi(S^{\{d\}})\mathds{E}\left( 2^{|r(z)|} | S \right) &= \phi(S)\mathds{E}\left( 2^{|r(z)|} | r(z) \subset S \right) - \sum_{d'=0}^{d-1}{\left(\sum_{T \in \Xi(S,d')}{\psi(T)\mathds{E}\left( 2^{|r(z)|} | T \right)}\right)} \\
    &= \epsilon^{n(1-\zeta(S))}\left( (2-\epsilon)^{n\zeta(S))} \right) - \sum_{d'=0}^{d-1}{\left(\sum_{T \in \Xi(S,d')}{\psi(T)\mathds{E}\left( 2^{|r(z)|} | T \right)}\right)} \mathrm{.}
    \end{split}
\end{equation}
Now using $\varphi(S,n,\epsilon,q)$ as defined in \eqref{eqn:x2_def_varphi}, we obtain 
\begin{equation}
    \label{x2_expectedValuePsi2}
    \begin{split}
    &\psi(S^{\{d\}})\mathds{E}\left( 2^{|r(z)|} | S \right) = (2-\epsilon)^n \varphi(S,n,\epsilon,q) - \!\!\sum_{d'=0}^{d-1}{\!\! \left( \!\!\!\!\!\! \sum_{\;\;\;\;T \in \Xi(S,d')}{\!\!\!\!\!\! \psi(T)\mathds{E}\left( 2^{|r(z)|} | T \right)}\!\! \right)} \mathrm{.}
    \end{split}
\end{equation}
At this point, \eqref{x2_expectedValuePsi2} may be expanded recursively in a similar fashion to \eqref{eqn:PsiExpanded1} to yield 
\begin{equation}
    \label{x2_expectedValuePsi_varphi}
    \begin{split}
    &\psi(S^{\{d\}})\mathds{E}\left( 2^{|r(z)|} | S \right) = \sum_{d'=0}^{d-1}{\left( 
 \!\!\!\!\!\sum_{\;\;\;\;T \in \Xi(S,d')}{\!\!\!\!\!\!c(d,d') (2\!-\!\epsilon)^n \varphi(T)} \!\right)} \mathrm{.}
    \end{split}
\end{equation}
Then substituting this expression into the expression \eqref{eqn:x2_lambda1} for $\lambda(n,\epsilon,q)$, it becomes clear that the sum may be expressed as a weighted sum of $(2\!-\!\epsilon)^n \varphi(S)$ terms for the subspaces $S$ of $W$. It is also clear upon inspection that the weighting constants are given precisely by the same expression \eqref{eqn:x2_Gamma_c_1} used for weighting the $\Phi(S)$ terms when calculating $\Lambda(\mu,n,q)$. The resulting expression for $\lambda(n,\epsilon,q)$ is 
\begin{equation}
    \label{eqn:x2_lambda_Gamma_1}
    \lambda(n,\epsilon,q) = \sum_{d=0}^{\kappa}{\left(\sum_{S \in \Xi(W,d)}{\left( \Gamma(\kappa,d) (2-\epsilon)^n \varphi(S) \right)}  \right)}-1 \mathrm{,}
\end{equation}
or using the values of $\Gamma(\kappa,d)$ given in \eqref{eqn:x2_Gamma_Final}, 
\begin{equation}
    \label{eqn:x2_lambda_Gamma_1}
    \begin{split}
    \lambda(n,\epsilon,q) &= \sum_{S \in \Xi(W,\kappa)}{\left( 2^{-\kappa} (2-\epsilon)^n \varphi(S) \right)} + \!\!\!\!\sum_{S \in \Xi(W,\kappa-1)}{\left( 2^{-\kappa} (2-\epsilon)^n \varphi(S) \right)}-1 \\
    &\;\;\;\;= 2^{-\kappa} (2-\epsilon)^n + \!\!\!\!\sum_{S \in \Xi(W,\kappa-1)}{\left( 2^{-\kappa} (2-\epsilon)^n \varphi(S) \right)}-1 \\
    &\;\;\;\;= 2^{-\kappa} (2-\epsilon)^n \left( 1 + \!\!\!\!\sum_{S \in \Xi(W,\kappa-1)}{ \varphi(S) } \right)-1 \mathrm{,}
    \end{split}
\end{equation}
as required by \eqref{eqn:x2_lambdaFinal}. 
\end{proof}

\subsection{$\chi^2$ Divergence Complexity Analysis}
Calculation of the $\chi^2$ divergence via subspace decomposition is significantly simpler than the calculation of equivocation loss by subspace decomposition. Examining \eqref{eqn:x2_LambdaFinal} and \eqref{eqn:x2_lambdaFinal}, we observe that it is necessary to calculate $\Phi(S,n,\mu,q)$ or $\varphi(S,n,\epsilon,q)$ for each of the $(2^{\kappa}-1)$ dimension-$(\kappa-1)$ subspaces of $W$. Calculating $\Phi(S^{\{\kappa-1\}},n,\mu,q)$ or $\varphi(S^{\{\kappa-1\}},n,\epsilon,q)$ requires calculating $\zeta(S,q)$, which in turn requires a summation over all of the $2^{\kappa-1}$ elements of $S^{\{\kappa-1\}}$. The total complexity to calculate $\Lambda(n,\mu,q)$ or $\lambda(n,\epsilon,q)$ is therefore $\mathcal{O}(2^{2\kappa})$. 

The complexity of calculating the $\chi^2$ divergence via enumeration of all possible revealed bit patterns is the same as that, $\mathcal{O}(n^2 2^n)$, of calculating equivocation loss by the same method (because the quantity in the expectation used for the $\chi^2$ divergence is obtained from the equivocation loss by an $\mathcal{O}(1)$ operation). When comparing this complexity to the $\mathcal{O}(2^{2\kappa})$ required to compute $\lambda(n,\epsilon,q)$ by subspace decomposition, we observe that subspace decomposition has an advantage provided that $n \gtrapprox 2\kappa$, i.e., for a code of rate greater than one half. 

\subsection{All-Zero Column Exclusion for $\chi^2$ Divergence}
In seeking codes that are optimal in terms of the $\chi^2$ divergence, it is helpful to first establish a result analogous to Theorem \ref{thm:zeroColumn} relating to the all-zero column. 
\begin{theorem}
    \label{thm:x2_zc}
    Any code definition $q$ satisfying the nonnegativity \eqref{eqn:QConstraintPositive} and unit sum \eqref{eqn:QConstraintTotal} constraints can be locally optimal in terms of $\chi^2$ divergence $\lambda(n,\epsilon,q)$ only if $q_0 = 0$. 
\end{theorem}
\begin{proof}
See Appendix \ref{apx:x2_zc_proof}. 
\end{proof}

\begin{corollary}
    \label{thm:x2_zcCorollary}
    If a code defined by $q$ is locally optimal in terms of $\chi^2$ divergence $\lambda(n,\epsilon,q)$ with respect to all the elements $q_{i}$ of $q$ except $q_0$ and subject to nonnegativity constraint \eqref{eqn:QConstraintPositive} and unit sum constraint \eqref{eqn:QConstraintTotal}, and if $q_0 = 0$, then $q$ is locally optimal with respect to all elements $q_i$ of $q$ subject to \eqref{eqn:QConstraintPositive} and \eqref{eqn:QConstraintTotal}.     
\end{corollary}
Using these results, we may perform optimality calculations for the $\chi^2$ divergence using the simplified code definition vector $\qbar$ instead of $q$. 

\subsection{$\chi^2$ Divergence of the Uniform Vector Fraction Code}
In Theorem \ref{thm:uvf}, it was shown that the uniform vector fraction code is locally optimal for its size in terms of expected equivocation loss. A similar analysis may be performed using $\chi^2$ divergence instead of equivocation loss, and because of the simplicity of the expression \eqref{eqn:x2_lambdaFinal} for $\lambda(n,\epsilon,q)$, this analysis reaches a stronger conclusion than it does for equivocation loss. Specifically, the uniform vector fraction code is shown to be globally, rather than locally, optimal. 
\begin{theorem}
    \label{thm:x2_uvf}
    The code defined by $\bar{\qbar}$ is globally optimal in terms of $\chi^2$ divergence $\lambda(n,\epsilon,q)$, subject to the nonnegativity constraint \eqref{eqn:QConstraintPositive} and the unit sum constraint \eqref{eqn:QConstraintTotal}. 
\end{theorem}
\begin{proof}
    See Appendix \ref{apx:x2_uvf_proof}. 
\end{proof}
Because the uniform vector fraction code is both globally optimal and realizable as the simplex code for $n=2^{\kappa}-1$, we may also make a conclusion about the performance of the simplex code. 
\begin{corollary}
    \label{thm:x2_simplexCorollary}
     Simplex codes are optimal for their size as secrecy codes over the BEWC in terms of ML eavesdropper decoding.
\end{corollary}

\subsection{$\chi^2$ Divergence of the First Subspace Exclusion Code}
In similar fashion to the expanded conclusion reached in Theorem \ref{thm:x2_uvf}, the simplicity of the expression \eqref{eqn:x2_lambdaFinal} for $\lambda(n,\epsilon,q)$ also provides for stronger conclusions regarding the $\chi^2$ divergence performance of subspace exclusion codes. In the case of the first subspace exclusion code, we impose the same constraints under the $\chi^2$ divergence as under equivocation loss and show the code to be globally, rather than locally, optimal. 
\begin{theorem}
    \label{thm:x2_sec1}
    The code definition vector $\check{\qbar}^{\{\kappa-1\}}$ is globally optimal in terms of $\chi^2$ divergence $\lambda(n,\epsilon,\qbar)$ subject to the nonnegativity constraint \eqref{eqn:QConstraintPositive}, the unit-sum constraint \eqref{eqn:QConstraintTotal}, and the radius constraint \eqref{eqn:sec_radiusConstraint}. 
\end{theorem}
\begin{proof}
    See Appendix \ref{apx:x2_sec1Proof}. 
\end{proof}
Because the augmented Hadamard code is equivalent to the first subspace exclusion code (with $n=2^{\kappa-1}$), we may also make a conclusion about the optimality of these codes. 
\begin{corollary}
    Augmented Hadamard codes are optimal for their size as secrecy codes over the BEWC in terms of ML eavesdropper decoding.
\end{corollary}

\subsection{$\chi^2$ Divergence of Other Subspace Exclusion Codes}
When considering subspace exclusion codes other than the first (that is, with $u<\kappa-1$), we may evaluate optimality using a technique similar to that used in Theorem \ref{thm:x2_sec1}. When the applicable constraints (\eqref{eqn:QConstraintPositive}, \eqref{eqn:QConstraintTotal}, and \eqref{eqn:sec_radiusConstraint}) are applied, however, the code definition vector $\check{\qbar}^{\{u\}}$ is in general not even locally optimal, as $\chi^2$ divergence is minimized for a given radius by excluding as large a subspace as possible. At the radius $|\rho(u)|$ required for the dimension-$u$ subspace exclusion code, however, excluding a subspace with dimension higher than $u$ results in violation of the realizability constraint \eqref{eqn:QConstraintRealizable} as well as the nonnegativity constraint \eqref{eqn:QConstraintPositive}, preventing the resulting $\qbar$ from defining a realizable code. To overcome this limitation, the natural solution is to add the constraint that $\qbar$ must maintain at least a minimum distance from any previous (higher-$u$) subspace exclusion code. Such a limitation is indeed sufficient to ensure that $\check{\qbar}^{\{u\}}$ is optimal. In our analysis, however, we show that it is in fact only necessary to require that $\qbar$ maintain a minimum distance from each of the first ($u=\kappa-1$) subspace exclusion codes. The minimum required distance is equal to the distance between $\check{\qbar}^{\{u\}}$ and $\check{\qbar}^{\{\kappa-1\}}$, so the constraint may be expressed as 
\begin{equation}
    \label{eqn:x2_QConstraintSEC1}
    |\qbar-\check{\qbar}^{[S]}| \geq |\check{\qbar}^{\{u\}} - \check{\qbar}^{\{\kappa-1\}}| \;\; \forall \;\;S \in \Xi(W,\kappa-1) \mathrm{.} 
\end{equation}
Using this constraint, all subspace exclusion codes $\check{\qbar}^{\{u\}}$ may be shown to be optimal. 
\begin{theorem}
    \label{thm:x2_sec2}
     The code definition vectors $\check{\qbar}^{\{u\}}$ are globally optimal in terms of $\chi^2$ divergence $\lambda(n,\epsilon,\qbar)$ subject to the nonnegativity constraint \eqref{eqn:QConstraintPositive}, the unit-sum constraint \eqref{eqn:QConstraintTotal}, the radius constraint \eqref{eqn:sec_radiusConstraint}, and the constraint \eqref{eqn:x2_QConstraintSEC1} on the minimum distance from a first subspace exclusion code. 
\end{theorem}
\begin{proof}
    See Appendix \ref{apx:x2_sec2Proof}. 
\end{proof}

Because the additional constraint \eqref{eqn:x2_QConstraintSEC1} does not eliminate any values of $\qbar$ which are realizable with $n=2^{\kappa}-2^{u}$, we may also make conclusions regarding realizations of subspace exclusion codes. 

\begin{corollary}
    \label{thm:x2_SEC2_Corollary}
     Subspace exclusion codes are optimal for their size as secrecy codes over the BEWC in terms of ML eavesdropper decoding.
\end{corollary}

\section{Conclusion}
Analysis of coset codes by the properties of their subspaces provides an intriguing new approach to understanding and designing short-blocklength secrecy codes. In this work, we have shown the utility of this approach in efficiently calculating equivocation, and we have demonstrated several analytical results relating to code performance. We have also used subspace decomposition to identify a number of desirable properties of the $\chi^2$ divergence as a metric for performance of coset codes over the BEWC. Although this discussion has been limited to secrecy coding over the BEWC, it is possible to apply subspace decomposition to any linear block code over a variety of wiretap channels, or even non-wiretap channels. This work is presented primarily as an initial demonstration of the usefulness of this technique. 

\appendices

\section{Proof of Lemma \ref{PhiNonnegativityLemma}}
\label{Appendix:phiNonnegativityProof}
As defined in \eqref{eqn:def_phi}, $\psi(S,n,\epsilon,q)$ is used to represent the probability of a revealed bit pattern $r(z)$ exactly spanning $S$. It is guaranteed to correctly represent this probability, however, only when the realizability constraint \eqref{eqn:QConstraintRealizable} is satisfied. Nevertheless, even when \eqref{eqn:QConstraintRealizable} is not satisfied, if \eqref{eqn:QConstraintPositive} and \eqref{eqn:QConstraintTotal} are satisfied, a stochastic system may be defined such that $\psi(S,n,\epsilon,q)$ represents the probability of a set of real outcomes of the system. This system, as we define it, consists of a set of $2^{\kappa}$ Bernoulli random variables $B_i$, $0 \leq i \leq 2^{\kappa}$ with success probability 
\begin{equation}
    \label{eqn:pnn_defBeta}
    \mathrm{Pr}[B_i = 1] = e^{n \ln(\epsilon) q_i} \mathrm{.} 
\end{equation} 
Then the probability of all $B_i=0$ for all $\nu(i)$ not lying within a subspace $S$ of $W$ is given by 
\begin{equation}
    \label{eqn:pnn_B_phi}
    \mathrm{Pr}[B_i=0 \forall i:\nu(i) \notin S] = e^{n \ln(\epsilon)(1-\zeta(S))} \mathrm{,}
\end{equation}
which is exactly equal to $\phi(S,n,\epsilon,q)$ as defined in \eqref{eqn:def_phi}. We now define a function $\psi'(S,n,\epsilon,q)$ to represent the probability that all the vectors $\nu(i)$ such that $B_i=1$ exactly span $S$. Because any set of vectors will exactly span one and only one subspace of $W$, the events represented by each $\psi'(S,n,\epsilon,q)$ are disjoint. Furthermore, if $B_i=0$ for all the $\nu(i) \notin S$, then the space exactly spanned by the $B_i=1$ must be a subspace of $S$, so 
\begin{equation}
    \label{eqn:pnn_psiSum}
    \sum_{T \subseteq S}{\psi'(T,n,\epsilon,q)} = \mathrm{Pr}[B_i\!=\!0 \;\forall\; i\!:\!\nu(i) \notin S] = \phi(S,n,\epsilon,q) \mathrm{.}
\end{equation}
We may now subtract from the sum on the left all elements except the $T=S$ element to obtain
\begin{equation}
    \label{eqn:pnn_psiRecursive1}
    \psi'(S,n,\epsilon,q) = \phi(S,n,\epsilon,q) - \sum_{T \subset S}{\psi'(T,n,\epsilon,q)} \mathrm{.}
\end{equation}
Then if $S$ has dimension $d$, the proper subspaces of $S$ have dimension less than $d$, and we may write 
\begin{equation}
    \label{eqn:pnn_psiRecursive2}
    \psi'(S^{\{d\}},n,\epsilon,q) = \phi(S,n,\epsilon,q) - \sum_{i = 0 }^{d - 1}{ \left(\!\!\!\!\!\!\!\!\sum_{\;\;\;\;\;\;T \in \Xi(S,i)}\!\!\!\!\!\!\!\! {\psi'(T,n,\epsilon,q)}\right)}\mathrm{.}
\end{equation}
This recursive relation is identical to that \eqref{eqn:def_psi} of $\psi(S,n,\epsilon,q)$, and thus $\psi(S,n,\epsilon,q) = \psi'(S,n,\epsilon,q)$. Then because $\psi'(S,n,\epsilon,q)$ represents a real probability and cannot be negative, $\psi(S,n,\epsilon,q)$ must also be nonnegative. 

\section{Proof of Lemma \ref{EtaPrimeSumLemma}}
\label{Appendix:etaPrimeProof}
We begin by noting that the $\eta'(a,b)$ satisfies the recursive property 
\begin{equation}
    \label{eqn:ep_etaPrimeRecursive1}
    \begin{split}
    \eta'(a,b) &= \binom{a}{b}_2 (-1)^{a-b} 2^{(a-b)(a-b-1)/2} \\
    &= \! \left( \! \binom{a\!-\!1}{b\!-\!1}_2 \!\!+ 2^{b} \binom{a\!-\!1}{b}_2 \right) (-1)^{a-b} 2^{(a-b)(a-b-1)/2} \\
    &= \! \binom{a\!-\!1}{b\!-\!1}_2  (-1)^{a-b} 2^{(a^2-2ab+b^2-a+b)/2} - \! \binom{a\!-\!1}{b}_2  (-1)^{(a-1)-b} 2^{(a^2-2ab+b^2-a+3b)/2} \\
    &= \eta'(a-1,b-1) - 2^{a-1} \eta'(a-1,b) \mathrm{.}
    \end{split}
\end{equation}

Next, we define a sequence, 
\begin{equation}
    \label{eqn:ep_def_f}
    f_n = \sum_{i=0}^{n}{\eta'(n,i)} \mathrm{,}
\end{equation}
to represent the sum found in \eqref{eqn:etaPrimeSum} for a given $n=a$. Using \eqref{eqn:ep_etaPrimeRecursive1}, $f_n$ may also be expressed recursively as 
\begin{equation}
    \label{eqn:ep_fRecursive1}
    f_n = \sum_{i=0}^{n}{\eta'(n-1,i-1)} - 2^{n-1} \sum_{i=0}^{n}{\eta'(n-1,i)} \mathrm{.} 
\end{equation}
Then noting that $\eta'(n-1,n)=\eta'(n-1,-1)=0$ and shifting the index of summation for the first sum by one, we obtain 
\begin{equation}
    \label{eqn:ep_fRecursive2}
    \begin{split}
    f_n &= \sum_{i=0}^{n-1}{\eta'(n-1,i)} - 2^{n-1} \sum_{i=0}^{n-1}{\eta'(n-1,i)} \\
    &= (1-2^{n-1}) \sum_{i=0}^{n-1}{\eta'(n-1,i)} \\
    &= (1-2^{n-1}) f_{n-1} \mathrm{.}  
    \end{split} 
\end{equation}
Thus for each $n>0$, $f_n$ is a multiple of $f_{n-1}$. It is easy to show that $f_0=1$, and from \eqref{eqn:ep_fRecursive2}, we can see that $f_1=0$, and therefore, $f_n=0$ for all $n>1$, proving \eqref{eqn:etaPrimeSum}. 

To prove \eqref{eqn:etaPrimePartialSum}, we define a function, 
\begin{equation}
    \label{eqn:ep_def_s}
    s(a,b) = \sum_{i=b}^{a}{\eta'(a,i)} \mathrm{,}
\end{equation}
to represent the sum found in \eqref{eqn:etaPrimePartialSum} for a given $a$ and $b$. Then using \eqref{eqn:ep_etaPrimeRecursive1}, $s(a,b)$ may be defined recursively as 
\begin{equation}
    \label{eqn:ep_sRecursive}
    \begin{split}
    s(a,b) &= \sum_{i=b}^{a}{\eta'(a-1,i-1) - 2^{a-1} \eta'(a-1,i)} \\
    &= \sum_{i=b}^{a}{\eta'(a-1,i-1)} - \sum_{i=b}^{a}{2^{a-1} \eta'(a-1,i)} \\
    &= \sum_{i=b-1}^{a-1}{\eta'(a-1,i)} - 2^{a-1} \sum_{i=b}^{a-1}{ \eta'(a-1,i)} \\
    &= s(a-1,b-1) - 2^{a-1} s(a-1,b) \mathrm{.}
    \end{split}
\end{equation}
This is an identical recursion relation to that observed for $\eta'(a,b)$ in \eqref{eqn:ep_etaPrimeRecursive1}, but $s(a,b)$ must satisfy different initial conditions- namely, $s(1,1)=1$ and $s(1,0)=0$, whereas $\eta'(1,1)=1$ and $\eta'(1,0)=-1$. To satisfy the recursion relation as well as these initial conditions, consider the function 
\begin{equation}
    \label{eqn:ep_def_sPrime}
    s'(a,b) = 2^{a-b}\eta'(a-1,b-1) \mathrm{.}
\end{equation}
Now we have $s'(1,1)=1$ and $s'(1,0)=0$, and using properties of Gaussian binomials, we may derive the recursive property 
\begin{equation}
    \label{eqn:ep_sPrimeRecursion}
    \begin{split}
    s'(a,b) &= \binom{a\!-\!1}{b\!-\!1}_{\!\!2} (-1)^{a-b}2^{(a-b)(a-b+1)/2} \\
    &= \left(\!\binom{a\!-\!2}{b\!-\!2}_{\!\!2}\!\!\! +2^{b-1}\binom{a\!-\!2}{b\!-\!1}_{\!\!2}\right) (-1)^{a-b}2^{(a^2-2ab+b^2+a-b)/2} \\
    &= \binom{a\!-\!2}{b\!-\!2}_{\!\!2} (-1)^{a-b}2^{(a^2-2ab+b^2+a-b)/2} - \binom{a\!-\!2}{b\!-\!1}_{\!\!2} (-1)^{a-b}2^{(a^2-2ab+b^2+a+b-2)/2} \\
    &= s'(a-1,b-1) - 2^{a-1}s'(a-1,b) \mathrm{.} 
    \end{split}
\end{equation}
Thus $s(a,b)$ and $s'(a,b)$ satisfy identical recursion relations with identical initial values and are therefore identical, and \eqref{eqn:etaPrimePartialSum} is true. 

Next, we consider \eqref{eqn:etaPrimeWeightedSum}. We expand $\eta'(i,b)$ and use the properties of Gaussian binomials to obtain 
\begin{equation}
    \label{eqn:ep_weighted1}
    \begin{split}
    \sum_{i=d'}^{d}{\binom{d}{i}_2 \eta'(i,d')} &= \sum_{i=d'}^{d}{\binom{d}{i}_2 \binom{i}{d'}_2 (-1)^{i-d'} 2^{(i-d')(i-d'-1)/2} }\\
    &\;\;\;\;= \binom{d}{d'}_2 \sum_{i=d'}^{d}{\binom{d-d'}{i-d'}_2 (-1)^{i-d'} 2^{(i-d')(i-d'-1)/2} }\\
    &\;\;\;\;= \binom{d}{d'}_2 \sum_{j=0}^{d-d'}{\binom{d-d'}{j}_2 (-1)^{j} 2^{(j)(j-1)/2} } \\
    &\;\;\;\;= \binom{d}{d'}_2 \sum_{j=0}^{d-d'}{\eta'(d-d',j)} \mathrm{.} 
    \end{split}
\end{equation}
Then using \eqref{eqn:etaPrimeSum}, we obtain 
\begin{equation}
    \label{eqn:ep_weighted2}
    \begin{split}
    &\sum_{i=d'}^{d}{\binom{d}{i}_2 \eta'(i,d')} = \begin{cases}
        \binom{d}{d'}_2 = 1 &\text{if } d = d'\\
        0 &\text{otherwise}
    \end{cases} \mathrm{,} 
    \end{split}
\end{equation}
proving \eqref{eqn:etaPrimeWeightedSum}. 

Now considering \eqref{eqn:KConstantsLemmaStatement}, we define a sequence,
\begin{equation}
    \label{eqn:ep_def_g}
    g_n = \sum_{j=1}^{n}{\left( j \cdot \eta'(n,j) \right)} \mathrm{,}
\end{equation}
to represent the left-hand side of \eqref{eqn:KConstantsLemmaStatement}. We then split out the last ($j=n$) term of the sum and use \eqref{eqn:etaPrimeWeightedSum} to obtain 
\begin{equation}
    \label{eqn:ep_g1}
    \begin{split}
    g_n &= n + \sum_{j=1}^{n-1}{\left( - j \sum_{i=j}^{n-1}{\left(\binom{n}{i}_2 \eta'(i,j)\right)} \right)} \\
    &= n - \sum_{i=1}^{n-1}{\left( \binom{n}{i}_2 \sum_{j=1}^{i}{\left(j \cdot \eta'(i,j)\right)} \right)} \\
    &= n - \sum_{i=1}^{n-1}{\left(\binom{n}{i}_2 g_i \right)} \mathrm{.} 
    \end{split}
\end{equation}
This recursive relation, when combined with the value $g_1=1$ for the first term of the sequence, is sufficient to fully define the $g_n$ for $n \geq 1$. Additionally, this relation may be stated more simply by combining the left-hand side of \eqref{eqn:ep_g1} with the final sum to obtain 
\begin{equation}
    \label{eqn:ep_g2}
    \begin{split}
    \sum_{i=1}^{n}{\left(\binom{n}{i}_2 g_i \right)} = n \mathrm{.} 
    \end{split}
\end{equation}

Next, we define another sequence, 
\begin{equation}
    \label{eqn:ep_def_gPrime}
    g'_n = \prod_{i=1}^{n-1}{(1-2^i)} \mathrm{,}
\end{equation}
to represent the right-hand side of \eqref{eqn:KConstantsLemmaStatement}, as well as a sequence
\begin{equation}
    \label{eqn:ep_def_h}
    h_n = \sum_{i=1}^{n}{\left(\binom{n}{i}_2 g'_i \right)} 
\end{equation}
and the sequence 
\begin{equation}
    \label{eqn:ep_def_hPrime}
    h'_n = h_n - h_{n-1} 
\end{equation}
to represent the difference between consecutive terms of $h_n$. 

Our approach at this point is to find a recursive definition of $h'_n$ in terms of $h'_{n-1}$. To do this, we use \eqref{eqn:ep_def_h} to express the difference between two consecutive terms, then substitute \eqref{eqn:ep_g2} and use the definition \eqref{eqn:GaussianBinomial} of the Gaussian binomial to yield 
\begin{equation}
    \label{eqn:ep_hPrime1}
    \begin{split}
    h'_n = h_n - h_{n-1} &= \sum_{i=1}^{n}{\left(\binom{n}{i}_2 g'_i \right)} - \sum_{i=1}^{n-1}{\left(\binom{n-1}{i}_2 g'_i \right)} \\
    &= \sum_{i=1}^{n}{\left(\binom{n}{i}_2 \prod_{j=1}^{i-1}{(1-2^j)} \right)} - \sum_{i=1}^{n-1}{\left(\binom{n-1}{i}_2 \prod_{j=1}^{i-1}{(1-2^j)} \right)} \\
    &= \sum_{i=1}^{n}{\left(2^{n-i} \binom{n-1}{i-1}_2 \prod_{j=1}^{i-1}{(1-2^j)} \right)} \\
    &= \sum_{i=1}^{n}{\left(2^{n-i} \prod_{j=1}^{i-1}{(1-2^{n-j})} \right)} \mathrm{.} 
    \end{split}
\end{equation}
With this expression, we can split out the $i=1$ term from the sum and transform the remaining terms to obtain 
\begin{equation}
    \label{eqn:ep_hPrime2}
    \begin{split}
    h'_n &= 2^{n-1} + \sum_{i=2}^{n}{\left(2^{n-i} \prod_{j=1}^{i-1}{(1-2^{n-j})} \right)} \\
    &= 2^{n-1} + \sum_{i=2}^{n}{\left(2^{n-i} (1-2^{n-1}) \prod_{j=2}^{i-1}{(1-2^{n-j})} \right)} \mathrm{.} 
    \end{split}
\end{equation}
Then substituting $i' = i-1$ and $j' = j-1$ gives 
\begin{equation}
    \label{eqn:ep_hPrime2}
    \begin{split}
    h'_n &= 2^{n-1} + \sum_{i'=1}^{n-1}{\left(2^{n-1-i'} (1-2^{n-1}) \prod_{j'=1}^{i'-1}{(1-2^{n-1-j'})} \right)} \\
    &= 2^{n-1} + (1-2^{n-1}) h'_{n-1} \mathrm{.} 
    \end{split}
\end{equation}
From this recursive relation, it is clear that if $h'_{n-1}=1$, then $h'_n=1$ as well. Then because $h'_1=1$, we may conclude that every $h'_n=1$. Because $h'_n$ represents the difference between consecutive $h_n$, and because $h_n=1$, we may then conclude that $h_n=n$, and from \eqref{eqn:ep_def_h}, we find that 
\begin{equation}
    \label{eqn:ep_h1}
    \sum_{i=1}^{n}{\left(\binom{n}{i}_2 g'_i\right)} = n \mathrm{,} 
\end{equation}
a relation identical to that \eqref{eqn:ep_g2} satisfied by $g_n$. Thus because $g_n$ and $g'_n$ satisfy identical recursion relations and have identical first term $g_1 = g'_1 = 1$, $g_n$ and $g'_n$ must be identical, and \eqref{eqn:KConstantsLemmaStatement} is proven. 

\section{Proof of Lemma \ref{thm:exponentialMersenneLemma}}
\label{apx:exponentialMersenneLemmaProof}
Begin by dividing out the top ($j=n$) term from the product in \eqref{eqn:MersenneExponentialProperty} and defining the result as a sequence $a_n$ with terms given by 
\begin{equation}
    \label{eqn:MersenneProof1}
    \begin{split}
    a_n &= \frac{1}{(1-2^n)} \sum_{i=0}^{n}{\left( 2^{bi} \prod_{j=i}^{n}{(1-2^j)} \right)} \\
    &= \sum_{i=1}^{n}{\left( (2^b)^i \prod_{j=i}^{n-1}{(1-2^j)}\right)} \mathrm{.}
    \end{split}
\end{equation}
Using \eqref{eqn:KConstantsLemmaStatement}, \eqref{eqn:MersenneProof1} becomes 
\begin{equation}
    \label{eqn:MersenneProof2}
    a_n \!=\! \sum_{i=1}^{n}{\!\left( \!(2^b)^i \!\sum_{j=1}^{n}{(-1)^{j-i} 2^{\frac{(j-i)(j-i+1)}{2}} \binom{j\!-\!1}{i\!-\!1}_{\!2} \binom{n\!-\!1}{j\!-\!1}_{\!2}} \right)} \mathrm{.}
\end{equation}
Rearranging terms and exchanging the sums yields 
\begin{equation}
    \label{eqn:MersenneProof3}
    a_n \!=\! \sum_{j=1}^{n} {\!\left(\! \binom{n\!-\!1}{j\!-\!1}_{\!2} \sum_{i=1}^{n} (2^b)^i {(-1)^{j-i} 2^{\frac{(j-i)(j-i+1)}{2}} \binom{j\!-\!1}{i\!-\!1}_{\!2} } \right)} \mathrm{.}
\end{equation}

Next, consider the following series of polynomials in $x$: 
\begin{equation}
    \label{eqn:MersenneProof4}
    p_j(x) = x(x-2)(x-4)(x-8) \dots (x-2^{j-1}) \mathrm{.} 
\end{equation}
Fully expanded, this polynomial may be represented in terms of the coefficients $c_{j,i}$ as 
\begin{equation}
    \label{eqn:MersenneProof5}
    p_j(x) = \sum_{i=1}^{n}{c_{j,i}x^i} \mathrm{.} 
\end{equation}
The values of the $c_{j,i}$ are not immediately apparent, but a simple recursion relation may be identified. Based on \eqref{eqn:MersenneProof4}, the $c_{j,i}$ are fully determined by the recursive relation 
\begin{equation}
    \label{eqn:MersenneProof6}
    c_{j,i} = 
    \begin{cases}
        0 & \text{ if } i \leq 0 \text{ or } j < i \\
        1 & \text{ if } j = i = 1 \\
        c_{j-1,i-1} - 2^{j-1}c_{j-1,i} &\text{ otherwise} \mathrm{.}
    \end{cases}  
\end{equation}
Next, define $d_{j,i}$ to match the coefficients of $(2^b)^i$ in \eqref{eqn:MersenneProof3} as
\begin{equation}
    \label{eqn:MersenneProof7}
    d_{j,i} = (-1)^{j-i} 2^{\frac{(j-i)(j-i+1)}{2}} \binom{j-1}{i-1}_2   \mathrm{.}
\end{equation}
Observe that if $i\leq 0$ or $j<i$, $d_{j,i}$ equals zero, and that if $i=j=1$, $d_{j,i}$ equals one. In all other cases, $d_{j,i}$ may be expressed in terms of other $d_{\cdot,\cdot}$ as  
\begin{equation}
    \label{eqn:MersenneProof8}
    \begin{split}
    d_{j,i} &= (-1)^{j-i} 2^{\frac{(j-i)(j-i+1)}{2}} \binom{j-1}{i-1}_2 \\
    &= (-1)^{j-i} 2^{\frac{(j-i)(j-i+1)}{2}} \left( 2^{i-1} \binom{j-2}{i-1}_2 + \binom{j-2}{i-2}_2 \right) \\
    &= (-1)^{j-i} 2^{\frac{((j-1)-(i-1))((j-1)-(i-1)+1)}{2}} \binom{(j-1)-1}{(i-1)-1}_2 - (-1)^{j-i-1} 2^{j-1} 2^{\frac{((j-1)-i)((j-1)-i+1)}{2}} \binom{(j-1)-1}{i-1}_2 \\
    &= d_{j-1,i-1} - 2^{j-1}d_{j-1,i} \mathrm{.}
    \end{split}
\end{equation}
The $d_{j,i}$ may thus be alternatively defined by the recursion relation 
\begin{equation}
    \label{eqn:MersenneProof8_5}
    d_{j,i} = 
    \begin{cases}
        0 & \text{ if } i \leq 0 \text{ or } j < i \\
        1 & \text{ if } j = i = 1 \\
        d_{j-1,i-1} - 2^{j-1}d_{j-1,i} &\text{ otherwise} \mathrm{.}
    \end{cases}  
\end{equation}
Because they are fully defined by identical recursion relations, $c_{j,i}$ and $d_{j,i}$ must be identical. Thus, the second sum in \eqref{eqn:MersenneProof3} may be replaced by the polynomial defined in \eqref{eqn:MersenneProof4} using $2^b$ as the polynomial variable to obtain 
\begin{equation}
    \label{eqn:MersenneProof9}
    a_n = \sum_{j=1}^{n} {\left( \binom{n-1}{j-1}_2 2^b \cdot \prod_{i=1}^{j-1}{(2^b-2^i)} \right)} \mathrm{.}
\end{equation}
Here, the product in \eqref{eqn:MersenneProof9} will either be positive if $b > {j-1}$, or if $b \leq {j-1}$, the product will contain a zero factor and will equal zero. Then each term in the sum in \eqref{eqn:MersenneProof9} consists of all positive or zero factors. Furthermore, at least the first term in the sum is positive, and therefore $a_n$ must be positive. Then because $(1-2^n)$ is always negative for $n>0$, by the definition \eqref{eqn:MersenneProof1} of $a_n$, the expression in \eqref{eqn:MersenneExponentialProperty} must be negative. 

\section{Proof of Lemma \ref{SuperexponentialLemma}}
\label{Appendix:SuperexponentialProof}
We first observe that the term before the product in \eqref{eqn:SuperexponentialProperty} may be rewritten using power series expansion as 
\begin{equation}
    \label{eqn:superExponsntialSeries}
    2^i\beta^{2^i} = 2^i e^{\ln(\beta) 2^i} = \sum_{b=0}^{\infty}{2^i \frac{\ln(\beta)^b 2^{b i}}{b!}} = \sum_{b=1}^{\infty}{\frac{\ln(\beta)^{b-1} }{(b-1)!} 2^{bi}} \mathrm{.} 
\end{equation}
Then substituting this term back into \eqref{eqn:SuperexponentialProperty} gives 
\begin{equation}
    \label{eqn:SuperexponentialProof2}
    \sum_{i=0}^{n}{\left( \left(\sum_{b=1}^{\infty}{\frac{\ln(\beta)^{b-1} }{(b-1)!} 2^{bi}}\right) \prod_{j=i}^{n}{(1-2^j)} \right)} < 0 \mathrm{.}
\end{equation}
The sums may now be rearranged to yield 
\begin{equation}
    \label{eqn:SuperexponentialProof2}
     \sum_{b=1}^{\infty}{\left(\frac{\ln(\beta)^{b-1} }{(b-1)!} \sum_{i=0}^{n}{\left( 2^{bi} \prod_{j=i}^{n}{(1-2^j)} \right)}\right)} < 0 \mathrm{.}
\end{equation}
Now the term in front of the second sum is always positive, and by Lemma \ref{thm:exponentialMersenneLemma}, the second sum is always negative. Thus the final result is an infinite sum of negative terms and is negative. 

\section{Proof of Theorem \ref{thm:zeroColumn}}
\label{Appendix:zcDerivative}
We prove this theorem by showing that for any $q$ satisfying \eqref{eqn:QConstraintPositive} and \eqref{eqn:QConstraintTotal} with $q_0 > 0$ (except the trivial case $q_0 = 1$, which is clearly not locally optimal), we may define a movement vector $\grave{q}$ with $\grave{q}_0 < 0$ such that $l(n,\epsilon,q)$ decreases as $q$ moves in the $\grave{q}$ direction. This $\grave{q}$ is defined as 
\begin{equation}
    \label{eqn:zc_defQGrave}
    \grave{q}_i = \begin{cases}
        q_i-1 & \text{if } i = 0 \\
        q_i & \text{otherwise} \mathrm{.} 
    \end{cases}    
\end{equation}
We now define $\grave{l}(n,\epsilon,q)$ to represent the rate of change of $l(n,\epsilon,q)$ as $q$ moves in the $\grave{q}$ direction. More precisely, 
\begin{equation}
    \label{eqn:zc_defLGrave}
    \grave{l}(n,\epsilon,q) = \frac{\partial l(n,\epsilon,q+x\grave{q})}{\partial x}\mathrm{.}
\end{equation}
To prove that $\grave{l}(n,\epsilon,q)$ is always negative, we further define a function $\grave{\psi}(S)$ to represent the rate of change of $\psi(S)$ with movement in the $\grave{q}$ direction, 
\begin{equation}
    \label{eqn:zc_defPsiGrave}
    \grave{\psi}(S) = \grave{\psi}(S,n,\epsilon,q) = \frac{\partial \psi(S,n,\epsilon,q+x\grave{q})}{\partial x} \mathrm{.} 
\end{equation}
The value of $\grave{\psi}(S)$ may be calculated by finding the rate of change $\grave{\phi}(S)$ of $\phi(S)$, which in turn depends on the rate of change $\grave{\zeta}(S)$ of $\zeta(S)$. Because every subspace $S$ includes the zero vector, $\grave{\zeta}(S)$ is simply given by 
\begin{equation}
    \label{eqn:zc_defZetaGrave}
    \grave{\zeta}(S) = \grave{\zeta}(S,q) = \sum_{i : \nu(i) \in S}{\grave{q}_i} = \zeta(S)-1 \mathrm{,} 
\end{equation}
and $\grave{\psi}(S)$ is given by 
\begin{equation}
    \label{eqn:zc_psiGrave1}
    \begin{split}
    \grave{\psi}(S^{\{d\}}) = \grave{\psi}(S^{\{d\}},n,\epsilon,q) = \grave{\phi}(S) - \sum_{i=0}^{d-1}{\left( \!\!\!\!\!\! \sum_{\;\;\;\;T \in \Xi(S,i)} {\!\!\!\! \grave{\psi}(T)} \right)} \mathrm{,}
    \end{split}
\end{equation}
with 
\begin{equation}
    \label{eqn:zc_defPhiGrave}
    \begin{split}
    \grave{\phi}(S) = \grave{\phi}(S,n,\epsilon,q) &= n \ln(\epsilon) (1-\zeta(S))e^{n \ln(\epsilon) (1-\zeta(S))} \\
    &= n \ln(\epsilon) (1-\zeta(S)) \phi(S) \mathrm{.}
    \end{split}
\end{equation}

Next, we define two functions, $\omega(d)$ and $\grave{\omega}(d)$. The $\omega(d)$ function represents the probability, for a given dimension $d$, of the rank of $G_{r(z)}$ being at most $d$, and the $\grave{\omega}(d)$ function representing the rate of change of $\omega(\cdot)$ with movement in the $\grave{q}$ direction. Recalling that the columns of $r(z)$ exactly span only one subspace $S$ of $W$, that the probability of $r(z)$ exactly spanning $S$ is given by $\psi(S)$, and that the rank of $G_{r(z)}$ is equal to the dimension of $S$, we may express $\omega(d)$ and $\grave{\omega}(d)$ as 
\begin{equation}
    \label{eqn:zc_defOmega}
    \omega(d) = \omega(d,n,\epsilon,q) = \sum_{i=0}^{d}{\left( \sum_{S \in \Xi(W,i)} {\psi(S)}  \right)}
\end{equation}
and 
\begin{equation}
    \label{eqn:zc_defOmegaGrave}
    \grave{\omega}(d) = \grave{\omega}(d,n,\epsilon,q) = \sum_{i=0}^{d}{\left( \sum_{S \in \Xi(W,i)} {\grave{\psi}(S)}  \right)} \mathrm{.} 
\end{equation}

From the relationship between $l(n,\epsilon,q)$ and $\psi(S,n,\epsilon,q)$ expressed in \eqref{eqn:equivocationLossEpsilonPfister} and \eqref{eqn:equivocationLossEpsilon1}, it is clear that the expected equivocation loss may be expressed in terms of the $\omega(d)$ function as 
\begin{equation}
    \label{eqn:zc_LossInTermsOfOmega}
    l(n,\epsilon,q) = n(1-\epsilon) + \sum_{d=0}^{\kappa-1}{\omega(d)} \mathrm{,}
\end{equation}
and 
\begin{equation}
    \label{eqn:zc_LossRateInTermsOfOmegaGrave}
    \grave{l}(n,\epsilon,q) = \sum_{d=0}^{\kappa-1}{\grave{\omega}(d)} \mathrm{.}
\end{equation}
Thus if it can be shown that each $\grave{\omega}(d) \leq 0$ for every $d$ and that $\grave{\omega}(d) < 0$ for at least one $d$, then $\grave{l}(n,\epsilon,q)$ is negative, and $q$ is not optimal. 

It is easy to show that $\grave{\omega}(d)$ is always negative for $d=0$, as \eqref{eqn:zc_defOmegaGrave} reduces to 
\begin{equation}
    \label{eqn:zc_LossRateInTermsOfOmegaGrave}
    \grave{\omega}(0,n,\epsilon,q) = \grave{\psi}(\{0\},n,\epsilon,q) = n \ln(\epsilon)(1-q_0) e^{n \ln(\epsilon)(1-q_0)} \mathrm{,} 
\end{equation}
in which all terms are positive except $\ln(\epsilon)$, which is always negative. 

It remains to show that $\grave{\omega}(d)$ is always negative or zero. To accomplish this, we use a method similar to that used to prove Theorem \ref{thm:ExpectedEquivocationFormula}- that is, we examine the recursive expansion \eqref{eqn:zc_psiGrave1} of the $\grave{\psi}(S)$ function in \eqref{eqn:zc_defOmegaGrave}. As in \eqref{eqn:PsiExpanded1}, the final expanded expression for $\grave{\psi}(S)$, and also therefore of $\grave{\omega}(d)$, is a sum of $\grave{\phi}(T)$ terms. Again, the $\grave{\phi}(T^{\{d\}})$ terms occur in equal frequency for each subspace $T$ of dimension $d$. The final expansion of $\grave{\omega}(d)$ may then be expressed using the $\eta(\cdot)$ function as 
\begin{equation}
    \label{eqn:zc_omegaEta1}
    \begin{split}
    \grave{\omega}(d) &= \sum_{d'=0}^{d}{\left( \!\!\left( \sum_{i=d'}^{d}{\binom{\kappa}{i}_{\!\!2} \eta(i,d') / \binom{\kappa}{d'}_{\!\!2}} \right) \left( \!\!\!\!\!\!\!\! \sum_{\;\;\;\;\;\;S \in \Xi(W,d')} {\!\!\!\!\!\!\!\! \grave{\phi}(S)}  \right) \!\!\right)} \\
    &= \sum_{d'=0}^{d}{\left( \!\!\left( \sum_{i=d'}^{d}{\binom{\kappa-d'}{i-d'}_{\!\!2} \eta(i-d',0)} \right) \left( \!\!\!\!\!\!\!\! \sum_{\;\;\;\;\;\;S \in \Xi(W,d')} {\!\!\!\!\!\!\!\! \grave{\phi}(S)}  \right) \!\!\right)} \\
    &= \sum_{d'=0}^{d}{\left( \!\!\left( \sum_{i=\kappa-d}^{\kappa-d'}{\!\!\binom{\kappa-d'}{\kappa\!-\!d'\!-\!i}_{\!\!2} \eta(\kappa\!-\!d'\!-\!i,0)}\! \right) \!\! \left( \!\!\!\!\!\!\!\! \sum_{\;\;\;\;\;\;S \in \Xi(W,d')} {\!\!\!\!\!\!\!\! \grave{\phi}(S)} \! \right) \!\!\right)} \\
    &= \sum_{d'=0}^{d}{\left( \!\!\left( \sum_{i=\kappa-d}^{\kappa-d'}{\eta(\kappa-d',i)} \right) \left( \!\!\!\!\!\!\!\! \sum_{\;\;\;\;\;\;S \in \Xi(W,d')} {\!\!\!\!\!\!\!\! \grave{\phi}(S)}  \right) \!\!\right)} \mathrm{.}
    \end{split}
\end{equation}
Then using \eqref{eqn:etaPrimePartialSum} from Lemma \ref{EtaPrimeSumLemma} and expanding $\grave{\phi}(S)$ using \eqref{eqn:zc_defPhiGrave}, this becomes 
\begin{equation}
    \label{eqn:zc_omegaEta2}
    \begin{split}
    \grave{\omega}(d) &= \sum_{d'=0}^{d}{\left(  2^{d-d'} \eta(\kappa\!-\!d'\!-\!1,\kappa\!-\!d\!-\!1) \left( \!\!\!\!\!\!\!\! \sum_{\;\;\;\;\;\;S \in \Xi(W,d')} {\!\!\!\!\!\!\!\! \grave{\phi}(S)}  \right) \right)} \\
    &= \sum_{d'=0}^{d}{\left( 2^{d-d'} \binom{\kappa\!-\!d'\!-\!1}{\kappa\!-\!d\!-\!1}_2 c(d,d') \left( \!\!\!\!\!\!\!\! \sum_{\;\;\;\;\;\;S \in \Xi(W,d')} {\!\!\!\!\!\!\!\! \grave{\phi}(S)}  \right) \right)} \\
    &= \! n \ln(\epsilon) \!\! \sum_{d'=0}^{d}{\!\left( \! c(d,d')  \!  \! \!\!\!\!\!\!\!\!\!\!\!\! \sum_{\;\;\;\;\;\;\;\;\;S \in \Xi(W,d')} {\!\!\!\!\!\!\!\!\!\!\!\! 2^{d\!-\!d'} \! \binom{\kappa\!-\!d'\!-\!1}{\kappa\!-\!d\!-\!1}_{\!2}   (1 \! - \! \zeta(S)) \phi(S)}  \!\! \right)} \! \mathrm{,}
    \end{split}
\end{equation}
where $c(\cdot)$ is as defined in \eqref{eqn:def_c_constants}. 
Using the properties of Gaussian binomials, the inner sum of \eqref{eqn:zc_omegaEta2} may be expanded to yield 
\begin{equation}
    \label{eqn:zc_omegaEtaInner1}
    \begin{split}
    &\sum_{S \in \Xi(W,d')} { \!\!\left( 2^{d\!-\!d'} \! \binom{\kappa\!-\!d'\!-\!1}{\kappa\!-\!d\!-\!1}_{\!2}   (1 \! - \! \zeta(S)) \phi(S) \!\right)} = \!\!\!\!\sum_{S \in \Xi(W,d')} { \!\!\left( \! \binom{\kappa\!-\!d'}{\kappa\!-\!d\!}_{\!2}   (1 \! - \! \zeta(S)) \phi(S) \!\right)} - \!\!\!\!\!\sum_{S \in \Xi(W,d')} { \!\!\left( \! \binom{\kappa\!-\!d'\!-\!1}{\kappa\!-\!d\!}_{\!2}   (1 \! - \! \zeta(S)) \phi(S) \!\right)} \mathrm{.}
    \end{split}
\end{equation}
Next, we note that 
\begin{enumerate*}
    \item the number of $d$-dimensional superspaces of a $d'$-dimensional space $S$ is equal to $\binom{\kappa-d'}{\kappa-d}_2$; and 
    \item for a given $d'$-dimensional space $S$, summing the quantity $\zeta(T)-\zeta(S)$ over all $d$-dimensional superspaces $T$ of $S$ yields $\binom{\kappa-d'-1}{\kappa-d}_2 (1-\zeta(S))$. 
\end{enumerate*}
Using these identities, \eqref{eqn:zc_omegaEtaInner1} becomes 
\begin{equation}
    \label{eqn:zc_omegaEtaInner2}
    \begin{split}
    \sum_{S \in \Xi(W,d')} { \!\!\left( 2^{d\!-\!d'} \! \binom{\kappa\!-\!d'\!-\!1}{\kappa\!-\!d\!-\!1}_{\!2}   (1 \! - \! \zeta(S)) \phi(S) \right)} &= \!\!\!\! \sum_{S \in \Xi(W,d')} {\!\! \left( \!\!\!\!\!\!\!\!\!\!\!\! \!\!\!\!\sum_{\;\;\;\;\;\;\;\;\;\;\;\;T^{\{d\}}:S \in \Xi(T,d')} { \!\!\!\!\!\!\!\!\!\!\!\!\!\!\!\!\!\! \left( (1 \! - \! \zeta(S)) \phi(S) \right)} \right) } - \!\!\!\! \sum_{S \in \Xi(W,d')} {\!\! \left( \!\!\!\!\!\!\!\!\!\!\!\!\!\!\!\! \sum_{\;\;\;\;\;\;\;\;\;\;\;\; T^{\{d\}}:S \in \Xi(T,d')} {\!\!\!\!\!\!\!\!\!\!\!\!\!\!\!\!\!\!  (\zeta(T) \! - \! \zeta(S)) \phi(S) } \right) } \\
    & = \!\!\!\! \sum_{S \in \Xi(W,d')} {\!\! \left( \!\!\!\!\!\!\!\!\!\!\!\!\!\!\!\! \sum_{\;\;\;\;\;\;\;\;\;\;\;\; T^{\{d\}}:S \in \Xi(T,d')} { \!\!\!\!\!\!\!\!\!\!\!\!\!\!\!\!\!\! \left( (1 \! - \! \zeta(T)) \phi(S) \right)} \right) } \\
    & = \!\!\!\! \sum_{T \in \Xi(W,d)} {\!\left( (1 \! - \! \zeta(T)) \!\!\!\! \sum_{S \in \Xi(T,d')} { \!\!\!\!\left(  \phi(S) \right)} \right) } \mathrm{.}
    \end{split}
\end{equation}
Substituting this back into \eqref{eqn:zc_omegaEta2} gives 
\begin{equation}
    \label{eqn:zc_omegaEta3}
    \begin{split}
    \grave{\omega}(d) &= \! n \ln(\epsilon) \!\! \sum_{d'=0}^{d}{\!\left( \! c(d,d')  \!\!\!\! \sum_{T \in \Xi(W,d)} {\left( (1 \! - \! \zeta(T)) \!\!\!\! \sum_{S \in \Xi(T,d')} {\!\!\!\! \left(  \phi(S) \right)} \right) }  \!\! \right)} \\
    &= \! n \ln(\epsilon) \!\! \sum_{T \in \Xi(W,d)} {\left( (1 \! - \! \zeta(T)) \sum_{d'=0}^{d}{\!\left( \! c(d,d')  \!\!\!\!  \!\!\!\! \sum_{S \in \Xi(T,d')} {\!\!\!\! \left( \phi(S) \right)} \right) }  \!\! \right)}  \mathrm{.}
    \end{split}
\end{equation}
Then recalling the $\psi(\cdot)$ analog of the expression found in \eqref{eqn:PsiExpanded1}, this becomes simply 
\begin{equation}
    \label{zc_omegaFinal}
    \grave{\omega}(d) = n \ln(\epsilon) \sum_{T \in \Xi(W,d)} {\left( (1 \! - \! \zeta(T)) \psi(T)  \right)} \mathrm{.}
\end{equation}
In this expression, $n$ is always positive, $\ln(\epsilon)$ is always zero or negative, $(1-\zeta(T))$ is zero or positive, and by Lemma \ref{PhiNonnegativityLemma}, $\psi(T)$ is zero or positive. Thus, $\grave{\omega}(d)$ is always zero or negative, completing the proof.

\section{Proof of Theorem \ref{thm:uvf}}
\label{Appendix:UVF_Proof}
The outline for this proof is as follows. We first compute the gradient $\nabla l(\qbar)$ and the Hessian matrix $\mathbf{H}(l)$ of the expected equivocation loss $l(\qbar) = l(n,\epsilon,\qbar)$ We next define the requirements of a unit movement vector $\dot{\qbar}$ which is compliant with \eqref{eqn:QConstraintPositive} and \eqref{eqn:QConstraintTotal}. Using these requirements, we show that such a $\dot{\qbar}$ multiplied by the $\nabla l(\qbar)$ at the point $\qbar=\bar{\qbar}$ is always zero. Finally, we show that the quadratic form $\dot{\qbar}^{\intercal}\mathbf{H}(l)\dot{\qbar}$ at $\qbar=\bar{\qbar}$ is always positive. Thus, at $\qbar=\bar{\qbar}$, movement in any direction compliant with \eqref{eqn:QConstraintPositive} and \eqref{eqn:QConstraintTotal} results in an increase of equivocation loss. 

We begin by computing $\nabla l(\qbar)$ from the expression \eqref{eqn:expectedEquivocationEpsilon} proven in Theorem \ref{thm:ExpectedEquivocationFormula}. From this expression, each element $\nabla l(\qbar)_i$ of $\nabla l(\qbar)$ is given by 
\begin{equation}
    \label{eqn:ufc_defNabla}
    \begin{split}
    \nabla l(\qbar)_{i} = \frac{\partial l(\qbar)}{\partial \qbar_i} &= \frac{\partial}{\partial \qbar_i} \sum_{\delta=1}^{\kappa}{\left( K_{\delta} \!\!\!\!\!\!\!\! \sum_{\;\;\;\; S \in \Xi(W,\kappa-\delta)}{\!\!\!\!\!\!\!\! \phi(S)} \right)} \\
    &=  \sum_{\delta=1}^{\kappa}{\left( K_{\delta} \!\!\!\!\!\!\!\! \sum_{\;\;\;\; S \in \Xi(W,\kappa-\delta)}{\!\!\!\! \frac{\partial \phi(S)}{\partial \qbar_i}} \right)} \mathrm{.} 
    \end{split}
\end{equation}
The derivative of $\phi(S)$ may be calculated using its definition \eqref{eqn:def_phi} and the definition \eqref{eqn:defZeta} of $\zeta(S)$ as 
\begin{equation}
    \label{eqn:ufc_deriv_phi}
    \begin{split}
    \frac{\partial \phi(S)}{\partial \qbar_i} &= -n \ln(\epsilon) e^{n \ln(\epsilon)(1-\zeta(S))} \frac{\partial \zeta(S)}{\partial \qbar_i} \\
    &= \begin{cases}
        -n \ln(\epsilon) e^{n \ln(\epsilon)(1-\zeta(S))} & \text{ if } \nu(i) \in S \\
        0 & \text{ otherwise} \mathrm{.} 
    \end{cases}
    \end{split}
\end{equation}
The second sum in \eqref{eqn:ufc_defNabla} may thus be taken over only those subspaces $S$ of which $\nu(i)$ is an element. The gradient is then given by 
\begin{equation}
    \label{eqn:ufc_nabla1}
    \begin{split}
    \nabla l(\qbar)_{i} &= -n \ln(\epsilon) \sum_{\delta=1}^{\kappa}{\left( K_{\delta} \!\!\!\!\!\!\!\!\!\!\!\!\!\!\!\!\!\!\!\! \sum_{\;\;\;\;\;\;\;\;\;\;\;\; S: S \in \Xi(W,\kappa-\delta), \nu(i) \in S}{\!\!\!\!\!\!\!\!\!\!\!\!\!\!\!\!\!\!\!\! e^{n \ln(\epsilon)(1-\zeta(S))}} \right)} \mathrm{.} 
    \end{split}
\end{equation}
At the uniform vector fraction vector $\qbar = \bar{\qbar}$, the value of $\zeta(S^{\{d\}})$ for any subspace $S$ of dimension $d$ is given by 
\begin{equation}
    \label{eqn:ufc_zeta}
    \zeta(S^{\{d\}}) = \sum_{i:\nu(i) \in S}{\qbar_i} = \frac{2^d-1}{2^\kappa-1} \mathrm{,}
\end{equation}
and the number of subspaces $S$ of dimension $d$ which contain any nonzero vector $\nu(i)$ is
\begin{equation}
    \label{eqn:ufc_nablaSuperSpaceCount}
    \left| \{ S: S \in \Xi(W,d), \nu(i) \in S \} \right| = \binom{\kappa-1}{d-1}_2 = \binom{\kappa-1}{\kappa-d}_2 \mathrm{.} 
\end{equation}
Substituting \eqref{eqn:ufc_zeta} and \eqref{eqn:ufc_nablaSuperSpaceCount} into \eqref{eqn:ufc_nabla1}, we find that the gradient at $\qbar=\bar{\qbar}$ is 
\begin{equation}
    \label{eqn:ufc_nablaCenter}
    \left. \nabla l(\qbar)_{i} \right|_{\qbar=\bar{\qbar}} = -n \ln(\epsilon) \sum_{\delta=1}^{\kappa}{\left( K_{\delta} \binom{\kappa-1}{\delta}_2 e^{n \ln(\epsilon)\frac{2^\kappa-2^{\kappa-\delta}}{2^\kappa-1}} \right)} \mathrm{.} 
\end{equation}

The elements of the Hessian matrix are then given by 
\begin{equation}
    \label{eqn:ufc_defHessian}
    \begin{split}
    \mathbf{H}(l)_{i,j} = \frac{\partial \nabla l(\qbar)_{i}}{\partial \qbar_j} &= \frac{\partial \zeta(S)}{\partial \qbar_j} n^2 \ln(\epsilon)^2 \sum_{\delta=1}^{\kappa}{\left( K_{\delta} \!\!\!\!\!\!\!\!\!\!\!\!\!\!\!\!\!\!\!\! \sum_{\;\;\;\;\;\;\;\;\;\;\;\; S: S \in \Xi(W,\kappa-\delta), \nu(i) \in S}{\!\!\!\!\!\!\!\!\!\!\!\!\!\!\!\!\!\!\!\! e^{n \ln(\epsilon)(1-\zeta(S))}} \right)} \\
    &= n^2 \ln(\epsilon)^2 \sum_{\delta=1}^{\kappa}{\left( K_{\delta} \!\!\!\!\!\!\!\!\!\!\!\!\!\!\!\!\!\!\!\!\!\!\!\!\!\!\!\! \sum_{\;\;\;\;\;\;\;\;\;\;\;\;\;\;\;\; S: S \in \Xi(W,\kappa-\delta), \nu(i) \in S, \nu(j) \in S}{\!\!\!\!\!\!\!\!\!\!\!\!\!\!\!\!\!\!\!\!\!\!\!\!\!\!\!\! e^{n \ln(\epsilon)(1-\zeta(S))}} \right)} \mathrm{.} 
    \end{split}
\end{equation}
To find the values of $\mathbf{H}(l)$ at $\qbar = \bar{\qbar}$, we use \eqref{eqn:ufc_nablaSuperSpaceCount} and note that the number of subspaces $S$ of dimension d which contain both $\nu(i)$ and $\nu(j)$ is given by 
\begin{equation}
    \label{eqn:ufc_nablaSuperSpaceCount2}
    \begin{split}
    &\left| \{ S: S \in \Xi(W,d), \nu(i) \in S, \nu(j) \in S \} \right| = \begin{cases}
        \binom{\kappa-1}{d-1}_2 = \binom{\kappa-1}{\kappa-d}_2 & \text{if } i = j \\
        \binom{\kappa-2}{d-2}_2 = \binom{\kappa-2}{\kappa-d}_2 & \text{if } i \neq j \mathrm{.} 
    \end{cases}  
    \end{split}
\end{equation}
Using these equations, elements of $\mathbf{H}(l)$ may be expressed as 
\begin{equation}
    \label{eqn:ufc_Hessian1}
    \begin{split}
    &\mathbf{H}(l)_{i,j} = \begin{cases}
        n^2 \ln(\epsilon)^2 \sum_{\delta=1}^{\kappa}{\left( K_{\delta} \binom{\kappa-1}{\delta}_2 e^{n \ln(\epsilon)\frac{2^\kappa-2^{\kappa-\delta}}{2^\kappa-1}} \right)} &\text{\!\!\!\!if } i=j \\
        n^2 \ln(\epsilon)^2 \sum_{\delta=1}^{\kappa}{\left( K_{\delta} \binom{\kappa-2}{\delta}_2 e^{n \ln(\epsilon)\frac{2^\kappa-2^{\kappa-\delta}}{2^\kappa-1}} \right)} &\text{\!\!\!\!if } i \neq j \mathrm{,}
    \end{cases}
    \end{split}
\end{equation}
and hence $\mathbf{H}(l)$ may be expressed as a sum of a diagonal matrix and a uniform matrix as 
\begin{equation}
    \label{eqn:ufc_Hessian2}
    \begin{split}
    \mathbf{H}(l) &= \! n^2 \! \ln(\epsilon)^2 \! \sum_{\delta=1}^{\kappa}{\left( \!\! K_{\delta} \! \left( \! \binom{\kappa \! - \! 1}{\delta}_{\!\!2} \!\!\! - \!\! \binom{\kappa \! - \! 2}{\delta}_{\!\!2} \right) \! e^{n \ln(\epsilon)\frac{2^\kappa-2^{\kappa-\delta}}{2^\kappa-1}} \right)}  \mathbf{I} + n^2 \ln(\epsilon)^2 \sum_{\delta=1}^{\kappa}{\left( K_{\delta} \binom{\kappa \! - \! 2}{\delta}_{\!\!2} e^{n \ln(\epsilon)\frac{2^\kappa-2^{\kappa-\delta}}{2^\kappa-1}} \right)} \mathbf{J} \\
    &= \! n^2 \! \ln(\epsilon)^2  \sum_{\delta=1}^{\kappa}{\left( \! K_{\delta}   2^{\kappa - \delta - 1} \binom{\kappa \! - \! 2}{\delta-1}_{\!\!2} e^{n \ln(\epsilon)\frac{2^\kappa-2^{\kappa-\delta}}{2^\kappa-1}} \right)}  \mathbf{I} + n^2 \ln(\epsilon)^2 \sum_{\delta=1}^{\kappa}{\left( K_{\delta} \binom{\kappa \! - \! 2}{\delta}_{\!\!2} e^{n \ln(\epsilon)\frac{2^\kappa-2^{\kappa-\delta}}{2^\kappa-1}} \right)} \mathbf{J} \mathrm{.}
    \end{split}
\end{equation}
where $\mathbf{I}$ is the identity matrix and $\mathbf{J}$ is the all-ones matrix. 

We now consider a valid movement vector $\dot{\qbar}$ for the code definition vector $\qbar$. At $\qbar = \bar{\qbar}$, the nonnegativity constraint \eqref{eqn:QConstraintPositive} does not affect any of the elements of $\qbar$. The unit sum constraint \eqref{eqn:QConstraintTotal} requires the elements of a valid $\dot{\qbar}$ to sum to zero. Additionally, we define $\dot{\qbar}$ to be a unit vector, so we have 
\begin{equation}
    \label{eqn:ufc_qbardotConstraintTotal}
    \sum_{i=1}^{2^\kappa-1}{\dot{\qbar}_i} = 0
\end{equation}
and 
\begin{equation}
    \label{eqn:ufc_qbardotConstraintUnit}
    \sum_{i=1}^{2^\kappa-1}{\dot{\qbar}_i^2} = 1 \mathrm{.} 
\end{equation}

We are now ready to evaluate the effect of movement from $\qbar = \bar{\qbar}$ in the $\dot{\qbar}$ direction. The gradient in the $\dot{\qbar}$ direction is given by 
\begin{equation}
    \label{eqn:ufc_qdotGradient}
    \begin{split}
    \left. \nabla l(\qbar) \cdot \dot{\qbar} \right|_{\qbar=\bar{\qbar}} &= \sum_{i=1}^{2^\kappa-1} {-n \ln(\epsilon) \sum_{\delta=1}^{\kappa}{\left( K_{\delta} \binom{\kappa-1}{\delta}_2 e^{n \ln(\epsilon)\frac{2^\kappa-2^{\kappa-\delta}}{2^\kappa-1}} \right)} \dot{\qbar}_i} \\
    &\;\; = -n \ln(\epsilon) \sum_{\delta=1}^{\kappa}{\left( K_{\delta} \binom{\kappa-1}{\delta}_2 e^{n \ln(\epsilon)\frac{2^\kappa-2^{\kappa-\delta}}{2^\kappa-1}} \right)} \sum_{i=1}^{2^\kappa-1} {\dot{\qbar}_i} \\
    &\;\; = 0 \mathrm{.}
    \end{split}
\end{equation}

Finally, we consider the second derivative of equivocation loss with movement in the $\dot{\qbar}$ direction. This is given by the quadratic form $\dot{\qbar}^{\intercal} \mathbf{H}(l) \dot{\qbar}$ and is equal to 
\begin{equation}
    \label{eqn:ufc_qdotHessian1}
    \begin{split}
    \dot{\qbar}^{\intercal} \mathbf{H}(l) \dot{\qbar} &= \! n^2 \! \ln(\epsilon)^2  \sum_{\delta=1}^{\kappa}{\left( \! K_{\delta}   2^{\kappa - \delta - 1} \binom{\kappa \! - \! 2}{\delta-1}_{\!\!2} e^{n \ln(\epsilon)\frac{2^\kappa-2^{\kappa-\delta}}{2^\kappa-1}} \right)}  \dot{\qbar}^{\intercal} \mathbf{I} \dot{\qbar} + n^2 \ln(\epsilon)^2 \sum_{\delta=1}^{\kappa}{\left( K_{\delta} \binom{\kappa \! - \! 2}{\delta}_{\!\!2} e^{n \ln(\epsilon)\frac{2^\kappa-2^{\kappa-\delta}}{2^\kappa-1}} \right)} \dot{\qbar}^{\intercal} \mathbf{J} \dot{\qbar} \\
    &= n^2 \! \ln(\epsilon)^2  \sum_{\delta=1}^{\kappa}{\left( \! K_{\delta}   2^{\kappa - \delta - 1} \binom{\kappa \! - \! 2}{\delta-1}_{\!\!2} e^{n \ln(\epsilon)\frac{2^\kappa-2^{\kappa-\delta}}{2^\kappa-1}} \right)} \mathrm{.}
    \end{split}
\end{equation}
Using the substitutions $a= \frac{1}{2} n^2\ln(\epsilon)^2 e^{\frac{2^{\kappa} n \ln(\epsilon)}{2^\kappa-1}}$, $b=e^{\frac{-n \ln(\epsilon)}{2^\kappa-1}}$, and $d=\kappa-\delta$, this becomes 
\begin{equation}
    \label{eqn:ufc_qdotHessian2}
    \begin{split}
    &\dot{\qbar}^{\intercal} \mathbf{H}(l) \dot{\qbar} = a \sum_{d=0}^{\kappa-1}{\left( \! K_{\kappa-d} \binom{\kappa \! - \! 2}{\kappa-d-1}_{\!\!2} 2^{d} b^{2^d} \right)} \mathrm{,}
    \end{split}
\end{equation}
and using the definitions \eqref{eqn:GaussianBinomial} of the Gaussian binomial and \eqref{eqn:KConstantsLemmaStatement} of $K_\delta$, we obtain 
\begin{equation}
    \label{eqn:ufc_qdotHessian3}
    \begin{split}
    &\dot{\qbar}^{\intercal} \mathbf{H}(l) \dot{\qbar} = a \sum_{d=0}^{\kappa-1}{\left( \! \prod_{i=1}^{\kappa-d-1}{\!\!\! (1-2^i)} \! \prod_{j=0}^{\kappa-d-2}{\!\! \frac{2^{\kappa-2-j}-1}{2^{\kappa-d-1-j}-1}} 2^{d} b^{2^d} \!\right)} \!\mathrm{.}
    \end{split}
\end{equation}
Then using the substitution $j = \kappa-d-1-i$ in the first product yields 
\begin{equation}
    \label{eqn:ufc_qdotHessian3}
    \begin{split}
    \dot{\qbar}^{\intercal} \mathbf{H}(l) \dot{\qbar} &= a \sum_{d=0}^{\kappa-1}{\left( \! \prod_{j=0}^{\kappa-d-2}{\!\!\!(1\!-\!2^{\kappa-d-1-j})} \!\! \prod_{j=0}^{\kappa-d-2}{\! \frac{1\!-\!2^{\kappa-2-j}}{1\!-\!2^{\kappa-d-1-j}}} 2^{d} b^{2^d} \right)} \\
    &= a \sum_{d=0}^{\kappa-1}{\left( \! 2^{d} b^{2^d} \prod_{j=0}^{\kappa-d-2}{(1-2^{\kappa-2-j})}  \right)} \! \mathrm{,}
    \end{split}
\end{equation}
and substituting $i=\kappa-2-j$ gives
\begin{equation}
    \label{eqn:ufc_qdotHessian4}
    \begin{split}
    \dot{\qbar}^{\intercal} \mathbf{H}(l) \dot{\qbar} &= a \sum_{d=0}^{\kappa-1}{\left( \! 2^{d} b^{2^d} \prod_{i=d}^{\kappa-2}{(1-2^{i})}  \right)} \\
    &= \frac{a}{(1-2^{\kappa-1})} \sum_{d=0}^{\kappa-1}{\left( \! 2^{d} b^{2^d} \prod_{i=d}^{\kappa-1}{(1-2^{i})}  \right)} \mathrm{.}
    \end{split}
\end{equation}
Here, we have $a$ positive, $(1-2^{\kappa-1})$ negative, and by Lemma \ref{SuperexponentialLemma}, the sum in \eqref{eqn:ufc_qdotHessian4} is negative. Thus, $\dot{\qbar}^{\intercal} \mathbf{H}(l) \dot{\qbar}$ is positive, completing the proof.

\section{Proof of Theorem \ref{thm:sec_localOptimality1}}
\label{Appendix:secLocalOptimalityProof}
The proof proceeds as follows. First, an expression for the gradient vector and the Hessian matrix of equivocation with respect to $\qbar$ are derived. Next, the properties of a movement vector $\dot{\qbar}$ and an acceleration vector $\ddot{\qbar}$ are defined such that movement defined by these vectors is compliant with constraints \eqref{eqn:QConstraintPositive}, \eqref{eqn:QConstraintTotal}, and \eqref{eqn:sec_radiusConstraint}. The product of the gradient with any such $\dot{\qbar}$ is then shown to be zero. Finally, the second derivative with movement defined by $\dot{\qbar}$ and $\ddot{\qbar}$ is shown to be always positive. It is worth noting that although the nonnegativity constraint \eqref{eqn:QConstraintPositive} is imposed to invoke Corollary \ref{thm:zeroColumnCorollary}, this proof ultimately demonstrates local optimality of elements of all elements of $\check{\qbar}^{\{\kappa-1\}}$ regardless of this constraint. 

First, as in the case of the uniform vector fraction code, the gradient $\nabla l(\qbar)$ is given by \eqref{eqn:ufc_nabla1}, but in this case, the value of $\zeta(S)$ depends on both the dimension $d$ of $S$ and the dimension $v$ of the overlapping subspace $S \cup U$ of $S$ and the excluded subspace $U$. 
\begin{equation}
    \label{eqn:sec_overlapSize}
    \zeta(S^{\{d\}}:\mathrm{dim}(S^{\{d\}} \cup U^{\{u\}}) = v) = \frac{2^d-2^v}{2^\kappa-2^u} \mathrm{.} 
\end{equation}

Then to calculate the $i^{\text{th}}$ element of $\nabla l(\qbar)$, it is necessary to find the number of subspaces which contain $\nu(i)$ and which have a given dimension $d$ and overlap $U$ at a subspace of dimension $v$. This quantity depends on whether $\nu(i)$ is an element of $U$ (or equivalently whether $i < 2^u$). For the case that $\nu(i) \in U$, it is given by
\begin{equation}
    \label{eqn:sec_overlapInnerCount}
    \begin{split}
    \left|\{S^{\{d\}}:\mathrm{dim}(S^{\{d\}} \cup U^{\{u\}})| = v\text{, }\nu(i < 2^u) \in S\}\right| &=  \binom{u}{v}_2 \frac{2^v-1}{2^u-1} \prod_{i=0}^{d-v-1}{\frac{2^\kappa-2^{u+i}}{2^d-2^{v+i}}} \\
    &=  \binom{u}{v}_2 \frac{2^v-1}{2^u-1} \prod_{i=0}^{d-v-1}{\frac{2^{\kappa-v}-2^{u-v+i}}{2^{d-v}-2^{i}}} \\
    &=  \binom{u-1}{v-1}_2 2^{(u-v)(d-v)}\prod_{i=0}^{d-v-1}{\frac{2^{\kappa-u}-2^{i}}{2^{d-v}-2^{i}}} \\
    &=  \binom{u-1}{u-v}_2 2^{(u-v)(d-v)} \binom{\kappa-u}{d-v}_2 \mathrm{,} 
    \end{split}
\end{equation}
while if $\nu(i) \notin U$, it is given by
\begin{equation}
    \label{eqn:sec_overlapOuterCount}
    \begin{split}
    \left|\{S^{\{d\}}:\mathrm{dim}(S^{\{d\}} \cup U^{\{u\}})| = v\text{, }\nu(i \geq 2^u) \in S\}\right| &=  \binom{u}{v}_2 \prod_{i=0}^{d-v-2}{\frac{2^\kappa-2^{u+1+i}}{2^d-2^{v+1+i}}} \\
    & =  \binom{u}{v}_2 \prod_{i=0}^{d-v-2}{\frac{2^{\kappa-v-1}-2^{u-v+i}}{2^{d-v-1}-2^{i}}} \\
    & =  \binom{u}{v}_2 2^{(u-v)(d-v-1)} \prod_{i=0}^{d-v-2}{\frac{2^{\kappa-u-1}-2^{i}}{2^{d-v-1}-2^{i}}} \\
    & =  \binom{u}{u-v}_2 2^{(u-v)(d-v-1)} \binom{\kappa-u-1}{d-v-1}_2 \mathrm{.} 
    \end{split}
\end{equation}

Using these quantities with \eqref{eqn:ufc_nabla1}, the elements of $\nabla l(\check{\qbar}^{\{u\}}$ are given by 
\begin{equation}
    \label{eqn:sec_nabla2}
    \begin{split}
    \nabla l(\check{\qbar}^{\{u\}})_{i} &= -n\ln(\epsilon) e^{n\ln(\epsilon)} \sum_{\delta=1}^{\kappa}{\left. K_{\delta} \sum_{v=0}^{u}{e^{-n\ln(\epsilon) \frac{2^{\kappa-\delta}-2^v}{2^\kappa-2^u}} \cdot } \right.} \left. \left. \begin{cases} & \left. \binom{u-1}{u-v}_2 2^{(u-v)(\kappa-\delta-v)} \binom{\kappa-u}{\kappa-\delta-v}_2  \right) \\
    &\;\;\;\; \text{ if } i \leq 2^u \\
    & \left. \binom{u}{u-v}_2 2^{(u-v)(\kappa-\delta-v-1)} \binom{\kappa-u-1}{\kappa-\delta-v-1}_2 \right) \\
    &\;\;\;\; \text{ if } i > 2^u \mathrm{.} 
    \end{cases}
     \right. \right.
    \end{split}
\end{equation}
For the remainder of this work, we simplify the notation by representing the quantity $-n\ln(\epsilon)$ by the variable $t$. Using this substitution and limiting the expression to the first ($u=\kappa-1$) subspace exclusion code, this simplifies to 
\begin{equation}
    \label{eqn:sec_nabla3}
    \begin{split}
    &\nabla l(\check{\qbar}^{\{\kappa-1\}})_{i} = \begin{cases} &t e^{-t} \sum_{\delta=1}^{\kappa}{\left( K_{\delta} \left( \binom{\kappa-2}{\delta-1}_2 + \binom{\kappa-2}{\delta}_2 2^{\delta} e^{t2^{-\delta}} \right) \right)}  \\
    & \;\;\;\; \text{ if } i \leq 2^u \\
    & t e^{-t} \sum_{\delta=1}^{\kappa}{\left( K_{\delta} \left( \binom{\kappa-1}{\delta}_2 e^{t2^{-\delta}} \right) \right)} \\ &\;\;\;\; \text{ if } i > 2^u \mathrm{.} 
    \end{cases}
    \end{split}
\end{equation}

To calculate elements of the Hessian matrix for $l(q)$, it is necessary to find the number of subspaces which contain both $\nu(i)$ and $\nu(j)$ for any given $i$ and $j$ and which have a given dimension $d$ and overlap $U$ at a subspace of dimension $v$. If $i=j$, this quantity is simply given by \eqref{eqn:sec_overlapInnerCount} or \eqref{eqn:sec_overlapOuterCount}, depending on whether $\nu(i)$ is within $U$. On the other hand, if $i \neq j$, this quantity is given by 
\begin{equation}
    \label{eqn:sec_overlapNoteqInInCount}
    \begin{split}
    \left|\{S^{\{d\}}:\mathrm{dim}(S \cup U) = v\text{, }\nu(i),\nu(j) \in S, U\}\right| &=  \binom{u-2}{v-2}_2 \prod_{i=0}^{d-v-1}{\frac{2^\kappa-2^{u+i}}{2^d-2^{v+i}}} \\
    &=  \binom{u-2}{v-2}_2 \prod_{i=0}^{d-v-1}{\frac{2^{\kappa-v}-2^{u-v+i}}{2^{d-v}-2^{i}}} \\
    &=  \binom{u-2}{v-2}_2 2^{(u-v)(d-v)} \prod_{i=0}^{d-v-1}{\frac{2^{\kappa-u}-2^{i}}{2^{d-v}-2^{i}}} \\
    &=  \binom{u-2}{u-v}_2 2^{(u-v)(d-v)} \binom{\kappa-u}{d-v}_2 \mathrm{,} 
    \end{split}
\end{equation}
if both $\nu(i)$ and $\nu(j)$ are within $U$. If $\nu(i)$ and $\nu(j)$ are not within $U$, there are three possible scenarios. First, one of $\nu(i)$ or $\nu(j)$ may be within $U$ and the other may be not within $U$. In this case, the number of superspaces is given by
\begin{equation}
    \label{eqn:sec_overlapNoteqInOutCount}
    \begin{split}
    \left|\{S^{\{d\}}:\mathrm{dim}(S \cup U) = v\text{, }\nu(i)\in S,U;\nu(j) \in S, \notin U\}\right| &=  \binom{u}{v}_2 \frac{2^v-1}{2^u-1} \prod_{i=0}^{d-v-2}{\frac{2^\kappa-2^{u+i+1}}{2^d-2^{v+i+1}}} \\
    &= \binom{u-1}{v-1}_2 \prod_{i=0}^{d-v-2}{\frac{2^{\kappa-v-1}-2^{u-v+i}}{2^{d-v-1}-2^{i}}} \\
    &= \binom{u-1}{v-1}_2 2^{(u-v)(d-v-1)} \prod_{i=1}^{d-v-1}{\frac{2^{\kappa-u-1}-2^{i}}{2^{d-v-1}-2^{i}}} \\
    &= \binom{u-1}{u-v}_2 2^{(u-v)(d-v-1)} \binom{\kappa-u-1}{d-v-1}_2 \mathrm{.} 
    \end{split}
\end{equation}
In the second scenario, neither $\nu(i)$ nor $\nu(j)$ are within $U$, but $\nu(i)+\nu(j)$ is within $U$. In this case, the desired number of superspaces is 
\begin{equation}
    \label{eqn:sec_overlapNoteqOutOutCombinCount}
    \begin{split}
    \left|\{S^{\{d\}}\!:\mathrm{dim}(S \cup U) \!=\! v\text{, }\nu(i),\nu(j) \!\in\! S, \!\notin\! U, \nu(i)\!+\!\nu(j) \!\in\! U\}\right| &=  \binom{u}{v}_2 \frac{2^v-1}{2^u-1} \prod_{i=0}^{d-v-2}{\frac{2^\kappa-2^{u+1+i}}{2^d-2^{v+1+i}}} \\
    &=  \binom{u-1}{v-1}_2 \prod_{i=0}^{d-v-2}{\frac{2^{\kappa-v-1}-2^{u-v+i}}{2^{d-v-1}-2^{i}}} \\
    &=  \binom{u-1}{v-1}_2 2^{(u-v)(d-v-1)} \prod_{i=0}^{d-v-2}{\frac{2^{\kappa-u-1}-2^{i}}{2^{d-v-1}-2^{i}}} \\
    &=  \binom{u-1}{u-v}_2 2^{(u-v)(d-v-1)} \binom{\kappa-u-1}{d-v-1}_2 \mathrm{.} 
    \end{split}
\end{equation}
In the third scenario, none of $\nu(i)$, $\nu(j)$, or $\nu(i)+\nu(j)$ are within $U$. In this case, the number of superspaces is given by 
\begin{equation}
    \label{eqn:sec_overlapNoteqOutOutComboutCount}
    \begin{split}
    \left|\{S^{\{d\}}\!:\mathrm{dim}(S \cup U) = v\text{, }\nu(i),\nu(j), \nu(i)+\nu(j) \in S, \notin U\}\right| &=  \binom{u}{v}_2 \prod_{i=0}^{d-v-3}{\frac{2^\kappa-2^{u+2+i}}{2^d-2^{v+2+i}}} \\
    & =  \binom{u}{v}_2 \prod_{i=0}^{d-v-3}{\frac{2^{\kappa-v-2}-2^{u-v+i}}{2^{d-v-2}-2^{i}}} \\
    & =  \binom{u}{v}_2 2^{(u-v)(d-v-2)} \prod_{i=0}^{d-v-3}{\frac{2^{\kappa-u-2}-2^{i}}{2^{d-v-2}-2^{i}}} \\
    & =  \binom{u}{u-v}_2 2^{(u-v)(d-v-2)} \binom{\kappa-u-2}{d-v-2}_2 \mathrm{.} 
    \end{split}
\end{equation}

Using these quantities, the elements of the Hessian matrix $\mathbf{H}(l)$ of $l(q)$ may be expressed as 
\begin{equation}
    \label{eqn:sec_HessianDef}
    \begin{split}
    &\mathbf{H}(l)_{i,j} = t^2 e^{-t} \sum_{\delta=1}^{\kappa}{\left. K_{\delta} \sum_{v=0}^{u}{e^{t \frac{2^{\kappa-\delta}-2^v}{2^{\kappa}-2^{u}}}} \cdot \right.} \begin{cases}
    & \left. \binom{u-1}{u-v}_2 2^{(u-v)(\kappa-\delta-v)} \binom{\kappa-u}{\kappa-\delta-v}_2  \right. \\
    &\;\;\;\; \text{ if } i=j \leq 2^u \\
    & \left. \binom{u}{u-v}_2 2^{(u-v)(\kappa-\delta-v-1)} \binom{\kappa-u-1}{\kappa-\delta-v-1}_2 \right. \\
    &\;\;\;\; \text{ if } i=j > 2^u \\
    & \left. \binom{u-2}{u-v}_2 2^{(u-v)(\kappa-\delta-v)} \binom{\kappa-u}{\kappa-\delta-v}_2  \right. \\
    &\;\;\;\; \text{ if } (i \leq 2^u) \neq (j \leq 2^u) \\
    & \left. \binom{u-1}{u-v}_2 2^{(u-v)(\kappa-\delta-v-1)} \binom{\kappa-u-1}{\kappa-\delta-v-1}_2  \right. \\
    &\;\;\;\; \text{ if } i < 2^u, j \geq 2^u \text{ or } i,j \geq 2^u, \nu(i)+\nu(j) \in U \\
    & \left. \binom{u}{u-v}_2 2^{(u-v)(\kappa-\delta-v-2)} \binom{\kappa-u-2}{\kappa-\delta-v-2}_2  \right. \\
    &\;\;\;\; \text{ if } i,j \geq 2^u, \nu(i)+\nu(j) \notin U \mathrm{.}
    \end{cases}
    \end{split}
\end{equation}
Limiting this expression to the case of $u=\kappa-1$ gives
\begin{equation}
    \label{eqn:sec_HessianLimited}
    \begin{split}
    &\mathbf{H}(l)_{i,j} = \begin{cases}
    & t^2 e^{-t} \sum_{\delta=1}^{\kappa}{\left( K_{\delta} \cdot \left( \binom{\kappa-2}{\delta-1}_2  + \binom{\kappa-2}{\delta}_2 2^{\delta} e^{t 2^{-\delta}}  \right) \right)} \\
    &\;\;\;\; \text{ if } i=j \leq 2^u \\
    & t^2 e^{-t} \sum_{\delta=1}^{\kappa}{\left( K_{\delta} \cdot \left( \binom{\kappa-1}{\delta}_2 e^{t 2^{-\delta}} \right) \right)} \\
    &\;\;\;\; \text{ if } i=j > 2^u \\
    & t^2 e^{-t} \sum_{\delta=1}^{\kappa}{\left( K_{\delta} \cdot \left( \binom{\kappa-3}{\delta-1}_2 + \binom{\kappa-3}{\delta}_2 2^{\delta} e^{t 2^{-\delta}} \right) \right)} \\
    &\;\;\;\; \text{ if } (i \leq 2^u) \neq (j \leq 2^u) \\
    & t^2 e^{-t} \sum_{\delta=1}^{\kappa}{\left( K_{\delta} \cdot \left( \binom{\kappa-2}{\delta}_2 e^{t 2^{-\delta}} \right) \right)} \\
    &\;\;\;\; \text{ if } i < 2^u, j \geq 2^u \text{ or } i,j \geq 2^u, \nu(i)+\nu(j) \in U \mathrm{.}
    \end{cases}
    \end{split}
\end{equation}

Next, we turn our attention to the properties of a compliant differential element for $\qbar$. It is clear that $\dot{\qbar}$ must satisfy \eqref{eqn:ufc_qbardotConstraintTotal} because $\qbar$ must satisfy \eqref{eqn:QConstraintTotal}. We also require $\qbar$ to satisfy \eqref{eqn:ufc_qbardotConstraintUnit}- that is, to be a unit vector. Additionally, because $\qbar$ must satisfy the radius constraint \eqref{eqn:sec_radiusConstraint}, $\dot{\qbar}$ must satisfy an orthogonality constraint 
\begin{equation}
    \label{eqn:sec_qbardotConstraintOrthogonal}
    \rho(u)^{\intercal} \dot{\qbar} = 0 \mathrm{.}
\end{equation}
Using the definition \eqref{eqn:sec_defRadius} of $\rho(u)$, this constraint may equivalently be expressed as 
\begin{equation}
    \label{eqn:sec_qbardotConstraintOrthogonal1}
    \begin{split}
    \sum_{i=1}^{2^u-1}{\left( -\frac{1}{2^\kappa-1} \dot{\qbar}_i \right)} + \sum_{i=2^u}^{2^\kappa-1}{\left( \frac{2^u-1}{(2^\kappa-1)(2^\kappa-2^u)} \dot{\qbar}_i \right)} &=  -\frac{1}{2^\kappa-1} \sum_{i=1}^{2^u-1}{\left( \dot{\qbar}_i \right)} + \frac{2^u-1}{(2^\kappa-1)(2^\kappa-2^u)} \sum_{i=2^u}^{2^\kappa-1}{\left( \dot{\qbar}_i \right)} \\
    & = \frac{1}{2^\kappa-1} \left( - \sum_{i=1}^{2^\kappa-1}{\left(\dot{\qbar}_i\right)} + \frac{2^\kappa-1}{2^\kappa-2^u} \sum_{i=2^u}^{2^\kappa-1}{\left(\dot{\qbar}_i\right)} \right) = 0\mathrm{.} 
    \end{split}        
\end{equation}
Then combining \eqref{eqn:sec_qbardotConstraintOrthogonal1} with \eqref{eqn:ufc_qbardotConstraintTotal}, we obtain simply 
\begin{equation}
    \label{eqn:sec_qbardotConstraintOrthogonal2}
    \sum_{i=2^u}^{2^\kappa-1}{\left(\dot{\qbar}_i\right)} = \sum_{i=1}^{2^u-1}{\left(\dot{\qbar}_i\right)} = 0 \mathrm{.}
\end{equation}

\eqref{eqn:ufc_qbardotConstraintUnit}, \eqref{eqn:ufc_qbardotConstraintTotal}, and \eqref{eqn:sec_qbardotConstraintOrthogonal} comprise a complete set of constraints on a first-order differential element $\dot{\qbar}$ for $\qbar$. There is, however, an additional constraint on the second-order differential element $\ddot{\qbar}$ arising from the radius constraint \eqref{eqn:sec_radiusConstraint} on $\qbar$. Because the surface defined by \eqref{eqn:sec_radiusConstraint} has curvature equal to $\frac{1}{|\rho(u)|}$, any first-order unit movement vector $\dot{\qbar}$ orthogonal to $\rho(u)$ must be accompanied by a second-order movement vector equal to 
\begin{equation}
    \label{eqn:sec_qdotdot}
    \ddot{\qbar} = -\frac{\rho(u)}{|\rho(u)|^2} \mathrm{.}
\end{equation}
The components of $\ddot{\qbar}$ are thus given by 
\begin{equation}
    \label{eqn:sec_qdotdot2}
    \ddot{\qbar}_i = \begin{cases}
        \frac{2^{\kappa}-2^{u}}{2^{u}-1} & \text{ if } 1 \leq i < 2^u \\
        -1 & \text{ if } i \geq 2^u \mathrm{.}
    \end{cases} \mathrm{,}
\end{equation}
or limited to $u=\kappa-1$, 
\begin{equation}
    \label{eqn:sec_qdotdot2}
    \ddot{\qbar}_i = \begin{cases}
        \frac{2^{\kappa-1}}{2^{\kappa-1}-1} & \text{ if } 1 \leq i < 2^{\kappa-1} \\
        -1 & \text{ if } i \geq 2^{\kappa-1} \mathrm{.}
    \end{cases} \mathrm{.}
\end{equation}

Using the expression \eqref{eqn:sec_nabla3} and \eqref{eqn:sec_qbardotConstraintOrthogonal2}, we may compute the derivative of $l(\qbar)$ with movement in the $\dot{\qbar}$ direction as  
\begin{equation}
    \label{eqn:sec_firstDeriv}
    \begin{split}
    \nabla l(\qbar) \dot{\qbar} &= t e^{-t} \sum_{i=1}^{2^{\kappa-1}-1}{\left( \dot{\qbar}_i \sum_{\delta=1}^{\kappa}{\left( K_{\delta} \left( \binom{\kappa-2}{\delta-1}_{\!2} \!\!\!+\! \binom{\kappa-2}{\delta}_{\!2} 2^{\delta} e^{t2^{-\delta}} \right) \right)} \right)} + t e^{-t} \sum_{i=2^{\kappa-1}}^{2^{\kappa}-1}{\left( \dot{\qbar}_i \sum_{\delta=1}^{\kappa}{\left( K_{\delta} \binom{\kappa-1}{\delta}_{\!2} e^{t2^{-\delta}} \right)} \right)} \\
    &= \sum_{i=1}^{2^{\kappa-1}-1}{ \!\!\!\! \left( \dot{\qbar}_i\right)} \sum_{\delta=1}^{\kappa}{\left( K_{\delta} \left( \binom{\kappa-2}{\delta-1}_{\!2} \!\!\!+\! \binom{\kappa-2}{\delta}_{\!2} 2^{\delta} e^{t2^{-\delta}} \right) \right)} + t e^{-t} \sum_{i=2^{\kappa-1}}^{2^{\kappa}-1}{\left(\dot{\qbar}_i \right)} \sum_{\delta=1}^{\kappa}{\left( K_{\delta} \binom{\kappa-1}{\delta}_{\!2} e^{t2^{-\delta}} \right)} \\
    &= t e^{-t} \cdot 0 \cdot \sum_{\delta=1}^{\kappa}{\left( K_{\delta} \left( \binom{\kappa-2}{\delta-1}_{\!2} \!\!\!+\! \binom{\kappa-2}{\delta}_{\!2} 2^{\delta} e^{t2^{-\delta}} \right) \right)} + t e^{-t} \cdot 0 \cdot \sum_{\delta=1}^{\kappa}{\left( K_{\delta} \binom{\kappa-1}{\delta}_{\!2} e^{t2^{-\delta}} \right)} = 0 \mathrm{.} 
    \end{split}
\end{equation}

We must now examine the second derivative of $l$ with respect to movement in a compliant $\dot{\qbar}$ direction. Such movement produces a second-order term resulting from the first-order movement term multiplied by the Hessian as well as from the second-order movement term multiplied by the gradient. Thus the second derivative is equal to 
\begin{equation}
    \label{eqn:sec_secDerivDef}
    \frac{\partial^2l(\qbar)}{\partial \qbar^2} = \dot{\qbar}^{\intercal}\mathbf{H}(l)\dot{\qbar} + \nabla l(\qbar) \ddot{\qbar} \mathrm{.} 
\end{equation}
To calculate the first term, we observe from \eqref{eqn:sec_HessianLimited} that $\left. \mathbf{H}(l) \right|_{\check{\qbar}^{\{\kappa-1\}}}$ may be expressed in terms of identity matrices and all-ones matrices as 
\begin{equation}
    \label{eqn:sec_HessianIJ}
    \begin{split}
    &\mathbf{H}(l) \rvert_{\check{\qbar}^{\{\kappa-1\}}} = \begin{bmatrix}
        D_1 \mathbf{J}_{2^{\kappa}-1}
    \end{bmatrix} + \begin{bmatrix}
        D_2 \mathbf{J}_{2^{\kappa-1}-1} & \mathbf{0} \\ \mathbf{0} & \mathbf{0}
    \end{bmatrix} + \begin{bmatrix}
        D_3 \mathbf{I}_{2^{\kappa-1}-1} & \mathbf{0} \\ \mathbf{0} & D_4 \mathbf{I}_{2^{\kappa-1}}
    \end{bmatrix} \mathrm{,} 
    \end{split}
\end{equation}
where 
\begin{equation}
    \label{eqn:sec_D1_def}
    D_1 = t^2 e^{-t} \sum_{\delta=1}^{\kappa}{\left( \!K_{\delta}\! \left( \!\!\binom{\kappa-2}{\delta}_{\!\!2} e^{t 2^{\!-\delta}} \right)\!\! \right)} \mathrm{,}
\end{equation}
\begin{equation}
    \label{eqn:sec_D2_def}
    D_2 = t^2 e^{-t} \sum_{\delta=1}^{\kappa}{\left( \!K_{\delta}\! \left( \!\!\binom{\kappa-3}{\delta-1}_{\!\!2} \!\!\!+\!\! \binom{\kappa-3}{\delta}_{\!\!2} 2^{\delta} e^{t 2^{\!-\delta}} \right)\!\! \right)} \!\!-\! D_1 \mathrm{,}
\end{equation}
\begin{equation}
    \label{eqn:sec_D3_def}
    \begin{split}
    D_3 &= t^2 e^{-t} \sum_{\delta=1}^{\kappa}{\left( \!K_{\delta}\! \left( \!\!\binom{\kappa-2}{\delta-1}_{\!\!2}  \!\!\!+\!\! \binom{\kappa-2}{\delta}_{\!\!2} 2^{\delta} e^{t 2^{\!-\delta}}  \right)\!\! \right)} \!-\! D_2 \!\!-\! D_1 \\
    &=\! t^2 e^{-t} \sum_{\delta=1}^{\kappa}{\left( \!K_{\delta} 2^{\kappa-\delta-1} \left( \!\!\binom{\kappa-3}{\delta-2}_{\!\!2} \!\!\!+\!\! \binom{\kappa-3}{\delta-1}_{\!\!2} 2^{\delta-1} e^{t 2^{\!-\delta}}  \right)\!\! \right)} \mathrm{,}
    \end{split}
\end{equation}
and 
\begin{equation}
    \label{eqn:sec_D4_def}
    \begin{split}
    D_4 &= t^2 e^{-t} \sum_{\delta=1}^{\kappa}{\left( \!K_{\delta}\! \left( \!\!\binom{\kappa-1}{\delta}_{\!\!2} e^{t 2^{\!-\delta}} \right)\!\! \right)} - D_1 \\
    &=\! t^2 e^{-t} \sum_{\delta=1}^{\kappa}{\left( \!K_{\delta}\! \left( \!\!\binom{\kappa-2}{\delta-1}_{\!\!2} 2^{\kappa-\delta-1} e^{t 2^{\!-\delta}} \right) \!\!\right)} \mathrm{.}
    \end{split}
\end{equation}
It is clear from inspection of \eqref{eqn:sec_HessianIJ} that the first two terms will equal zero when multiplied in a quadratic form with a $\dot{\qbar}$ that satisfies \eqref{eqn:ufc_qbardotConstraintTotal} and \eqref{eqn:sec_qbardotConstraintOrthogonal2}. Furthermore, it is clear that a quadratic form of the third term of \eqref{eqn:sec_HessianIJ} with a unit vector $\dot{\qbar}$ must have value between $D_3$ and $D_4$. 

The next step in the proof is to show that $D_3 \geq D_4$. To do this, we first recall the expression \eqref{eqn:MersenneProof9} for the quantity $a_n$ defined by \eqref{eqn:MersenneProof1}. With minor changes to variables and limits, we may write 
\begin{equation}
    \label{eqn:sec_MersenneExpression}
    \begin{split}
    \sum_{i=1}^{n}{\left( 2^{ci} \prod_{j=i}^{n}{(1-2^j)} \right)} &= (1-2^n) \sum_{j=1}^{n}{\left( \binom{n-1}{j-1}_2 2^{c} \cdot \prod_{i=1}^{j-1}{(2^c-2^i)} \right)} \mathrm{.}
    \end{split}
\end{equation}

By the properties of Gaussian binomials, $D_4$ may be expressed as 
\begin{equation}
    \label{eqn:sec_D4_1}
    \begin{split}
    D_4 &= t^2 e^{-t} \sum_{\delta=1}^{\kappa}{\left( \!K_{\delta} 2^{\kappa-\delta-1} \!\left( \!\!\binom{\kappa\!-\!3}{\delta\!-\!2}_{\!\!2} \!\!\!+\!\! \binom{\kappa\!-\!3}{\delta\!-\!1}_{\!\!2} 2^{\delta-1} \right) e^{t 2^{-\delta}} \right)} \\
    &= D_3 - t^2 e^{-t} \sum_{\delta=2}^{\kappa}{\left( K_{\delta} \cdot 2^{\kappa-\delta-1}  \binom{\kappa-3}{\delta-2}_{\!\!2} \left(1-e^{t 2^{-\delta}} \right)  \right)} \mathrm{.} 
    \end{split}
\end{equation}
The difference between $D_3$ and $D_4$ is then given by 
\begin{equation}
    \label{eqn:sec_D3-D4}
    \begin{split}
    (D_3-D_4)/(t^2 e^{-t} )&= \!\sum_{\delta=2}^{\kappa}{\left( \prod_{i=1}^{\delta-1}{(1\!-\!2^i)} \cdot 2^{\kappa-\delta-1}  \prod_{\iota=0}^{\delta-3}{\frac{1 \!-\! 2^{\kappa-3-\iota}}{1 \!-\! 2^{\delta-2-\iota}} } \left(1\!-\!e^{t 2^{-\delta}} \right) \! \right)} \\
    &= \!\sum_{\delta=2}^{\kappa}{\left( \prod_{i=1}^{\delta-1}{(1\!-\!2^{\delta-i})} \cdot 2^{\kappa-\delta-1}  \prod_{\iota=2}^{\delta-1}{\frac{1 \!-\! 2^{\kappa-1-\iota}}{1 \!-\! 2^{\delta-\iota}} } \left(1\!-\!e^{t 2^{-\delta}} \right) \! \right)} \\
    &= \!\sum_{\delta=2}^{\kappa}{\left( (1\!-\!2^{\delta-1}) \cdot 2^{\kappa-\delta-1}  \prod_{\iota=2}^{\delta-1}{(1 \!-\! 2^{\kappa-1-\iota} )} \left(1\!-\!e^{t 2^{-\delta}} \right) \! \right)} \\
    &= \!\sum_{\delta=2}^{\kappa}{\left( (1\!-\!2^{\delta-1}) \cdot 2^{\kappa-\delta-1}  \prod_{j=\kappa-\delta}^{\kappa-3}{(1 \!-\! 2^{j} )} \left(1\!-\!e^{t 2^{-\delta}} \right) \! \right)} \\
    &= \!\sum_{i=0}^{\kappa-2}{\left( (1-2^{\kappa-i-1}) \cdot 2^{i-1}  \prod_{j=i}^{\kappa-3}{(1 \!-\! 2^{j} )} \left(1\!-\!e^{t 2^{i-\kappa}} \right) \! \right)} \\
    &= \!\frac{1}{1-2^{\kappa-2}} \sum_{i=1}^{\kappa-2}{\left( 2^{i-1}  \prod_{j=i}^{\kappa-2}{(1 \!-\! 2^{j} )} \left(1\!-\!e^{t 2^{i-\kappa}} \right) \! \right)} - \frac{2^{\kappa-2}}{1-2^{\kappa-2}} \sum_{i=1}^{\kappa-2}{\left(  \prod_{j=i}^{\kappa-2}{(1 \!-\! 2^{j} )} \left(1\!-\!e^{t 2^{i-\kappa}} \right) \! \right)}\mathrm{.}
    \end{split}
\end{equation}
We now expand the exponentials in \eqref{eqn:sec_D3-D4} to yield 
\begin{equation}
    \label{eqn:sec_D3-D4_1}
    \begin{split}
    (D_3-D_4)/(t^2 e^{-t} ) &= -\frac{1}{1\!-\!2^{\kappa-2}} \sum_{i=1}^{\kappa-2}{\left( 2^{i-1}  \prod_{j=i}^{\kappa-2}{(1 - 2^{j} )} \sum_{\iota=1}^{\infty}{\frac{(2^{-\kappa}t)^\iota 2^{i\iota}}{\iota!}}  \right)}  + \frac{2^{\kappa-2}}{1-2^{\kappa-2}} \sum_{i=1}^{\kappa-2}{\left(  \prod_{j=i}^{\kappa-2}{(1 - 2^{j} )} \sum_{\iota=1}^{\infty}{\frac{(2^{-\kappa}t)^\iota 2^{i\iota}}{\iota!}}  \right)} \\
    &= -\sum_{\iota=1}^{\infty}{\frac{(2^{-\kappa}t)^\iota }{\iota!} \frac{1}{1\!-\!2^{\kappa-2}} \frac{1}{2} \sum_{i=1}^{\kappa-2}{\left( 2^{i(\iota+1)} \prod_{j=i}^{\kappa-2}{(1 \!-\! 2^{j} )}   \right)}} \!+\! \sum_{\iota=1}^{\infty}{\frac{(2^{-\kappa}t)^\iota}{\iota!} \frac{2^{\kappa-2}}{1-2^{\kappa-2}} \sum_{i=1}^{\kappa-2}{\left(  2^{i\iota} \prod_{j=i}^{\kappa-2}{(1 \!-\! 2^{j} )}   \right)}}\mathrm{.}
    \end{split}
\end{equation}
Next, we use \eqref{eqn:sec_MersenneExpression} to obtain 
\begin{equation}
    \label{sec_D3-D4_2}
    \begin{split}
    D_3-D_4 = &-t^2 e^{-t} \sum_{\iota=1}^{\infty}{\frac{(2^{-\kappa}t)^\iota }{\iota!} \frac{1}{1-2^{\kappa-2}} \frac{1}{2} \left( (1-2^{\kappa-2}) \sum_{j=1}^{\kappa-2}{\left( \binom{\kappa-3}{j-1}_2 2^{\iota+1} \cdot \prod_{i=1}^{j-1}{(2^{\iota+1}-2^i)} \right)} \right)} \\
    &+ t^2 e^{-t} \sum_{\iota=1}^{\infty}{\frac{(2^{-\kappa}t)^\iota }{\iota!} \frac{1}{1-2^{\kappa-2}} \frac{1}{2} \left((1-2^{\kappa-2}) 2^{\kappa-1} \sum_{j=1}^{\kappa-2}{\left( \binom{\kappa-3}{j-1}_2 2^{\iota} \cdot \prod_{i=1}^{j-1}{(2^{\iota}-2^i)} \right)} \right)} \\
    &= t^2 e^{-t} \sum_{\iota=1}^{\infty}{\left( \frac{(2^{-\kappa}t)^\iota }{\iota!} 2^{\iota}  \sum_{j=1}^{\kappa-2}{\left( \binom{\kappa-3}{j-1}_2 \cdot \right.} \left. \left( 2^{\kappa-2} \prod_{i=1}^{j-1}{(2^{\iota}-2^i)} - \prod_{i=1}^{j-1}{(2^{\iota+1}-2^i)} \right) \!\! \right) \vphantom{\frac{(2^{-\kappa}t)^\iota }{\iota!} 2^{\iota}  \sum_{j=1}^{\kappa-2}{\left( \binom{\kappa-3}{j-1}_2 \cdot \right.}} \!\!\! \right)} \\
    &= t^2 e^{-t} \sum_{\iota=1}^{\infty}{\left( \frac{(2^{-\kappa}t)^\iota }{\iota!} 2^{\iota}  \sum_{j=1}^{\kappa-2}{\left( \binom{\kappa-3}{j-1}_2 \cdot \right.} \left. \left( 2^{\kappa-2} (2^{\iota}-2^{j-2}) \prod_{i=1}^{j-2}{(2^{\iota}-2^i)} - \prod_{i=0}^{j-2}{(2^{\iota+1}-2^{i+1})} \right) \!\! \right) \vphantom{\frac{(2^{-\kappa}t)^\iota }{\iota!} 2^{\iota}  \sum_{j=1}^{\kappa-2}{\left( \binom{\kappa-3}{j-1}_2 \cdot \right.}} \!\!\! \right)} \\
    &= t^2 e^{-t} \sum_{\iota=1}^{\infty}{\left( \frac{(2^{-\kappa}t)^\iota }{\iota!} 2^{\iota}  \sum_{j=1}^{\kappa-2}{\left( \binom{\kappa-3}{j-1}_2 \cdot \right.} \left. \left( \left( 2^{\kappa-2} (2^{\iota}-2^{j-2}) - 2^{j-1} (2^{\iota}-1) \right) \prod_{i=1}^{j-2}{(2^{\iota}-2^i)} \right) \!\! \right) \vphantom{\frac{(2^{-\kappa}t)^\iota }{\iota!} 2^{\iota}  \sum_{j=1}^{\kappa-2}{\left( \binom{\kappa-3}{j-1}_2 \cdot \right.}} \!\!\! \right)} \mathrm{.}
    \end{split}
\end{equation}
Now considering the last line of \eqref{sec_D3-D4_2}, observe that if $j-2 \geq \iota$, the product contains a zero term, so the entire expression is equal to zero. If $j-2 < \iota$, then since $j$ attains a maximum value of $\kappa-2$, the first term of $2^{\kappa-2}(2^\iota-2^{j-2})$ is always greater than the second term of $2^{j-1}(2^\iota-1)$. Thus all terms are zero or positive, and we may conclude that $D_3 \geq D_4$. Then we may further conclude that $\dot{\qbar}^{\intercal} \mathbf{H}(l) \dot{\qbar} \geq D_4$.

The next step in calculating the second derivative is to consider the second term of \eqref{eqn:sec_secDerivDef}, which arises as a result of the requirement that $\qbar$ retain a fixed distance from $\bar{\qbar}$. This term may be expanded using \eqref{eqn:sec_nabla3} and \eqref{eqn:sec_qdotdot} to yield 
\begin{equation}
    \label{eqn:sec_NablaDotDot1}
    \begin{split}
    \nabla l(\qbar) \ddot{\qbar} &= \! \frac{(2^{\kappa-1}\!\!-\!1) 2^{\kappa-1}}{2^{\kappa-1}\!\!-\!1} t e^{-t} \!\sum_{\delta=1}^{\kappa}{\!\left( \!\!K_{\delta}\! \left( \!\binom{\kappa\!-\!2}{\delta\!-\!1}_{\!\!2} \!\!+\! \binom{\kappa\!-\!2}{\delta}_{\!\!2} 2^{\delta} e^{t2^{-\delta}} \!\right) \!\right)} - (2^{\kappa-1}) t e^{-t} \sum_{\delta=1}^{\kappa}{\left( K_{\delta} \left( \!\binom{\kappa-1}{\delta}_{\!\!2} e^{t2^{-\delta}} \right) \! \right)} \\
    & = \!2^{\kappa-1} t e^{-t} \sum_{\delta=1}^{\kappa}{\left( K_{\delta} \left( \!\binom{\kappa-2}{\delta-1}_{\!\!2} + \binom{\kappa-2}{\delta}_{\!\!2} 2^{\delta} e^{t2^{-\delta}} - \binom{\kappa-1}{\delta}_{\!\!2} e^{t2^{-\delta}} \right) \!\right)} \\
    & =\! 2^{\kappa-1} t e^{-t} \sum_{\delta=1}^{\kappa}{\left( K_{\delta} \binom{\kappa-2}{\delta-1}_{\!\!2} \left( 1- e^{t2^{-\delta}} \right) \!\right)}  \mathrm{.} 
    \end{split}
\end{equation}
Combining this with \eqref{eqn:sec_D4_def} gives 
\begin{equation}
    \label{eqn:sec_secDeriv_2}
    \begin{split}
    \frac{\partial^2l(\qbar)}{\partial \qbar^2} \geq D_4 + \nabla l(\qbar) \ddot{\qbar} &= 2^{\kappa-1} t e^{-t} \sum_{\delta=1}^{\kappa}{\left( \!K_{\delta}\! \left( \binom{\kappa\!-\!2}{\delta\!-\!1}_{\!\!2} \left( 2^{-\delta} t e^{t 2^{-\delta}} \!+\! 1 \!-\! e^{t2^{-\delta}} \right) \right) \right)} \mathrm{.} 
    \end{split}
\end{equation}
Dividing out the always-positive pre-sum terms and continuing, we have 
\begin{equation}
    \label{eqn:sec_secDeriv_3}
    \begin{split}
    \frac{\partial^2l(\qbar)}{\partial \qbar^2} \frac{1}{2^{\kappa-1}te^{-t}}    &\geq\! \sum_{\delta=1}^{\kappa}{\left( K_{\delta} \left( \binom{\kappa-2}{\delta-1}_2 \left( 2^{-\delta} t e^{t 2^{-\delta}} + 1 - e^{t2^{-\delta}} \right) \right) \right)} \\
    & =\! \sum_{\delta=1}^{\kappa}{\!\left( \prod_{i=1}^{\delta-1}{\!(1\!-\!2^i)} \!\left( \prod_{\iota=0}^{\delta-2}{\!\frac{1 \!-\! 2^{\kappa-2-\iota}}{1 \!-\! 2^{\delta-1-\iota}} } \!\left( 2^{-\delta} t e^{t 2^{-\delta}} \!\!\!+\! 1 \!-\! e^{t2^{-\delta}} \right)\! \right)\! \right)} \\
    & =\! \sum_{\delta=1}^{\kappa}{\left( \left( \prod_{\iota=0}^{\delta-2}{(1 - 2^{\kappa-2-\iota})} \left( 2^{-\delta} t e^{t 2^{-\delta}} + 1 - e^{t2^{-\delta}} \right) \right) \right)} \\
    & =\! \sum_{d=0}^{\kappa-1}{\left( \left( \prod_{\iota=0}^{\kappa-d-2}{(1 - 2^{\kappa-2-\iota})} \left( 2^{d-\kappa} t e^{t 2^{d-\kappa}} \!\!\!+\! 1 \!-\! e^{t2^{d-\kappa}} \right) \!\right)\! \right)} \\
    & =\! \sum_{d=0}^{\kappa-1}{\left( \left( \prod_{i=d}^{\kappa-2}{(1 - 2^{i})} \left( \sum_{j=2}^{\infty}{\frac{(j-1)(t2^{d-\kappa})^j}{j!}} \right) \right) \right)} \\
    & =\! \frac{1}{1-2^{\kappa-1}} \sum_{j=2}^{\infty}{\frac{(j-1)(t2^{-\kappa})^j}{j!}}\sum_{d=0}^{\kappa-1}{\left( 2^{jd} \prod_{i=d}^{\kappa-1}{(1 - 2^{i})} \right)} \mathrm{.} 
    \end{split}
\end{equation}
Now by Lemma \ref{thm:exponentialMersenneLemma}, the last sum of \eqref{eqn:sec_secDeriv_3} is negative, as is the $1/(1-2^{\kappa-1})$ term, while all other terms are positive. Thus, we have  $\partial^2l(\qbar)/\partial \qbar^2 \geq D_4 + \nabla l(\qbar) \ddot{\qbar} > 0$. 

\section{Proof of Theorem \ref{thm:x2_zc}}
\label{apx:x2_zc_proof}
We prove this theorem using an approach similar to that used in Appendix \ref{Appendix:zcDerivative} for Theorem \ref{thm:zeroColumn}- we show that movement in the direction of $\grave{q}$ as defined in \eqref{eqn:zc_defQGrave} always results in a reduction of $\lambda(n,\epsilon,q)$. Hence, we first find the derivative $\grave{\lambda}(n,\epsilon,q)$ of $\lambda(n,\epsilon,q)$ with movement in the $\grave{q}$ direction. This is given by 
\begin{equation}
    \label{eqn:x2_zc_lambdaGraveDef}
    \grave{\lambda}(n,\epsilon,q) = (2-\epsilon)^n 2^{-\kappa} \sum_{S \in \Xi(W,\kappa-1)}{\grave{\varphi}(S,n,\epsilon,q)}
    \mathrm{,}
\end{equation}
where $\grave{\varphi}(S,n,\epsilon,q)$ is the rate of change of $\varphi(S,n,\epsilon,q)$ and is given (in terms of $\grave{\zeta}(S,q)$ as defined in \eqref{eqn:zc_defZetaGrave}) by
\begin{equation}
    \label{eqn:x2_zc_varphiGraveDef}
    \begin{split} 
    \grave{\varphi}(S,n,\epsilon,q) &= -n\ln(\frac{\epsilon}{2-\epsilon}) \grave{\zeta}(S,q) e^{n\ln(\frac{\epsilon}{2-\epsilon})(1-\zeta(S,q))} \\
    &= n\ln(\frac{\epsilon}{2-\epsilon}) (1-\zeta(S,q)) e^{n\ln(\frac{\epsilon}{2-\epsilon})(1-\zeta(S,q))}\mathrm{.}
    \end{split}
\end{equation}
Now we note that $(1-\zeta(S,q))$ is always nonnegative. Furthermore, $(1-\zeta(S,q))$ will be positive for at least one $S$ (ensuring that $\grave{\lambda}(n,\epsilon,q)$ is negative) except in the case that $q_0=1$. Then because the $q$ defined by the case $q_0=1$ is clearly not locally optimal and $\lambda(n,\epsilon,q)$ always decreases with movement in the $\grave{q}$ direction, any $q$ with $q_0 \neq 0$ is not locally (or globally) optimal. 

\section{Proof of Theorem \ref{thm:x2_uvf}}
\label{apx:x2_uvf_proof}
We begin by introducing a vector-valued function $\mathcal{Z}$ of $\qbar$ defined by 
\begin{equation}
    \label{eqn:x2_ZetaDef}
    \mathcal{Z}(\qbar)_i = \zeta(\Xi(W,\kappa-1)_i,\qbar) \mathrm{.}
\end{equation}
That is, $\mathcal{Z}$ represents the value of $\zeta(S^{\{\kappa-1\}},\qbar)$ for every dimension-$(\kappa-1)$ subspace of the global space $W$. The only liberty we have taken with this definition is that we have implied an ordering of the elements of $\Xi(W,\kappa-1)$. The particular ordering is not important, but for completeness, we use the convention that the $i^{\mathrm{th}}$ element of $\Xi(W,\kappa-1)$ is that defined by $\nu(i)^{\intercal} \; \mathrm{flip}(v)=0 \; \forall \; v \in \Xi(W,\kappa-1)_i$. That is, the bit-reversed version of the $i^{\mathrm{th}}$ binary vector defines the null space of the $i^{\mathrm{th}}$ subspace. The $\chi^2$ divergence may be expressed simply in terms of a vector $\xi = \mathcal{Z}(\qbar)$ as 
\begin{equation}
    \label{eqn:x2_lambdaOfZ}
    \lambda(n,\epsilon,\qbar) = (2-\epsilon)^{n}2^{-\kappa} \left( 1+ e^{-\tau}\sum_{i=1}^{2^{\kappa}-1}{ e^{\tau \xi_i} } \right) -1 \mathrm{,}
\end{equation}
where 
\begin{equation}
    \label{eqn:x2_tauDef}
    \tau = -n \mathrm{ln} \!\! \left( \frac{\epsilon}{2-\epsilon} \right) \mathrm{.}    
\end{equation}
(See \eqref{eqn:x2_lambdaFinal} and \eqref{eqn:x2_def_varphi}.) 

We note several important properties of the $\mathcal{Z}(\qbar)$ function. First, the number of elements in $\mathcal{Z}(\qbar)$ is equal to the number of elements in $\qbar$. Second, $\mathcal{Z}(\qbar)$ is a linear function and may be defined in terms of a truncated binary Hadamard matrix. Such a matrix is full-rank, so there is a one-to-one relationship between $\qbar$ and $\mathcal{Z}(\qbar)$. Third, for any $\qbar$ that satisfies the unit-sum constraint \eqref{eqn:QConstraintTotal}, $\mathcal{Z}(\qbar)$ also satisfies the element-wise sum constraint 
\begin{equation}
    \label{eqn:x2_ZetaConstraintTotal}
    \sum_{i=1}^{2^\kappa-1}{\left(\mathcal{Z}(\qbar)_i\right)} = 2^{\kappa-1}-1 \mathrm{.} 
\end{equation}
Fourth, for vectors $\qbar$ and that comply with the unit-sum constraint \eqref{eqn:QConstraintTotal}, distances between corresponding $\mathcal{Z}(\qbar)$ are linearly related. The relationship may be determined by considering the sum of the squared elements of $\Delta \mathcal{Z}(\qbar)$. Each $\Delta \mathcal{Z}(\qbar)_i^2$ term accumulates $2^{\kappa-1}-1$ on-diagonal terms of the form $\Delta \qbar_i^2$ and $(2^{\kappa-1}-1)(2^{\kappa-1}-2)$ off-diagonal terms of the form $\Delta \qbar_i \Delta \qbar_j$. When all the $\Delta \mathcal{Z}(\qbar)^2$ are summed, each on-diagonal term occurs with equal frequency, so there are a total of $2^{\kappa-1}-1$ of each on-diagonal term. Similarly, all off-diagonal terms occur with equal frequency, so there are a total of $(2^{\kappa-1}-1)(2^{\kappa-1}-2)(2^{\kappa}-1)/((2^{\kappa}-1)(2^{\kappa}-2))$ of each off-diagonal term. The sum of all the squared $\Delta \mathcal{Z}(\qbar)$ terms may then be expressed as 
\begin{equation}
    \label{eqn:x2_zetaSumSquares}
    \begin{split}
    \sum_{i=1}^{2^{\kappa}\!-1}{\!\Delta \mathcal{Z}(\qbar)_i^2} &= \! (2^{\kappa-1}\!-\!1) \!\! \sum_{i=1}^{2^{\kappa}\!-1}{\! \Delta \qbar_i^2} + (2^{\kappa-2}\!-\!1)  \sum_{i=1}^{2^{\kappa}\!-1}{  \sum_{ j \neq i = 1}^{2^{\kappa}\!-1}{ \!\! \Delta \qbar_i \Delta \qbar_j}} \\
    &=  \! 2^{\kappa-2} \! \sum_{i=1}^{2^{\kappa}-1}{\! \Delta \qbar_i^2} + (2^{\kappa-2}\!-\!1)  \sum_{i=1}^{2^{\kappa}\!-1}{\;  \sum_{j=1}^{2^{\kappa}\!-1}{\! \Delta \qbar_i \Delta \qbar_j}} \\
    &= \! 2^{\kappa-2} \! \sum_{i=1}^{2^{\kappa}-1}{\Delta \qbar_i^2} + (2^{\kappa-2}\!-\!1) \! \left(\sum_{i=1}^{2^{\kappa}-1}{\Delta \qbar_i } \right)^2\mathrm{.}
    \end{split}
\end{equation}
Then because the elements of $\qbar$ have constant sum, the elements of $\Delta \qbar$ sum to zero, so the distance relationship, expressed in terms of squared distance, is 
\begin{equation}
    \label{eqn:x2_ZetaRelationDistance}
    |\mathcal{Z}(\qbar)-\mathcal{Z}(\qbar')|^2 = 2^{\kappa-2} \cdot |\qbar-\qbar'|^2 \mathrm{.} 
\end{equation}
Fifth, when the set of allowable values for $\qbar$ is further restricted by the nonnegativity constraint \eqref{eqn:QConstraintPositive} (in addition to the unit-sum constraint \eqref{eqn:QConstraintTotal}), the allowable values of $\mathcal{Z}(\qbar)$ are also restricted. These restrictions clearly include the requirement that the elements of $\mathcal{Z}(\qbar)$ must be nonnegative: 
\begin{equation}
    \label{eqn:x2_ZetaConstraintPositive}
    \mathcal{Z}(\qbar)_i \geq 0 \;\; \forall \; i \mathrm{.} 
\end{equation}
This constraint is not sufficient to ensure that the corresponding $\qbar$ satisfies \eqref{eqn:QConstraintPositive}. Nevertheless, the analyses performed in this work only require that $\mathcal{Z}(\qbar)$ satisfy \eqref{eqn:x2_ZetaConstraintPositive}. 

Using these properties, we may identify the vector $\xi$ which maximizes $\lambda(n,\epsilon,\mathcal{Z}^{-1}(\xi))$ subject to \eqref{eqn:x2_ZetaConstraintTotal} and \eqref{eqn:x2_ZetaConstraintPositive}. (Such a maximum is guaranteed to exist because \eqref{eqn:x2_ZetaConstraintTotal} and \eqref{eqn:x2_ZetaConstraintPositive} constrain $\xi$ to a compact set.) Then if the corresponding $\qbar = \mathcal{Z}^{-1}(\xi)$ satisfies all applicable constraints, it is guaranteed to be optimal. 

With this approach in mind, we now show that the optimal $\xi$ subject to \eqref{eqn:x2_ZetaConstraintTotal} and \eqref{eqn:x2_ZetaConstraintPositive} is given by the uniform-valued vector, denoted  $\bar{\xi}$, defined by 
\begin{equation}
    \label{x2_zetaUniformFracDef}
    \bar{\xi}_i = \frac{2^{\kappa-1}-1}{2^{\kappa}-1} \mathrm{.} 
\end{equation}
We do this by considering the minimum and maximum elements of $\xi$, denoted $\xi_a$ and $\xi_b$, respectively. Then if $\xi_a \neq \xi_b$, we may define a movement direction
\begin{equation}
    \label{eqn:x2_zetaDotUniform}
    \dot{\xi}_i = \begin{cases}
        1 &\text{if } i=a \\
        -1 &\text{if } i=b \\
        0 &\text{otherwise} \mathrm{.}
    \end{cases} 
\end{equation}
Movement in this direction is clearly compliant with constraints \eqref{eqn:x2_ZetaConstraintTotal} and \eqref{eqn:x2_ZetaConstraintPositive}, and the corresponding change $\dot{\lambda}(n,\epsilon,\qbar)$ in $\lambda(n,\epsilon,\qbar)$ is given by 
\begin{equation}
    \label{eqn:x2_lambdaDot}
    \dot{\lambda}(n,\epsilon,\qbar) = (2-\epsilon)^n 2^{-\kappa} \tau \left( e^{\tau \xi_a} - e^{\tau \xi_b} \right) \mathrm{.} 
\end{equation}
Then because $\xi_a < \xi_b$, $\dot{\lambda}(n,\epsilon,\qbar)$ is clearly negative, and $\xi$ is not optimal. Because the optimum is guaranteed to exist and no $\xi$ with $\mathrm{min}_i(\xi_i) < \mathrm{max}_i(\xi_i)$ is optimal, the optimum must be the uniform vector $\bar{\xi}$. 

As a final step, we must show that $\bar{\xi}$ corresponds to a valid $\qbar$. This is easily accomplished by verifying that $\bar{\xi} = \mathcal{Z}(\bar{\qbar})$, and $\bar{\qbar}$ satisfies \eqref{eqn:QConstraintPositive} and \eqref{eqn:QConstraintTotal}, completing the proof. 

\section{Proof of Theorem \ref{thm:x2_sec1}}
\label{apx:x2_sec1Proof}
We start with the same framework used to prove Theorem \ref{thm:x2_uvf}- that is, we search for a $\xi = \mathcal{Z}(\qbar)$ which optimizes $\lambda(n,\epsilon,\qbar)$. In this case, though, in addition to the nonnegativity constraint \eqref{eqn:x2_ZetaConstraintPositive} and the sum constraint \eqref{eqn:x2_ZetaConstraintTotal}, we also require $\xi$ to satisfy a radius constraint. The form of this constraint may be derived from the radius constraint \eqref{eqn:sec_radiusConstraint} applied to $\qbar$ and the relationship \eqref{eqn:x2_ZetaRelationDistance} of distances in the $\qbar$ domain to distances in the $\xi=\mathcal{Z}(\qbar)$ domain. We state this constraint in terms of squared distance as 
\begin{equation}
    \label{eqn:x2_ZetaConstraintRadius}
    |\xi-\bar{\xi}|^2 = \frac{2^{\kappa-2}(2^{\kappa-1}-1)}{(2^{\kappa}-2^{\kappa-1})(2^{\kappa}-1)} = \frac{(2^{\kappa-1}-1)}{2(2^{\kappa}-1)}\mathrm{.} 
\end{equation}

Now the $\dot{\xi}$ defined in \eqref{eqn:x2_zetaDotUniform} is no longer, in general, a valid movement direction, as it does not maintain compliance with the radius constraint \eqref{eqn:x2_ZetaConstraintRadius}. Instead, we consider a set of three unique elements of a given $\xi$: the smallest element, $\xi_a = \mathrm{min}_i(\xi_i)$; the second smallest element, $\xi_b = \mathrm{min}_i(\xi_i:i \neq a)$; and the largest element, $\xi_c = \mathrm{max}_i(\xi_i)$. We next define a movement direction $\dot{\xi}$ as  
\begin{equation}
    \label{eqn:x2_zetaDotSEC1}
    \dot{\xi}_i = \begin{cases}
        \xi_b - \xi_c &\text{if } i=a \\
        \xi_c - \xi_a &\text{if } i=b \\
        \xi_a - \xi_b &\text{if } i=c \\
        0 &\text{otherwise} \mathrm{.}
    \end{cases} 
\end{equation}
Movement in the direction defined by $\dot{\xi}$ maintains compliance of $\xi$ with \eqref{eqn:x2_ZetaConstraintTotal}, \eqref{eqn:x2_ZetaConstraintPositive}, and \eqref{eqn:x2_ZetaConstraintRadius}, provided $\xi_a>0$. (The case that $\xi_a=0$ would ordinarily be a separate edge case that needs to be checked for optimality, but we show below that this definition of $\dot{\xi}$ is sufficient to rule out all other cases, and the optimum does indeed require that $\xi_a=0$.) If a second derivative is required for $\xi$ with movement in the $\dot{\xi}$ direction, $\ddot{\xi}$ should be given by 
\begin{equation}
    \label{eqn:x2_zetaDotDotSEC1}
    \ddot{\xi}_i = \begin{cases}
        \xi_b + \xi_c - 2\xi_a &\text{if } i=a \\
        \xi_c + \xi_a - 2\xi_b &\text{if } i=b \\
        \xi_a + \xi_b - 2\xi_c &\text{if } i=c \\
        0 &\text{otherwise} \mathrm{.}
    \end{cases} 
\end{equation}
to ensure compliance with \eqref{eqn:x2_ZetaConstraintRadius}. The rate of change $\dot{\lambda}(n,\epsilon,\qbar)$ which corresponds to movement in the $\dot{\xi}$ direction is equal to 
\begin{equation}
    \label{eqn:x2_lambdaDotSEC1_1}
    \begin{split}
    \dot{\lambda}(n,\epsilon,\qbar) &= (2\!-\!\epsilon)^n 2^{-\kappa} \tau \!\left( \dot{\xi}_a e^{\tau \xi_a} + \dot{\xi}_b e^{\tau \xi_b} + \dot{\xi}_c e^{\tau \xi_c} \right) \\
    &= (2\!-\!\epsilon)^n 2^{-\kappa} \tau \!\left( (\xi_b\!-\!\xi_c)e^{\tau \xi_a} \!\!+\! (\xi_c\!-\!\xi_a)e^{\tau \xi_b} \!\!+\! (\xi_a\!-\!\xi_b)e^{\tau \xi_c} \right) \!\mathrm{.} 
    \end{split}
\end{equation}
Then dividing out always-positive terms and taking the series expansion of the exponential terms, we obtain 
\begin{equation}
    \label{eqn:x2_lambdaDotSEC1_2}
    \begin{split}
    \frac{\dot{\lambda}(n,\epsilon,\qbar)}{(2-\epsilon)^n 2^{-\kappa} \tau} &= \sum_{j=0}^{\infty}{\left( \frac{\tau^j}{j!}\left( (\xi_b-\xi_c)\xi_a^j + (\xi_c-\xi_a)\xi_b^j + (\xi_a-\xi_b)\xi_c^j \right) \right)} \\
    &= \sum_{j=0}^{\infty}{\left( \frac{\tau^j}{j!} (\xi_c-\xi_b)(\xi_a-\xi_c)(\xi_b-\xi_a)\Upsilon(\xi_a,\xi_b,\xi_c,j\!-\!1) \! \right)} \mathrm{,} 
    \end{split}
\end{equation}
where $\Upsilon(a,b,c,n)$ is the homogeneous polynomial or order $n$ in variables $a$, $b$, and $c$, with all coefficients equal to one (and with the convention that $\Upsilon(a,b,c,-1) = 0$). Considering \eqref{eqn:x2_lambdaDotSEC1_2}, clearly if $\xi_a < \xi_b < \xi_c$, there is exactly one negative term, so $\dot{\lambda}(n,\epsilon,\qbar)$ is always negative, and such a $\xi$ is not optimal. 

The only remaining cases are that $\xi_a=\xi_b$ or $\xi_b=\xi_c$. In either case, it is clear from \eqref{eqn:x2_lambdaDotSEC1_1} that $\dot{\lambda}(n,\epsilon,\qbar)$ is equal to zero. Then we calculate $\ddot{\lambda}(n,\epsilon,\qbar)$ using \eqref{eqn:x2_lambdaDotSEC1_1} and \eqref{eqn:x2_zetaDotDotSEC1} to obtain 
\begin{equation}
    \label{eqn:x2_lambdaDotDotSEC1_1}
    \begin{split}
    \ddot{\lambda}(n,\epsilon,\qbar) \!&=\! (2\!-\!\epsilon)^n 2^{-\kappa} \tau \!\left( \!\left(\tau (\xi_b\!-\!\xi_c)^2 \!+\! \xi_b\!+\!\xi_c\!-\!2\xi_a\right)\! e^{\tau \xi_a} \right. \!\!+\! \left(\tau (\xi_c\!-\!\xi_a)^2 \!+\! \xi_c \!+\! \xi_a \!-\! 2\xi_b \right) \!e^{\tau \xi_b} \!\!+\! \left. \left(\tau (\xi_a\!-\!\xi_b)^2 \!+\! \xi_a \!+\! \xi_b \!-\! 2\xi_c \right) \!e^{\tau \xi_c} \right) \!\mathrm{.}
    \end{split}
\end{equation}
Considering the case of $\xi_b = \xi_c$, \eqref{eqn:x2_lambdaDotDotSEC1_1} reduces to 
\begin{equation}
    \label{eqn:x2_lambdaDotDotSEC1_2}
    \begin{split}
    \ddot{\lambda}(n,\epsilon,\qbar) = (2\!-\!\epsilon)^n 2^{1-\kappa} \tau\! \left( \left(\xi_c\! -\! \xi_a \right)\! e^{\tau \xi_a}\!\! +\! \left(\tau (\xi_c\! -\! \xi_a)^2\!\! +\! \xi_a\! -\! \xi_c \right)\! e^{\tau \xi_c} \right)\! \mathrm{.}
    \end{split}
\end{equation}
With the substitution $\upsilon = \tau(\xi_c-\xi_a)$, this becomes 
\begin{equation}
    \label{eqn:x2_lambdaDotDotSEC1_3}
    \begin{split}
    &\ddot{\lambda}(n,\epsilon,\qbar) = (2\!-\!\epsilon)^n 2^{1-\kappa} \upsilon e^{\tau\xi_a} \! \left( 1 -e^{\upsilon} + \upsilon e^{\upsilon} \right) \mathrm{.}
    \end{split}
\end{equation}
All terms in \eqref{eqn:x2_lambdaDotDotSEC1_3} are positive, so $\ddot{\lambda}(n,\epsilon,\qbar)$ is also positive, and the case of $\xi_b=\xi_c$ is a candidate for a global minimum. 

Now considering the case that $\xi_a=\xi_b$, $\dot{\lambda}(n,\epsilon,\qbar)$ is still zero, and \eqref{eqn:x2_lambdaDotDotSEC1_1} reduces to 
\begin{equation}
    \label{eqn:x2_lambdaDotDotSEC1_4}
    \begin{split}
    \ddot{\lambda}(n,\epsilon,\qbar) &= (2\!-\!\epsilon)^n 2^{1-\kappa} \tau\! \left( \left(\xi_a\! -\! \xi_c \right)\! e^{\tau \xi_c}\!\! +\! \left(\tau (\xi_a\! -\! \xi_c)^2\!\! +\! \xi_c\! -\! \xi_a \right)\! e^{\tau \xi_a} \right)\! \mathrm{.}
    \end{split}
\end{equation}
With the substitution $\upsilon' = \tau(\xi_a-\xi_c)$ (which now implies that $\upsilon'$ is negative), this becomes 
\begin{equation}
    \label{eqn:x2_lambdaDotDotSEC1_5}
    \begin{split}
    &\ddot{\lambda}(n,\epsilon,\qbar) = (2\!-\!\epsilon)^n 2^{1-\kappa} \upsilon' e^{\tau\xi_c} \! \left( 1 -e^{\upsilon'} + \upsilon' e^{\upsilon'} \right) \mathrm{.}
    \end{split}
\end{equation}
Here, all terms are positive except the $\upsilon'$ term, which is negative. Thus, $\ddot{\lambda}(n,\epsilon,\qbar)$ is negative, and the case of $\xi_a=\xi_b$ is not a candidate for a global minimum. 

Because all other possibilities for optimal $\xi_a$, $\xi_b$, and $\xi_c$ have been excluded, the optimum must have $\xi_a<\xi_b=\xi_c$. But because $\xi_b$ is defined as the second-smallest element of $\xi$, this requires that all elements of $\xi$ except the minimum element must be equal. Because of constraints \eqref{eqn:x2_ZetaConstraintTotal} and \eqref{eqn:x2_ZetaConstraintRadius}, the only viable solutions must then have a minimum element equal to zero, and all other elements equal to one half. This solution also complies with the nonnegativity constraint \eqref{eqn:x2_ZetaConstraintPositive}, and its counterpart $\qbar = \mathcal{Z}^{-1}(\xi)$ defines a $\qbar$ with all elements in one $(\kappa-1)$-dimensional subspace equal to zero and all other elements set to $2^{1-\kappa}$. This $\qbar$ not only satisfies all applicable constraints (\eqref{eqn:QConstraintPositive}, \eqref{eqn:QConstraintTotal}, and \eqref{eqn:sec_radiusConstraint}), it is also a subspace exclusion code, equivalent to $\check{\qbar}^{\{\kappa-1\}}$, as required.  

\section{Proof of Theorem \ref{thm:x2_sec2}}
\label{apx:x2_sec2Proof}
We begin by examining the minimum distance constraint \eqref{eqn:x2_QConstraintSEC1} in conjunction with the $\xi$ radius constraint \eqref{eqn:sec_radiusConstraint}. We convert both of these from the $\qbar$ domain to the $\xi$ domain using \eqref{eqn:x2_ZetaRelationDistance} and express them in terms of squared distance to yield 
\begin{equation}
    \label{eqn:x2_ZetaConstraintRadiusSEC2}
    |\xi - \bar{\xi}|^2 = \frac{2^{\kappa-2}\left( 2^{u}-1\right)}{(2^{\kappa}-2^{u})(2^{\kappa}-1)}
\end{equation}
and 
\begin{equation}
    \label{eqn:x2_ZetaConstraintMinDistSEC2}
    \begin{split}
    \left|\xi - \check{\xi}^{[j]} \right|^2 \geq 2^{\kappa-2} \left| \check{\qbar}^{\{u\}} - \check{\qbar}^{\{\kappa-1\}} \right|^2 &= 2^{\kappa-2}\!\! \left(\! (2^{\kappa-1}\!\!\!-\!2^{u}) \! \left( \!\frac{1}{2^\kappa\!-\!2^u} \!\right)^2 \!\!\!+\! 2^{\kappa-1} \!\left( \!\frac{1}{2^{\kappa-1}} \!-\! \frac{1}{2^{\kappa}\!-\!2^{u}}\! \right)^2 \right) \\
    &=  \frac{(2^{\kappa-1}-2^{u})2^{\kappa-1} + \left( 2^{\kappa-1}-2^{u} \right)^2}{2(2^{\kappa}-2^{u})^2} \\
    &= \frac{2^{\kappa-2}-2^{u-1}}{2^{\kappa}-2^{u}} \mathrm{,}
    \end{split}
\end{equation}
where $\check{\xi}^{[j]}$ is, as the notation suggests, the representation in the $\xi$ domain of the subspace exclusion code formed by excluding the $j^{\mathrm{th}}$ $(\kappa-1)$-dimensional subspace. It is thus equal to $\mathcal{Z}(\check{\qbar}^{[\Xi(W,\kappa-1)_j]})$, and its elements are given by 
\begin{equation}
    \label{eqn:x2_ZetaCheckDef}
    \check{\xi}^{[j]}_i = \begin{cases}
        0 & \text{if } i=j \\
        \frac{1}{2} &\text{otherwise.}
    \end{cases}
\end{equation}

Because $\bar{\xi}$ is constant, \eqref{eqn:x2_ZetaConstraintRadiusSEC2} may be converted to a constraint on the magnitude of $\xi$, giving 
\begin{equation}
    \label{eqn:x2_zetaMagnitudeSEC2}
    \begin{split}
    |\xi|^2 &= \frac{2^{\kappa-2}\left( 2^{u}-1\right)}{(2^{\kappa}-2^{u})(2^{\kappa}-1)} + 2\bar{\xi}^{\intercal} - |\bar{\xi}|^2 \\
    &= \frac{2^{\kappa-2}\left( 2^{u}-1\right)}{(2^{\kappa}-2^{u})(2^{\kappa}-1)} + \frac{(2^{\kappa-1}-1)^2}{2^{\kappa}-1} \mathrm{.} 
    \end{split}
\end{equation}

Then considering \eqref{eqn:x2_ZetaConstraintMinDistSEC2}, the expression left of the inequality may be expanded using \eqref{eqn:x2_ZetaCheckDef} to yield
\begin{equation}
    \label{eqn:x2_ZetaMinDistEquivalent1}
    \begin{split}
    \left|\xi - \check{\xi}^{[j]} \right|^2 &= \sum_{i=1}^{2^{\kappa}-1}{\left(\xi_i^2 -2\xi_1 \check{\xi}^{[j]}_i + (\check{\xi}^{[j]}_i)^2  \right)}\\
    &= |\xi|^2 +|\check{\xi}^{[j]}|^2 -2\left( \sum_{i \in \{[\![ 0,2^\kappa-1]\!] \setminus j\}}{\frac{1}{2}\xi_i} \right) \\
    &= |\xi|^2 + \frac{2^{\kappa-1}-1}{2} - \left( \left(\sum_{i=1}^{2^{\kappa}-1}{\xi_i}\right) - \xi_j \right) \\
    &= \xi_j \!+\! \frac{2^{\kappa-2}\!\left( 2^{u}\!-\!1\right)}{(2^{\kappa}\!-\!2^{u})(2^{\kappa}\!-\!1)} \!+\! \frac{(2^{\kappa-1}\!\!-\!1)^2}{2^{\kappa}\!-\!1} \!+\! \frac{2^{\kappa-1}\!\!-\!1}{2} \!-\! (2^{\kappa-1}\!\!-\!1) \\
    &= \xi_j - \frac{2^{\kappa-2}-2^{u-1}}{2^{\kappa}\!-\!2^{u}} \mathrm{.} 
    \end{split}
\end{equation}
Substituting this expression back into \eqref{eqn:x2_ZetaConstraintMinDistSEC2} gives simply 
\begin{equation}
    \label{eqn:x2_ZetaConstraintMinDistSimplified}
    \xi_j \geq \frac{2^{\kappa-1}-2^{u}}{2^{\kappa}\!-\!2^{u}} \; \forall \; j \in [\![ 0,2^\kappa-1]\!] \mathrm{.} 
\end{equation}
This clearly supersedes the nonnegativity constraint \eqref{eqn:x2_ZetaConstraintPositive}, so we now require the optimal $\xi$ that satisfies only the element-wise sum constraint \eqref{eqn:x2_ZetaConstraintTotal}, the radius constraint \eqref{eqn:x2_ZetaConstraintRadiusSEC2}, and the element minimum constraint \eqref{eqn:x2_ZetaConstraintMinDistSimplified}. 

We may now use the same movement direction $\dot{\xi}$ defined in \eqref{eqn:x2_zetaDotSEC1} and its associated result in the proof of Theorem \ref{thm:x2_sec1}- namely: if there are three distinct elements $\xi_a$, $\xi_b$, and $\xi_c$ of $\xi$ such that $\xi_a \leq \xi_b < \xi_c$, and $\xi_a$ is not at a minimum boundary, then $\xi$ is not optimal. We choose $\xi_a$ to be the minimum element which is greater than the minimum bound of \eqref{eqn:x2_ZetaConstraintMinDistSimplified}, $\xi_b$ to be the next largest element after $\xi_a$, and $\xi_c$ to be the maximum element. Then the optimal solution must have $\xi_b=\xi_c$. That is, there may be at most one element of $\xi$ which is not equal to either $\xi_{\mathrm{min}}=(2^{\kappa-1}-2^{u})/(2^{\kappa}-2^{u})$ or to the maximum element value $\xi_c$. If we now use $N$ to represent the number of elements of $\xi$ equal to $\xi_{\mathrm{min}}$, we may use the remaining constraints \eqref{eqn:x2_ZetaConstraintTotal} and \eqref{eqn:x2_ZetaConstraintRadius} to calculate the required $\xi_a$ and $\xi_c$. We do this by splitting the vector $\xi$ into two smaller vecors: $\xi_{\text{-}}$, composed of $N$ elements equal to $\xi_{\mathrm{min}}$ and $\xi_{\text{+}}$, composed of $\xi_a$ and $(2^{\kappa}-2-N)$ elements equal to $\xi_c$. We can now express the means ($\bar{\xi_{\text{-}}}$ and $\bar{\xi_{\text{+}}})$ and squared distances from the mean ($\varrho_{\text{-}}^2$ and $\varrho_{\text{+}}^2)$ for these vectors in terms of $N$ and $\xi_{\Delta}=\xi_c-\xi_a$ as 
\begin{equation}
    \label{eqn:x2_SEC2N1}
    \bar{\xi_{\text{\text{-}}}} = \xi_{\mathrm{min}}=\frac{2^{\kappa-1}-2^{u}}{2^{\kappa}-2^{u}} \mathrm{,}
\end{equation}
\begin{equation}
    \label{eqn:x2_SEC2N2}
    \begin{split}
    &\bar{\xi_{\text{\text{+}}}} = \frac{2^{\kappa-1}-1 - N\xi_{\mathrm{min}}}{2^{\kappa}-N-1} = \frac{2^{\kappa-1}-1 - N\frac{2^{\kappa-1}-2^{u}}{2^{\kappa}-2^{u}}}{2^{\kappa}-N-1} \mathrm{,}
    \end{split}
\end{equation}
\begin{equation}
    \label{eqn:x2_SEC2N3}
    \varrho^2_{\text{-}} = 0 \mathrm{,}
\end{equation}
and
\begin{equation}
    \label{eqn:x2_SEC2N4}
    \varrho^2_{\text{+}} = \frac{2^{\kappa}-N-2}{2^{\kappa}-N-1} \xi_{\Delta}^2 \mathrm{.} 
\end{equation}
These quantities may then be related to the squared distance $\varrho$ of $\xi$ from its mean $\bar{\xi}$ as 
\begin{equation}
    \label{eqn:x2_SECN5}
    \varrho^2 = \varrho^2_{\text{-}} + \varrho^2_{\text{+}} + N(\bar{\xi_{\text{-}}}-\bar{\xi})^2 + (2^{\kappa}-1-N)(\bar{\xi_{\text{+}}}-\bar{\xi})^2 \mathrm{.}
\end{equation}
Substituting \eqref{eqn:x2_ZetaConstraintRadiusSEC2}, \eqref{eqn:x2_SEC2N2}, \eqref{eqn:x2_SEC2N3}, and \eqref{eqn:x2_SEC2N4} into \eqref{eqn:x2_SECN5} gives 
\begin{equation}
    \label{eqn:x2_SECN6}
    \begin{split}
    \frac{2^{\kappa-2}\left( 2^{u}-1\right)}{(2^{\kappa}-2^{u})(2^{\kappa}-1)} &= \frac{2^{\kappa}-N-2}{2^{\kappa}-N-1} \xi_{\Delta}^2 + N\left(\xi_{\mathrm{min}} - \bar{\xi}\right)^2 + (2^{\kappa}-N-1) \left(\frac{2^{\kappa-1}-1-N\xi_{\mathrm{min}}}{2^{\kappa}-N-1} - \bar{\xi}\right)^{2} \\
    &= \frac{2^{\kappa}-N-2}{2^{\kappa}-N-1} \xi_{\Delta}^2 + N\xi_{\mathrm{min}}^2 + \frac{\left(2^{\kappa-1}-1-N\xi_{\mathrm{min}}\right)^2}{2^{\kappa}-N-1} - (2^{\kappa}-2)\bar{\xi} + (2^{\kappa}-1)\bar{\xi}^2 \\
    &= \frac{2^{\kappa}-N-2}{2^{\kappa}-N-1} \xi_{\Delta}^2 + \frac{(2^{\kappa}-1)N}{2^{\kappa}-N-1}\xi_{\mathrm{min}}^2 - \frac{(2^{\kappa}-2)N}{2^{\kappa}-N-1}\xi_{\mathrm{min}} + \frac{(2^{\kappa-1}-1)^2}{2^{\kappa}-N-1} - (2^{\kappa}-2)\bar{\xi} +(2^{\kappa}-1)\bar{\xi}^2 \mathrm{.}
    \end{split}
\end{equation}
Multiplying by $(2^{\kappa}-N-1)$ and substituting in \eqref{x2_zetaUniformFracDef} then gives
\begin{equation}
    \label{eqn:x2_SECN7}
    \begin{split}
    \frac{(2^{\kappa}\!-\!N\!-\!1)2^{\kappa-2}\left( 2^{u}\!-\!1\right)}{(2^{\kappa}\!-\!2^{u})(2^{\kappa}\!-\!1)} &= (2^{\kappa}\!-\!N\!-\!2) \xi_{\Delta}^2 \!+\!  N(2^{\kappa}\!-\!1)\xi_{\mathrm{min}}^2 \!-\! N(2^{\kappa}\!-\!2)\xi_{\mathrm{min}} \!+\! N\frac{(2^{\kappa}\!-\!2)(2^{\kappa-1}\!-\!1)}{(2^{\kappa}\!-\!1)} \!-\! N\frac{(2^{\kappa-1}\!-\!1)^2}{(2^{\kappa}\!-\!1)} \\
    &= (2^{\kappa}\!-\!N\!-\!2) \xi_{\Delta}^2 \!+\!  \frac{N \left((2^{\kappa}\!-\!1)\xi_{\mathrm{min}} \!-\! (2^{\kappa-1}\!-\!1)\right)^2}{(2^{\kappa}\!-\!1)}\mathrm{.}
    \end{split}
\end{equation}
Multiplying by $(2^{\kappa}-2^{u})^2(2^{\kappa}-1)$ and substituting in \eqref{eqn:x2_SEC2N1}, we now have 
\begin{equation}
    \label{eqn:x2_SECN8}
    \begin{split}
    (2^{\kappa}\!-\!N\!-\!1)2^{\kappa-2}\left( 2^{u}\!-\!1\right)(2^{\kappa}\!-\!2^{u}) &= (2^{\kappa}\!-\!2^{u})^2(2^{\kappa}\!-\!1)(2^{\kappa}\!-\!N\!-\!2) \xi_{\Delta}^2 + N \left((2^{\kappa}\!-\!1)(2^{\kappa-1}\!-\!2^{u}) \!-\! (2^{\kappa-1}\!-\!1)(2^{\kappa}\!-\!2^{u})\right)^2 \\
    &= (2^{\kappa}\!-\!2^{u})^2(2^{\kappa}\!-\!1)(2^{\kappa}\!-\!N\!-\!2) \xi_{\Delta}^2 + N \left(2^{\kappa-1}(1\!-\!2^{u})\right)^2 \mathrm{.}
    \end{split}
\end{equation}
We now solve for $\xi_{\Delta}^2$ to obtain 
\begin{equation}
    \label{eqn:x2_SECN9}
    \begin{split}
    \xi_{\Delta}^2 \! &= \frac{(2^{\kappa}\!-\!N\!-\!1)2^{\kappa-2}\left( 2^{u}\!-\!1\right)(2^{\kappa}\!-\!2^{u}) \!-\! N \!\left(2^{\kappa-1}(1 \!-\! 2^{u})\right)^2}{(2^{\kappa}\!-\! 2^{u})^2(2^{\kappa}\!-\! 1)(2^{\kappa}\!-\!N\!-\!2)} \\
    &= \frac{(2^{u}\!-\!1)2^{\kappa-2} \left( 2^{\kappa}\!-\! 2^{u}\!-\!N 2^{u} \right)}{(2^{\kappa}\!-\! 2^{u})^2(2^{\kappa}\!-\!N\!-\!2)} \mathrm{.}
    \end{split}
\end{equation}
We now observe several properties of the expression for $\xi_{\Delta}^2$ in \eqref{eqn:x2_SECN9}. First, $\xi_{\Delta}^2$ is monotonically decreasing in $N$ (with a pole at $N=2^{\kappa}-2$, which is not a concern because $N \geq 2^{\kappa}-2$ does not result in a $\xi$ that meets the required $\xi_b=\xi_c$ condition). Second, at $N=2^{\kappa-u}-1$, $\xi_{\Delta}$ is equal to zero, indicating that $\xi_a = \xi_c$, and the corresponding $\xi$ is given by 
\begin{equation}
    \label{eqn:x2_SEC2_FinalXi}
    \xi_i = \begin{cases}
        \xi_{\mathrm{min}} &\text{ if } i < 2^{u} \\
        \frac{1}{2} &\text{ otherwise. }
    \end{cases}
\end{equation}
Third, at $N=2^{\kappa-u}-2$, $\xi_{\Delta}$ is equal to $(1/2-\xi_{\mathrm{min}})$, indicating that $\xi_a = \xi_{\mathrm{min}}$, and $\xi$ in this case is identical to that \eqref{eqn:x2_SEC2_FinalXi} of the $N=2^{\kappa-u}-1$ case. Then because $\xi_{\Delta}^2$ is monotonically decreasing in $N$, any $N>2^{\kappa-u}-1$ results in negative $\xi_{\Delta}^2$, implying that no real solution exists, and any $N<2^{\kappa-u}-2$ will result in $\xi_a<\xi_{\mathrm{min}}$, violating the minimum element constraint. Thus, the only possible values for $N$ result in the same $\xi$, and this $\xi$ (defined in \eqref{eqn:x2_SEC2_FinalXi}) must be optimal. As a final step, we seek a valid $\qbar$ which produces the $\xi=\mathcal{Z}(\qbar)$ described above. This $\qbar$ is simply $\check{\qbar}^{\{u\}}$ as required, completing the proof. 

\bibliographystyle{IEEEtran}
\bibliography{./SDCC}

\end{document}